\documentclass[
  twocolumn,          
  aps,              
  prx,              
  superscriptaddress,
  nofootinbib
]{revtex4-2}

\usepackage[utf8]{inputenc}
\usepackage{amsmath,amssymb,bm}
\usepackage{physics}
\usepackage{bbold}
\usepackage{color}
\usepackage{hyperref}
\usepackage{circuitikz}
\usepackage{amsthm}

\newtheorem{proposition}{Proposition}
\newtheorem{theorem}{Theorem}
\newtheorem{lemma}{Lemma}
\newtheorem{corollary}{Corollary}
\theoremstyle{remark}
\newtheorem{remark}{Remark}

\begin{document}

\title{Fingerprints of classical memory in quantum hysteresis}

\author{Francesco Caravelli}
 \affiliation{Los Alamos National Laboratory, Los Alamos, New Mexico 87545, USA}
\affiliation{*Dipartimento di Fisica dell’Universita di Pisa, Largo Bruno Pontecorvo 3, I-56127 Pisa, Italy}

\date{\today}

\begin{abstract}
We present a simple framework for classical and quantum ``memory'' in which the
Hamiltonian at time $t$ depends on past values of a control Hamiltonian through a
causal kernel. This structure naturally describes finite-bandwidth or filtered
control channels and provides a clean way to distinguish between memory in the
control and genuine non-Markovian dynamics of the state. We focus on models where
$H(t)=H_0+\int_{-\infty}^{t}K(t-s)\,H_1(s)\,ds$, and illustrate the framework on
single-qubit examples such as $H(t)=\sigma_z+\Phi(t)\sigma_x$ with
$\Phi(t)=\int_{-\infty}^{t}K(t-s)\,u(s)\,ds$. We derive basic properties of such dynamics,
discuss conditions for unitarity, give an equivalent time-local description for exponential
kernels, and show explicitly how hysteresis arises in the response of a driven qubit.

\end{abstract}

\maketitle


\section{Introduction}
\footnote[1]{caravelli@lanl.gov    
* guest affiliations. }
High-fidelity coherent control is a basic requirement across quantum information
science, from gate synthesis and dynamical decoupling to calibration,
characterization, and error mitigation.\cite{NielsenChuang,Glaser2015,ViolaKnillLloyd1999,Temme2017}
In the idealized description used in much of the theory, the experimenter
prescribes a time-dependent Hamiltonian and the device follows it exactly.
This abstraction is analytically convenient and often accurate at the level of
first principles, but it hides an important practical layer: the Hamiltonian
experienced by the quantum system is not programmed directly. Instead, it is
generated and delivered by a classical control stack, as for instance room-temperature
electronics, wiring, attenuators, bias tees, packaging, and on-chip
filters, with finite bandwidth and its own internal dynamics.\cite{Krantz2019} 
This issue is particularly  timely in view of recent evidence of memory effects in out-of-equilibrium protocols on quantum
annealers by Pelofske and collaborators \cite{pelofske,barrows}.

A simple intuition comes from elementary electronics. Driving an RC circuit with
a time-dependent voltage does not produce an output that tracks the input
instantaneously; rather, the output relaxes toward it with a characteristic time
constant, leading to smoothing and a phase lag that depend on the drive history.
Analogous effects appear broadly in quantum-control platforms: control lines act
as low-pass filters, cables and packaging introduce dispersion, bias tees and
amplifiers contribute dynamical response, and resonator-mediated actuation adds
latency and ringing.\cite{Krantz2019,Motzoi2009,BornemanCory2012BandwidthLimited}
As a result, the field that actually reaches the device can carry substantial
\emph{classical memory}, even when the device itself is perfectly isolated and
its intrinsic evolution is unitary.

This distinction matters both conceptually and operationally. Classical memory
in the control channel can generate signatures that resemble dissipation or
non-Markovianity, for example, hysteresis loops when plotting an observable
against the commanded waveform, despite the absence of any environmental memory
in the quantum state evolution.\cite{BreuerLainePiilo2009,RivasHuelgaPlenio2014,deVegaAlonso2017}
Conversely, genuine quantum memory originates in system--environment coupling
and must be diagnosed by observables that are insensitive to purely classical
distortions.\cite{BreuerLainePiilo2009,RivasHuelgaPlenio2014,deVegaAlonso2017}
Separating these mechanisms is therefore essential for interpreting data and for
designing robust control protocols.

Motivated and inspired by these experimental works, we adopt a minimal framework in which
the commanded waveform $u(t)$ generated at room temperature is mapped to an
in-situ realized field $\Phi(t)$ by a causal filter,
\begin{eqnarray}
    \Phi(t)&&=\Phi_0+(K*u)(t)\nonumber \\
    &&=\Phi_0+\int_{-\infty}^{t} K(t-s)\,u(s)\,ds,
    \label{eq:intro_convolution}
\end{eqnarray}
and $\Phi(t)$ is the control amplitude that enters the device Hamiltonian. This
kernel inside the Hamiltonian has been used sporadically in the past, and it is used in the
quantum-control literature to model hardware distortion, finite bandwidth, and
ringdown.\cite{PhysRevApplied.4.024012,BornemanCory2012BandwidthLimited,Wittler2021C3Toolset}
In particular, Hincks \emph{et al.} incorporate explicit distortion dynamics
(including nonlinearities) directly into optimal-control synthesis, so that the
optimized command $u(t)$ implements a target unitary under the \emph{realized}
waveform that reaches the device.\cite{PhysRevApplied.4.024012} Related approaches
appear in resonator-limited control \cite{BornemanCory2012BandwidthLimited} and in
integrated calibration/optimization toolchains for superconducting circuits.\cite{Wittler2021C3Toolset}

While we adopt the same physical separation between commanded and realized
control, our goal is not pulse design under a known distortion model. Instead,
we use the filtered-control viewpoint to develop a \emph{memory diagnostics}
program based on geometric loop measures. Concretely, we  treat the control as a dynamical memory system in its own right and
quantify its hysteretic response by loop areas such as
$\mathcal A_{u\Phi}=\oint \Phi\,du$;
as a result, we show that even when the device responds adiabatically to the \emph{realized}
trajectory (so that $O(t)\approx f(\Phi(t))$ and $\mathcal A_{\Phi O}\approx 0$),
a nonzero commanded-loop area $\mathcal A_{uO}=\oint O\,du$ can persist as a
deterministic fingerprint of classical filtering, with controlled small-amplitude
relations linking $\mathcal A_{uO}$ to $\mathcal A_{u\Phi}$;
we then use $\mathcal A_{\Phi O}$ as a diagnostic that separates classical
control memory from genuinely state-history-dependent effects (e.g.\ coherent
nonadiabaticity) at fixed realized field. In this sense, our framework is aimed
at the interpretation and separation of mechanisms, rather than at the inversion of the
distortion map.

In the second part of the manuscript, and to keep the control memory both experimentally grounded and analytically
tractable, we connect the kernel $K$ to (classical) passive RC line models to build the intuition.  Classical
network-synthesis constraints imply that passive RC ladders have transfer
functions with poles on the negative real axis, so their impulse responses are
sums (or mixtures) of decaying exponentials. This motivates the exponential mode
embedding used throughout the paper and provides a concrete physical meaning for
the rates as relaxation time scales of the control. 
Dissipation in the control stack is most commonly modeled by treating the channel as an \emph{open} quantum system, i.e., by coupling it to uncontrolled degrees of freedom and describing the reduced channel dynamics with an effective Gorini-Kossakowski-Sudarshan-Lindblad, or GKLS  generator (see for instance \cite{Cattaneo2021} for a review). From a complementary microscopic viewpoint, however, eliminating the environmental degrees of freedom leads to a generalized Langevin description \cite{GardinerCollett1985} in which \emph{memory and fluctuations arise together}: the same bath modes that generate a retarded friction (memory) kernel also produce stochastic forces, whose correlations are constrained by the fluctuation--dissipation theorem \cite{Kubo1957Irreversible1,Forster1975}. In this work we focus on the deterministic memory kernel governing the mean input--output response of the channel; the accompanying noise, while essential for a fully consistent description, is not modeled explicitly here but is nevertheless linked to the kernel through the corresponding fluctuation--dissipation relations.

We also give a microscopic derivation (in the Kubo linear response formalism) in which the control
is modeled as a (possibly large) quantum channel driven at an input port by the
classical source $u(t)$. In linear response, the realized field is a retarded
susceptibility convolved with the command, and we show that under weak coupling conditions the device
evolves under an effective filtered-control Hamiltonian, with fluctuations
entering only as a separate weak correction. We highlight that the classical and quantum framework are mutually consistent in the presence of dissipation.

The manuscript is organized as follows. Section~II introduces the kernel-filtered control model and the standing
assumptions ensuring unitary device evolution at the level treated here.
Section~III shows how common causal kernels admit a time-local embedding via a
finite set of auxiliary filter modes. Section~IV defines loop measures in the
$(u,\Phi)$, $(u,O)$, and $(\Phi,O)$ planes and introduces the adiabatic
response function $f(\Phi)$. Section~V treats the nonadiabatic regime and
clarifies why loop areas alone do not certify open-system quantum memory.
Sections~VI--VII connect the kernel description to RC-network control lines and
work out single-qubit case studies (including an exactly solvable commuting
benchmark). Section~VIII presents numerical experiments mapping the frequency
dependence of the loop measures, and Sec.~IX concludes. Most technical
derivations are collected in the appendices.

\section{Hamiltonians with a Memory Kernel}
\label{sec:model}
In this first part of the manuscript, we discuss a simple model of the control. We provide a general treatment of filtered Hamiltonians, discussed in \cite{BornemanCory2012BandwidthLimited,PhysRevApplied.4.024012}.
\subsection{Model and standing assumptions}
\label{subsec:model_definition}

We consider closed-system dynamics generated by Hamiltonians whose control
component is processed through a causal response function. Concretely, we
consider
\begin{eqnarray}
    \hat H(t)&=&\hat H_A+\hat H^\prime_c(t)\\
    &\equiv&\hat H_A+\int_{-\infty}^{t}K(t-s)\,\hat H_c(s)\,ds,
    \label{eq:Hkernel_general_rewrite}
\end{eqnarray}
where $\hat H_A$ is a fixed (time-independent) drift Hamiltonian and $\hat H^\prime_c(s)$
is a Hermitian control Hamiltonian. The scalar function $K$ is a kernel
encoding the impulse response of a classical control channel.\cite{OppenheimWillsky1997,AgarwalLang2005}
Unless stated
otherwise, we assume that $K$ is \emph{causal}, e.g.
\begin{equation}
    K(\tau)=0\qquad \text{for }\tau<0,
\end{equation}
so that $\hat H(t)$ depends only on past values of the control.

The convolution in Eq.~\eqref{eq:Hkernel_general_rewrite} is understood as an
operator-valued integral (intended as a Bochner integral \cite{Mikusiski1978}). A simple set of sufficient conditions for
well-posedness is that $K\in L^1(\mathbb R_+)$ and that
$\sup_{s\in\mathbb R}\|\hat H_c(s)\|<\infty$, which guarantees that the integral
exists and defines a bounded operator for each $t$.\cite{Kato1995,ReedSimon1972}
These assumptions are
mild and encompass the kernels typically used to model finite bandwidth or
linear filtering (e.g.\ sums of decaying exponentials).\cite{OppenheimWillsky1997}

It is convenient to introduce the filtered control operator
\begin{equation}
   \hat H^\prime_c(t):=\int_{-\infty}^{t}K(t-s)\,\hat H_c(s)\,ds,
    \label{eq:H1_tilde_def}
\end{equation}
so that $\hat H(t)=\hat H_A+ \hat H^\prime_c(t)$. In this form, all history dependence
is confined to the map $\hat H_c(\cdot)\mapsto \hat H_c^\prime(\cdot)$, which is
linear and time-translation invariant. In other words, the control channel
acts as a classical LTI filter on an operator-valued input.\cite{OppenheimWillsky1997,BoydChua1985}
A direct generalization of this integral is the case of $K$ promoted to an operator is not proportional to the identity. This phenomenological Hamiltonian can be justified both via a classical and quantum treatment, that we discuss later in Sec. \ref{sec:physics}.

Now, we note that causality implies that specifying the control trajectory
$\hat H_c(s)$ for $s\le t$ uniquely determines the instantaneous generator $\hat H(t)$.
The resulting dynamics is nevertheless local in the state: once $\hat H(t)$ is
fixed, the Schr\"odinger equation retains its standard first-order form and
does not involve the past history of the quantum state.\cite{ReedSimon1972}
In this sense, the
``memory'' described by $K$ is classical and resides in the actuation
mechanism, not in the quantum dynamics of the system itself.

The RC circuit provides a useful archetype. A voltage command applied at the
input of an RC filter produces an output that relaxes toward the command
with a characteristic time constant.\cite{AgarwalLang2005}
At the Hamiltonian level, this
corresponds to replacing the commanded control by a smoothed and phase-lagged
realized control, obtained by convolution with an exponentially decaying
kernel.\cite{OppenheimWillsky1997}
More general kernels represent more elaborate linear responses,
including multi-timescale relaxation (sums of exponentials) and resonant
features (damped oscillatory responses), while still maintaining a simple,
causal description at the Hamiltonian level.\cite{OppenheimWillsky1997,Magos1970}

For the evolution generated by Eq.~\eqref{eq:Hkernel_general_rewrite} to be
unitary, the instantaneous Hamiltonian must be Hermitian for all $t$.
Since $\hat H_A$ and $\hat H_c(s)$ are assumed Hermitian, one has
\begin{equation}
    \hat H(t)^\dagger
    =
    \hat H_A
    +\int_{-\infty}^{t}K(t-s)^{*}\,\hat H_c(s)\,ds.
\end{equation}
Thus, for a scalar kernel acting multiplicatively on a Hermitian operator
input, a transparent sufficient condition for Hermiticity is that the kernel
be real-valued,
\begin{equation}
    K(\tau)\in\mathbb R \qquad \text{for all }\tau,
    \label{eq:K_real_condition}
\end{equation}
so that the filtered control $\hat H_c^\prime(t)$ is a real linear
combination of Hermitian operators. Under Eq.~\eqref{eq:K_real_condition}, the
propagator is unitary for each fixed control trajectory.\cite{ReedSimon1972,Kato1995}

We emphasize that Eq.~\eqref{eq:K_real_condition} is a condition on the
effective description rather than a fundamental constraint: in applications,
$K$ is typically identified with a classical impulse response (hence real),
or with the real part of a causal response function inferred from
calibration data.\cite{OppenheimWillsky1997,Krantz2019,Motzoi2009}
In either case, the framework provides a compact way to
incorporate finite-bandwidth distortions while maintaining a strictly
unitary closed-system description.\cite{NielsenChuang,ReedSimon1972}

The discussion simplifies dramatically for the case of a scalar control
$\hat H_c(t)=u(t)\hat M$ filtered through a causal kernel:
the realized drive is $w(t)=(K*u)(t)$ and the Hamiltonian reads
$\hat H(t)=\hat H_A+w(t)\hat M$. 
Above, $u(t)$ is the \textit{commanded} control parameter, e.g. what an experimenter would set on the software or the voltage generator. In App.\ref{app:dyson_kernel_time_independent_H0} we provide general statements for the unitary evolution operator based on the properties of the kernel $K(\tau)$. 

Before we continue, let us now briefly comment on the units.
With $\hbar=1$, Hamiltonians have units of inverse time. Since the filtered
control enters as a causal convolution
$
\hat H_c'(t)=\int_{-\infty}^t K(t-s)\,\hat H_c(s)\,ds,
$
dimensional consistency requires $[K]\,[\hat H_c]\,[\mathrm{time}]=[\hat H]$.
In the convention we use here where the control term is written as
$\hat H_c(s)=u(s)\,\hat M$ with a fixed dimensionless generator $\hat M$
(e.g.\ Pauli operators), the drive amplitude $u$ carries the physical energy
scale, and one has $[K]=1/\mathrm{time}$ so that
$u_{\mathrm{eff}}(t)=\int_{-\infty}^t K(t-s)u(s)\,ds$ has the same units as $u$.
In this convention the integrated weight
$g=\int_0^\infty K(\tau)\,d\tau$ is dimensionless, and the Markov/instantaneous
limit $K(\tau)\to g\,\delta(\tau)$ corresponds to a memoryless rescaling
$u_{\mathrm{eff}}(t)\to g\,u(t)$.

Later, when we interpret $u(t)$ as a laboratory control such as a voltage
waveform $V(t)$ (rather than an already-calibrated energy amplitude), an
additional conversion factor is implicitly present: the device couples to an
electrical potential through an appropriate charge-like scale (e.g.\ $e$ or
an effective charge/lever arm set by circuit geometry), so the Hamiltonian drive
takes the schematic form $\hat H_c(t)\sim (\text{coupling})\times V(t)\times \hat M$.
One may adopt ``electrical natural units'' in which this charge scale is set to
unity (informally $e=1$), absorbing the conversion into the definition of $u$.
Restoring physical units then amounts to re-inserting the appropriate factor
(for example $e$ times a dimensionless lever arm, or more generally the relevant
calibration constant mapping volts to energy) so that the product has units of
energy. Which prefactor appears depends on the experimental realization (gate
voltage, flux bias, piezo drive, etc.), but once this mapping $V\mapsto u$ is
fixed, the kernel units follow as above and the convolution remains dimensionally
consistent.

\subsection{Time-local embedding for scalar filtered control}
\label{sec:embedding_scalar}

A key advantage of kernel-filtered Hamiltonians is that broad and physically
relevant classes of causal kernels admit an equivalent \emph{time-local}
realization once one introduces a small set of auxiliary classical variables.
This converts the nonlocal convolution in the Hamiltonian into an ordinary
differential equation (ODE) for the realized field, coupled to the usual
Schr\"odinger equation.\cite{OppenheimWillsky1997,Kailath1980}
The quantum evolution remains unitary (for Hermitian
instantaneous Hamiltonians), while the control channel becomes an explicit
finite-dimensional dynamical system.

We focus on the case in which the filtered control enters as a scalar
amplitude multiplying a fixed Hermitian operator,
\begin{equation}
    \hat H(t)=\hat H_A+\Phi(t)\,\hat M,
    \qquad
    \Phi(t)=\int_{-\infty}^{t}K(t-s)\,u(s)\,ds,
    \label{eq:Phi_conv_scalar}
\end{equation}
with real command $u(t)$ and causal kernel $K(\tau)$. In this setting, $\Phi(t)$ is the \textit{realized field}, e.g. the effective drive amplitude reaching the device. While this choice will be motivated shortly, it is evident from its form that it is a form of classical memory in the Hamiltonian control.

A particularly convenient class is given by kernels representable as finite
(or truncated) sums of decaying exponentials,
\begin{equation}
    K(\tau)=\sum_{k=1}^{K_{\max}} c_k e^{-\nu_k \tau}\,\Theta(\tau),
    \qquad
    \nu_k>0,\ \ c_k\in\mathbb{R}.
    \label{eq:K_sum_exp_main}
\end{equation}
As we show in the example later, this is motivated by the physics of a control line,
where passive RC ladders and lossy transmission channels naturally generate
multi-exponential impulse responses.\cite{Magos1970,FialkowGerst1951}
This includes, as special cases, the single-mode RC kernel, where the decay
rates $\nu_k$ form a ladder $\nu_k=2\pi k/\beta$.

Defining auxiliary filter modes
\begin{equation}
    \Phi_k(t):=\int_{-\infty}^{t} c_k e^{-\nu_k(t-s)}u(s)\,ds,
    \qquad
    \Phi(t)=\sum_{k=1}^{K_{\max}}\Phi_k(t),
    \label{eq:Phi_modes_def_main}
\end{equation}
one finds that each mode satisfies a first-order ODE,
\begin{equation}
    \dot\Phi_k(t)=-\nu_k\Phi_k(t)+c_k u(t),
    \qquad k=1,\ldots,K_{\max}.
    \label{eq:Phi_mode_ode_main}
\end{equation}
Consequently, the kernel-filtered Schr\"odinger dynamics is equivalent to the
time-local coupled system
\begin{equation}
\begin{cases}
    i\,\dfrac{d}{dt}\ket{\psi(t)}
    =\Big(\hat H_A+\Big[\sum_{k=1}^{K_{\max}}\Phi_k(t)\Big]\hat M\Big)\ket{\psi(t)},\\[6pt]
    \dot\Phi_k(t)=-\nu_k\Phi_k(t)+c_k u(t),
    \qquad k=1,\ldots,K_{\max}.
\end{cases}
\label{eq:schrodinger_filter_coupled_main}
\end{equation}
The structure of these equations are similar, in spirit, to those of circuits with memory \cite{memristors,caravelli,diventra}.

In this representation the memory kernel is replaced by a finite number of
classical state variables $\{\Phi_k\}$ that store the control history.\cite{Kailath1980,OppenheimWillsky1997}
We will argue below why this representation is general enough in the case of a
dissipative (and classically treated) \textit{control} channel.\cite{AgarwalLang2005,Magos1970}

The model of eqns \eqref{eq:schrodinger_filter_coupled_main} also makes precise
the case in which there is the absence of memory. In the strictly memoryless case, one expects
$\Phi(t)$ to track the command $u(t)$ instantaneously, $\Phi(t)=g\,u(t)$ for some
(static) gain $g$. At the kernel level, this corresponds to
\begin{equation}
    K(\tau)\longrightarrow g\,\delta(\tau),
    \label{eq:delta_target}
\end{equation}
so that the convolution in \eqref{eq:Phi_conv_scalar} collapses to
$\Phi(t)=g\,u(t)$.\cite{Lighthill1958,OppenheimWillsky1997}

For the exponential family \eqref{eq:K_sum_exp_main}, the $\delta$-limit is
obtained by sending all filter time-scales to zero while keeping the total
gain fixed. A convenient sufficient condition is
\begin{equation}
    \nu_k \to \infty
    \quad \text{for all } k,
    \qquad
    \frac{c_k}{\nu_k}\to g_k,
    \qquad
    g:=\sum_{k=1}^{K_{\max}} g_k <\infty,
    \label{eq:delta_scaling_condition}
\end{equation}
with $g_k$ finite constants. Under \eqref{eq:delta_scaling_condition},
$K(\tau)$ converges to $g\,\delta(\tau)$ in the distributional sense on
$\mathbb R_+$,\cite{Lighthill1958,Rudin1987}
and hence $\Phi(t)\to g\,u(t)$ for sufficiently regular inputs.\cite{Rudin1987}

There are several ways to derive this limit. A direct derivation follows by expanding the ODEs. From
\eqref{eq:Phi_mode_ode_main} one may write
\begin{equation}
    \Phi_k(t)=\frac{c_k}{\nu_k}\,u(t)-\frac{1}{\nu_k}\,\dot\Phi_k(t).
    \label{eq:Phi_k_algebraic}
\end{equation}
Summing over $k$ and using $\Phi=\sum_k\Phi_k$ yields
\begin{equation}
    \Phi(t)=\sum_{k=1}^{K_{\max}}\frac{c_k}{\nu_k}\,u(t)
    -\sum_{k=1}^{K_{\max}}\frac{1}{\nu_k}\,\dot\Phi_k(t).
    \label{eq:Phi_sum_algebraic}
\end{equation}
If $\nu_k\to\infty$ while $c_k/\nu_k\to g_k$ and the $\dot\Phi_k$ remain
bounded on the time window of interest, the second term vanishes and one
obtains the instantaneous relation
\begin{equation}
    \Phi(t)\longrightarrow g\,u(t),
    \qquad
    g=\sum_{k=1}^{K_{\max}} g_k.
    \label{eq:Phi_instantaneous_limit}
\end{equation}
Equivalently, the kernel converges to a delta distribution with weight $g$.
Indeed, for the single-mode case $K(\tau)=c e^{-\nu\tau}\Theta(\tau)$, if one
sets $c=\nu$ then $K(\tau)=\nu e^{-\nu\tau}\Theta(\tau)$ is an approximate
identity on $\mathbb R_+$ and converges to $\delta(\tau)$ as $\nu\to\infty$.\cite{Lighthill1958,Rudin1987}

While we will discuss later the physical origin of this model, let us discuss what truncating the series implies. First, truncations yield
finite sums of exponentials with rates $\nu_k$ growing linearly with $k$; the
memoryless limit corresponds to pushing all relevant rates far above the
drive bandwidth. In the frequency domain, this is the familiar condition that the
transfer function of the filter becomes approximately constant over the
support of the protocol, so that the realized control is proportional to the
command with negligible phase lag.\cite{OppenheimWillsky1997,Kailath1980}

\subsection{Measures of hysteresis}
\label{sec:hysteresis_definitions}

Let us now move on to the notion of ``memory'' between a control and an observable.
We use the term \emph{hysteresis} in an operational sense: a driven protocol is
hysteretic whenever the relationship between two time-dependent quantities fails
to be single-valued over a cycle. This definition is directly analogous to the
standard use of hysteresis in magnetism, where a multivalued relation between
magnetic induction $B$ and applied field $H$ over a cycle produces a
characteristic loop, and where the loop area quantifies the energy dissipated
per cycle.\cite{Bertotti1998,Mayergoyz2003,Preisach1935}

In the present setting, such multivalued behavior can arise from multiple
mechanisms. One is \emph{classical memory} in the control channel, encoded by a
causal kernel $K$, which makes the realized field $\Phi(t)$ depend on the history
of the commanded signal $u(t)$. A second is \emph{state-lag} relative to the
realized Hamiltonian: even in a closed system this can occur through
non-adiabatic (finite-rate) driving and produces a purely \emph{dynamical}
hysteresis that vanishes in the quasi-static limit.\cite{BornFock1928,Kato1950,Teufel2003}
More generally, additional physical sources of loops may appear (e.g., due to
ergodicity breaking or open-system effects), and the goal is to disentangle
control-channel memory from these state-dynamics contributions.

A useful feature of the scalar filtered-control model adopted throughout the
paper is that control-channel and state-dynamics effects can be separated
experimentally and numerically by comparing hysteresis measures in the
$(u,\Phi)$ plane (control channel only) and in the $(\Phi,O)$ plane (state
response to the realized Hamiltonian). The notation here is consistent with the
model introduced earlier:
\begin{eqnarray}
    \hat H(t) &=& \hat H_A + \Phi(t)\,\hat M, \nonumber \\
    \Phi(t)  &=& \int_{-\infty}^{t} K(t-s)\,u(s)\,ds,
    \label{eq:hyst_model_scalar}
\end{eqnarray}
with real-valued command $u(t)$ and causal kernel $K$.\cite{OppenheimWillsky1997,Kailath1980}
Throughout, the observable $O(t)$ is always computed from the evolution under
$\hat H(t)$ with the \emph{realized} field $\Phi(t)$.

Assume a periodic protocol with period $T$ after transients have died out, so that
\begin{equation}
    u(t+T)=u(t).
    \label{eq:steady_cycle}
\end{equation}
The drive channel determines a closed parametric curve in the $(u,\Phi)$ plane,
\begin{equation}
    \Gamma_{u\Phi}:\ t\in[0,T]\mapsto \big(u(t),\Phi(t)\big).
\end{equation}
If the channel is memoryless, $\Phi(t)$ is (approximately) an instantaneous
function of $u(t)$ and $\Gamma_{u\Phi}$ collapses to a single-valued curve. With
memory, the same value of $u$ can occur at two different times in the cycle with
different values of $\Phi$, producing a loop, in direct analogy with $B$--$H$
hysteresis loops in ferromagnets or in capacitors/inductors.\cite{Bertotti1998,Mayergoyz2003}

We quantify this \emph{control-channel hysteresis} by the oriented area enclosed
by $\Gamma_{u\Phi}$,
\begin{equation}
    \mathcal{A}_{u\Phi}
    :=
    \oint_{\Gamma_{u\Phi}} \Phi\,du
    =
    \int_0^T \Phi(t)\,\dot u(t)\,dt.
    \label{eq:A_uPhi_def}
\end{equation}
A nonzero $\mathcal{A}_{u\Phi}$ is an intrinsic signature of memory in the map
$u\mapsto \Phi$: it vanishes whenever $\Phi$ is a single-valued function of $u$
along the cycle. In the embedding of Sec.~\ref{sec:embedding_scalar}, this area
measures the extent to which the auxiliary filter variables $\{\Phi_k(t)\}$ lag
the command $u(t)$, in the same way that magnetic hysteresis-loop area measures
lag between $B$ and $H$ in classical Preisach-type descriptions.\cite{Preisach1935,Mayergoyz2003,BoydChua1985}

To probe hysteresis originating from the \emph{state response} to the realized
field, we analogously consider the closed curve in the $(\Phi,O)$ plane,
\begin{equation}
    \Gamma_{\Phi O}:\ t\in[0,T]\mapsto \big(\Phi(t),O(t)\big),
\end{equation}
and define the corresponding oriented area
\begin{equation}
    \mathcal{A}_{\Phi O}
    :=
    \oint_{\Gamma_{\Phi O}} O\,d\Phi
    =
    \int_0^T O(t)\,\dot \Phi(t)\,dt.
    \label{eq:A_PhiO_def}
\end{equation}
In the quasi-static limit for a closed system, $O(t)$ becomes (approximately) a
single-valued function of $\Phi(t)$ and $\mathcal{A}_{\Phi O}\to 0$, whereas
finite-rate (non-adiabatic) driving can produce $\mathcal{A}_{\Phi O}\neq 0$ as a
form of dynamical hysteresis.

Finally, it is useful to note that an ``apparent'' loop can arise if one plots
$O$ against the \emph{command} $u$ rather than against the realized field $\Phi$.
Define
\begin{equation}
    \mathcal{A}_{uO}
    :=
    \oint O\,du
    =
    \int_0^T O(t)\,\dot u(t)\,dt.
    \label{eq:A_uO_def}
\end{equation}
Even when the state has no hysteresis with respect to the realized field (so that
$\mathcal{A}_{\Phi O}=0$), one may still obtain $\mathcal{A}_{uO}\neq 0$ purely
from control-channel memory.

A minimal example is given by pure control-lag produces an apparent loop.
Consider the lag kernel
\begin{equation}
    K(t)=\delta(t-t_{\rm lag}),
    \label{eq:lag_kernel}
\end{equation}
so that
\begin{equation}
    \Phi(t)=\int_{-\infty}^{t}\delta(t-s-t_{\rm lag})\,u(s)\,ds = u(t-t_{\rm lag}).
    \label{eq:lag_realized}
\end{equation}
Assume the system response is strictly memoryless with respect to the realized
field, i.e., $O(t)=f(\Phi(t))$ for a single-valued function $f$. Then the curve
$\Gamma_{\Phi O}$ is single-valued and $\mathcal{A}_{\Phi O}=0$ identically.
Nevertheless, the parametric curve $t\mapsto (u(t),O(t))=(u(t),f(u(t-t_{\rm lag})))$
can enclose a nonzero area $\mathcal{A}_{uO}$ and thus display a seemingly
hysteretic loop in the $(u,O)$ plane. For concreteness, one may take, e.g.,
\begin{equation}
    f(\Phi)=\tanh(a\Phi),
    \label{eq:example_f_tanh}
\end{equation}
which yields a pronounced loop in $(u,O)$ for cyclic drives $u(t)$ despite the
absence of any state memory.\footnote{The author is indebted to P. Sathe for suggesting this example.}

More generally, it is easy to see that the exponential kernel
\begin{equation}
    K(\tau)=\alpha e^{-\alpha\tau}\Theta(\tau),
\end{equation}
the realized field obeys $\dot\Phi=\alpha\big(u-\Phi\big)$. For a sinusoidal
command $u(t)=u_0\sin(\omega t)$, the steady-state response takes the form
\begin{equation}
    \Phi(t)=u_0\frac{\alpha}{\sqrt{\alpha^2+\omega^2}}
    \sin(\omega t-\delta),
    \qquad
    \delta=\arctan\!\frac{\omega}{\alpha}.
    \label{eq:Phi_sine_response}
\end{equation}
The resulting curve $\Gamma_{u\Phi}$ is an ellipse, and the loop area is
\begin{equation}
    |\mathcal{A}_{u\Phi}|
    =
    \pi u_0^2\,\frac{\alpha\omega}{\alpha^2+\omega^2}.
    \label{eq:A_uPhi_explicit}
\end{equation}
Thus control hysteresis is largest when $\omega$ is comparable to the filter
rate $\alpha$, and it vanishes both for $\omega\ll \alpha$ (negligible lag)
and for $\omega\gg \alpha$ (strong attenuation), exactly as in classical
single-pole low-pass responses in electronics and control.\cite{OppenheimWillsky1997,Kailath1980}

Let $\hat O$ be a time-independent observable and define
\begin{equation}
    O(t):=\langle \hat O\rangle_t.
\end{equation}
Plotting the observable against the command produces a closed loop
\begin{equation}
    \Gamma_{uO}:\ t\in[0,T]\mapsto \big(u(t),O(t)\big),
\end{equation}
and we define its oriented area as in eqn. (\ref{eq:A_uO_def}).
In general, $\mathcal{A}_{uO}$ can be nonzero for two reasons: because the
control channel is hysteretic ($u\mapsto\Phi$ is history dependent), or
because the quantum state itself exhibits dynamical lag relative to the
realized Hamiltonian, as in generic nonadiabatic or non-Markovian
evolution.\cite{BornFock1928,Kato1950,BreuerPetruccione2002,RivasHuelgaPlenio2014}

To isolate the quantum contribution, it is natural to ``factor out'' the
control channel by plotting the same observable against the realized field
$\Phi(t)$, defining
\begin{equation}
    \Gamma_{\Phi O}:\ t\in[0,T]\mapsto \big(\Phi(t),O(t)\big),
\end{equation}
and the corresponding loop area as in eqn. (\ref{eq:A_PhiO_def}).
By construction, $\mathcal{A}_{\Phi O}$ is insensitive to purely classical
filtering when the quantum response is effectively instantaneous in $\Phi$.
In that regime, residual hysteresis in the $(\Phi,O)$ plane reflects genuine
quantum dynamical effects rather than distortions in the control channel.

A few results can be proved in this setting that are well known, but that we report here for completeness.
For a closed cycle $\Phi(0)=\Phi(T)$ we write the loop area as the 
integral $A_{uO}:=\oint O\,d\Phi$; when $\Phi$ is absolutely continuous this
reduces to $A_{uO}=\int_0^T O(u)\,\Phi'(u)\,du$.
A key consequence is that if the response is \emph{single-valued} in the realized
field, $O(u)=f(\Phi(u))$, then $A_{uO}=\oint f(\Phi)\,d\Phi=0$, i.e.\ a nonzero
$A_{uO}$ necessarily implies multibranch (history-dependent) behavior.
Independently of uniqueness, boundedness of $O$ yields an a priori variation
bound
\begin{equation}
|A_{uO}|\le \frac{O_{\max}-O_{\min}}{2}\int_0^T |\Phi'(u)|\,du,
\end{equation}
and, under a single turning-point sweep, a two-branch representation
$A_{uO}=\int_{\Phi_{\min}}^{\Phi_{\max}}(O_\uparrow-O_\downarrow)\,d\Phi$ that
immediately gives lower bounds when the branch gap has a definite sign.
All proofs and precise assumptions are collected in
Appendix~\ref{app:integralsbounds}.

\subsection{Classical memory fingerprint in the adiabatic regime}
\label{sec:adiabatic_regime}

In this section we address the regime in which the quantum dynamics is adiabatic
with respect to the \emph{realized} Hamiltonian path generated by the filtered
control $\Phi(t)$. In this limit the quantum response carries essentially no
additional ``memory'' beyond the instantaneous value of $\Phi$, and any
hysteresis observed in the $(u,O)$ plane can be understood as the image, under
a static nonlinear map $O=f(\Phi)$, of the purely classical hysteresis already
present in the control channel $(u,\Phi)$.\cite{BornFock1928,Kato1950,Teufel2003}

We use the scalar filtered-control model,
\begin{equation}
    \hat H(t)=\hat H_A+\Phi(t)\,\hat M,
    \qquad
    \Phi(t)=\int_{-\infty}^{t}K(t-s)\,u(s)\,ds,
\end{equation}
and consider protocols for which transients have decayed so that $u$ and
$\Phi$ are periodic with period $T$.

It is convenient to view the Hamiltonian as a one-parameter family
\begin{equation}
    H(\Phi)=\hat H_A+\Phi\,\hat M,
\end{equation}
and to interpret the drive as prescribing a path $\Phi(t)$ in parameter space.
Assume that along the relevant range of $\Phi$ the spectrum splits into
 energy levels separated by a strictly positive gap. Denote by $\Pi(\Phi)$
the spectral projector onto a chosen band and assume $\Phi\mapsto \Pi(\Phi)$
is smooth.\cite{Teufel2003}

A convenient formulation of adiabatic following is the Kato form: there
exists an intertwiner $W(t)$ solving
\begin{equation}
    \dot W(t)=\big[\dot \Pi(t),\Pi(t)\big]\,W(t),
    \qquad
    W(0)=\mathbb 1,
    \label{eq:kato_intertwiner}
\end{equation}
where $\Pi(t):=\Pi(\Phi(t))$.\cite{Kato1950,Teufel2003,AvronElgart1999}
The operator $W(t)$ transports the instantaneous
spectral subspace along the path, in the sense that
\begin{equation}
    \Pi(t)=W(t)\,\Pi(0)\,W^\dagger(t).
    \label{eq:projector_transport}
\end{equation}
In the adiabatic regime, the exact unitary $\hat U(t)$ generated by $\hat H(t)$ remains
close to a product of the dynamical phase within each band and the
intertwiner $W(t)$, and transitions between distinct levels are suppressed.
Rigorous adiabatic theorems provide quantitative bounds on this approximation
in terms of the gap and time scales of the protocol.\cite{Teufel2003,JansenRuskaiSeiler2007}
For the present purposes, this means that expectation values of observables
are well approximated by functions of the instantaneous spectral data of
$H(\Phi(t))$, rather than by functionals of the entire past history of $\Phi$.
In other words, in the adiabatic limit $\mathcal{A}_{\Phi O}$ becomes a purely
geometric and spectral quantity, while $\mathcal{A}_{uO}$ retains the imprint
of classical memory in the control channel via the map $u\mapsto \Phi$.

The relevant small parameter controlling adiabaticity is set by the speed of
the realized protocol compared to the instantaneous gap.\cite{Kato1950,Teufel2003,JansenRuskaiSeiler2007}
Since $\Phi(t)$ is itself produced by a filter, the adiabatic condition should
be formulated in terms of $\dot\Phi(t)$ (not $\dot u(t)$): memory in the
control can slow down or phase-shift $\Phi$ relative to $u$, thereby modifying
adiabaticity even at fixed command frequency.

\subsection{Adiabatic response function $f(\Phi)$ and vanishing of $\mathcal A_{\Phi O}$}
\label{subsec:fPhi_and_AphiO}

Let $\hat O$ be a time-independent observable and define $O(t)=\langle \hat
O\rangle_t$. In the adiabatic regime with respect to the realized Hamiltonian
path, the state remains (to a good approximation) confined to the
instantaneous spectral decomposition of $H(\Phi)$ with negligible inter-band
transitions. A convenient representation that fixes notation is the adiabatic
diagonal ensemble associated with the initial state, closely related to the
diagonal ensembles used to describe relaxation in isolated quantum systems.\cite{RigolDunjkoOlshanii2008,Polkovnikov2011}

Let $H(\Phi)$ have instantaneous spectral projectors $\Pi_i(\Phi)$ associated
with eigenvalues $\varepsilon_i(\Phi)$, with constant ranks along
the path. Define the initial band weights
\begin{equation}
    p_i=\mathrm{Tr}\!\big(\Pi_i(\Phi(0))\,\rho(0)\big),
    \qquad
    \sum_i p_i=1,
\end{equation}
and introduce the adiabatic reference state
\begin{equation}
    \rho_{\mathrm{ad}}(\Phi)=\sum_i p_i\,\frac{\Pi_i(\Phi)}{\mathrm{Tr}\,\Pi_i(\Phi)}.
    \label{eq:rho_ad_def_main}
\end{equation}
For nondegenerate energy levels, this reduces to frozen populations in instantaneous
eigenstates, consistent with standard adiabatic theorems.\cite{BornFock1928,Kato1950,Teufel2003}
The corresponding adiabatic response function is
\begin{equation}
    f(\Phi)=\mathrm{Tr}\!\big(\rho_{\mathrm{ad}}(\Phi)\,\hat O\big)
    =\sum_i \frac{p_i}{\mathrm{Tr}\,\Pi_i(\Phi)}\,\mathrm{Tr}\!\big(\Pi_i(\Phi)\,\hat O\big).
    \label{eq:fPhi_main}
\end{equation}
Along an adiabatic trajectory $\Phi(t)$ one then has
\begin{equation}
    O(t)\approx f(\Phi(t)).
    \label{eq:O_following}
\end{equation}

Since $f(\Phi)$ is single-valued, the oriented loop area in the $(\Phi,O)$
plane vanishes,
\begin{equation}
    \mathcal A_{\Phi O}
    =\oint O\,d\Phi
    \approx \oint f(\Phi)\,d\Phi
    =0.
\end{equation}
Thus, once $\Phi$ is taken as the independent variable, adiabatic quantum
evolution does not generate additional hysteresis beyond what is already
present in the control channel.

Even though $\mathcal A_{\Phi O}$ vanishes in the adiabatic regime, the loop
in the $(u,O)$ plane is generically nontrivial whenever the \emph{control
channel} is hysteretic. Combining the loop-area definitions with the
adiabatic approximation \eqref{eq:O_following} yields
\begin{equation}
    \mathcal A_{uO}
    =\oint O\,du
    \approx \int_0^T f(\Phi(t))\,\dot u(t)\,dt.
    \label{eq:AuO_adiabatic_start}
\end{equation}
Thus, in the adiabatic regime the observable loop is the image of the
classical loop $t\mapsto(u(t),\Phi(t))$ under the nonlinear transformation
$\Phi\mapsto f(\Phi)$.

A transparent consequence is obtained in the weak-modulation limit. Suppose
the command is scaled as $u(t)=\varepsilon u_1(t)$ with $\varepsilon\ll 1$.
By linearity of the filter, $\Phi(t)=\varepsilon \Phi_1(t)$, and expanding
$f(\Phi)$ about $\Phi=0$ gives
\begin{equation}
    f(\Phi)=f(0)+f'(0)\,\Phi+O(\Phi^2).
\end{equation}
The constant term drops out of \eqref{eq:AuO_adiabatic_start} because the
cycle is closed, $\int_0^T \dot u(t)\,dt=u(T)-u(0)=0$, leaving
\begin{equation}
    \mathcal A_{uO}
    \approx f'(0)\int_0^T \Phi(t)\,\dot u(t)\,dt
    = f'(0)\,\mathcal A_{u\Phi}
    +O(\varepsilon^3).
    \label{eq:AuO_AuPhi_linear_main}
\end{equation}
Equation \eqref{eq:AuO_AuPhi_linear_main} separates roles: the control memory
enters only through the classical loop area $\mathcal A_{u\Phi}$, while the
quantum system contributes a static factor $f'(0)$ determined by spectral
properties of $H(\Phi)$ and by the choice of observable and initial state.

Differentiating \eqref{eq:fPhi_main} yields a projector expression for
$f'(0)$ (assuming $\hat O$ has no explicit $\Phi$ dependence),
\begin{equation}
    f'(0)
    =\sum_i \frac{p_i}{\mathrm{Tr}\,\Pi_i(0)}\,
    \mathrm{Tr}\!\Big(\big(\partial_\Phi \Pi_i(\Phi)\big)_{\Phi=0}\,\hat O\Big),
    \label{eq:fprime_projector_main}
\end{equation}
which can be made fully explicit using standard resolvent identities for
$\partial_\Phi\Pi_i$ when the gap is nonzero.\cite{Teufel2003}
As we show in App.~\ref{app:expressionfprime0},
introducing the reduced resolvent on the complement of the $i$th band,
$$
R_i(0)=(\mathbb 1-\Pi_i(0))\,(H(0)-\varepsilon_i(0))^{-1}\,(\mathbb 1-\Pi_i(0)),
$$
one has the Kato-type identity\cite{Kato1950}
$$
\big(\partial_\Phi\Pi_i(\Phi)\big)_{\Phi=0}
=
R_i(0)\,\hat M\,\Pi_i(0)+\Pi_i(0)\,\hat M\,R_i(0),
$$
and therefore
$$
f'(0)=2\sum_i \frac{p_i}{\mathrm{Tr}\,\Pi_i(0)}\,
\mathsf{Re}\,\mathrm{Tr}\!\Big(\Pi_i(0)\,\hat O\,R_i(0)\,\hat M\,\Pi_i(0)\Big).
$$
Equivalently, expanding $R_i(0)$ over the remaining  energy levels yields the
sum-over-levels (first-order perturbation) form\cite{SakuraiNapolitano2017}
$$
f'(0)=2\sum_i \frac{p_i}{\mathrm{Tr}\,\Pi_i(0)}\sum_{j\neq i}
\frac{\mathsf{Re}\,\mathrm{Tr}\!\Big(\Pi_i(0)\,\hat O\,\Pi_j(0)\,\hat M\,\Pi_i(0)\Big)}
{\varepsilon_i(0)-\varepsilon_j(0)},
$$
which makes the dependence of the proportionality factor $f'(0)$ on the
spectrum of $H(0)$ and on the matrix elements of $\hat O$ and $\hat M$ explicit.

\subsection{Nonadiabatic response: separating control-channel memory from quantum information backflow}
\label{sec:nonadiabatic_response}

Section~\ref{sec:adiabatic_regime} treated the regime in which the observable
response can be expressed as a single-valued function of the realized control
field, $O(t)\simeq f(\Phi(t)),$ so that the loop in the $(\Phi,O)$ plane
collapses and $\mathcal A_{\Phi O}\simeq 0.$ We now move beyond this limit.

Outside the adiabatic regime, the state does not remain confined to an
instantaneous eigenspace of $H(\Phi(t))$. Even though the qubit is still closed
and the control memory is still classical, the observable acquires dependence
on the history of $\Phi$ through coherent transitions and phase accumulation.
Operationally, the loop area $\mathcal A_{\Phi O}$ becomes the relevant
diagnostic: it vanishes when $O$ is (approximately) single-valued in $\Phi$,
and it becomes nonzero when the same realized field value $\Phi$ can correspond
to different quantum states depending on where one is on the cycle.

A systematic analytic treatment can be built from an adiabatic perturbation
theory in the instantaneous eigenbasis of $H(\Phi)$ that, however, goes beyond the scope of this work. Nonetheless, for an initial ground-state
preparation, the leading nonadiabatic correction is controlled by the
off-diagonal coupling $\langle e(\Phi)|\partial_t g(\Phi)\rangle$, which is
proportional to $\dot\Phi$ and suppressed by the gap. This suggests the scaling
structure
\begin{equation}
    \mathcal A_{\Phi O}
    \;\sim\;
    \ell(\hat O)\times
    \int_0^T \frac{\dot \Phi(t)^2}{\Delta(\Phi(t))^3}\,dt,
\end{equation}
where $\ell(\hat O)$ is a prefactor depending on the choice of observable, so
that, at fixed realized amplitude $A$, nonadiabatic loop effects grow with
frequency and are suppressed by a large gap. The precise prefactor and
functional form depend on the chosen observable and on the initial state, and
can be evaluated explicitly for $\hat O=\sigma_x,\sigma_z$ once the
instantaneous eigenvectors are fixed.

In the exponential kernel model, increasing $\omega$ simultaneously (i) increases the
classical phase lag and (ii) decreases the realized amplitude $A$. These
effects compete:
\begin{equation}
    \text{classical hysteresis} \;\; |\mathcal A_{u\Phi}| \;\;
    \text{peaks at} \;\; \omega\tau_c\approx 1,
\end{equation}
while
\begin{equation}
    \text{nonadiabaticity (at fixed $A$)} \;\; \text{typically grows with} \;\; \omega.
\end{equation}
Because $A=A(\omega)$ decreases with $\omega$ for a fixed command amplitude
$u_0$, the net nonadiabatic signature can exhibit a crossover: as $\omega$
increases from zero, classical hysteresis first grows; at larger $\omega$,
classical hysteresis decays, but nonadiabatic corrections may remain visible
if the realized sweep rate $\omega A(\omega)$ is still large enough relative
to the gap scale. This is exactly the regime where $\mathcal A_{\Phi O}$
becomes essential: it isolates state-history effects that cannot be attributed
to the control filter alone.

In the nonadiabatic regime the state no longer follows instantaneous spectral
subspaces of $\hat H(t)=\hat H_A+\Phi(t)\hat M,$ and the mapping $\Phi \mapsto O$ can
become history dependent even when the only source of \emph{physical} memory
is the classical filtering relation $\Phi = K*u.$ Consequently,
$\mathcal A_{\Phi O}$ can be nonzero purely because of coherent unitary
dynamics, as in standard studies of nonadiabatic work and dissipation in
driven closed systems~\cite{Campisi2011,Talkner2007}.

Beyond these general comments, this section has two aims. First, we state an exact identity that rewrites
cyclic ``work-like'' integrals as commutator-weighted functionals of the
control, without any adiabatic approximation. Second, we explain why such
nonadiabatic hysteresis measures, while operationally useful, do not by
themselves certify \emph{quantum memory} in the open-system sense of
information backflow~\cite{Breuer2016,Rivas2014}. The technical derivations
and bounds are deferred to
Appendix~\ref{app:nonadiabatic_cyclic_integrals}.

We consider a cyclic protocol of duration $T$ in steady periodic response,
$$u(T)=u(0),\qquad \rho(T)=\rho(0),$$
and define the loop functional
\begin{equation}
    I \;=\; \oint dt\;\mathrm{Tr}\!\big(\rho(t)\hat O\big)\,\dot u(t);
    \label{eq:I_def_main_clean}
\end{equation}
for simplicity, we consider a time-independent observable $\hat O$. For unitary
evolution, $\dot\rho(t)=-i[\hat H(t),\rho(t)],$ integration by parts yields the exact
identity
\begin{equation}
    I \;=\; i\oint dt\;u(t)\,\big\langle[\hat O,\hat H(t)]\big\rangle_t.
    \label{eq:I_commutator_main_clean}
\end{equation}
We specialize to the scalar filtered-control model used throughout the paper,
\begin{equation}
    \hat H(t)=\hat H_A+\Phi(t)\hat M,
    \qquad
    \Phi(t)=\int_{0}^{t}K(t-s)\,u(s)\,ds,
    \label{eq:scalar_filtered_model_main_clean}
\end{equation}
with real causal kernel $K$ and real command $u$. Introducing the
time-dependent commutator weights
\begin{equation}
    a(t)=\big\langle[\hat O,\hat H_A]\big\rangle_t,
    \qquad
    b(t)=\big\langle[\hat O,\hat M]\big\rangle_t,
    \label{eq:ab_def_main_clean}
\end{equation}
one obtains an exact decomposition
\begin{equation}
    I
    \;=\;
    i\oint_0^T dt\;u(t)\,a(t)
    \;+\;
    i\oint_0^T dt\;u(t)\,\Phi(t)\,b(t).
    \label{eq:I_split_main_clean}
\end{equation}
The first term is present even in the instantaneous-control limit and depends
on correlations between the command waveform and the internal ``bias''
quantity $a(t)$. The second term is the contribution in which the control
history enters explicitly through the convolution $\Phi=K*u$. When $b(t)$
is slowly varying or approximately constant over a cycle, the second term
reduces to a quadratic functional of $u$ (as in the adiabatic discussion),
but Eq.~\eqref{eq:I_split_main_clean} remains exact in the fully nonadiabatic
unitary regime.

A compact estimate that makes the kernel dependence explicit is obtained by
bounding the two contributions separately. Writing
$\|u\|_2=\|u\|_{L^2(0,T)}$ and
$\|K\|_1=\|K(\tau)\Theta(\tau)\|_{L^1(\mathbb R_+)},$
Appendix~\ref{app:nonadiabatic_cyclic_integrals} shows that
\begin{equation}
    |I|
    \;\le\;
    \|u\|_{2}\,\|a\|_{2}
    \;+\;
    \|b\|_{\infty}\,\|K\|_{1}\,\|u\|_{2}^{2},
    \label{eq:I_total_bound_main}
\end{equation}
where $\|b\|_\infty=\mathrm{sup}_{t\in[0,T]}|b(t)|.$
The first term is independent of the kernel and scales linearly with the
drive amplitude, while the second term is the kernel-mediated contribution
and scales quadratically, with explicit dependence on the control-channel
memory strength through $\|K\|_1$.

\section{Results}
\subsection{Properties of the decay channels}
\label{subsec:single_channel}

We begin with the minimal control-channel model: a single passive first-order
low-pass element relating the commanded waveform $u(t)$ to the realized field
$\Phi(t)$ at the device node,
\begin{equation}
    \tau_c\,\dot \Phi(t)+\Phi(t)=u(t),
    \qquad
    \tau_c>0.
    \label{eq:single_RC}
\end{equation}
Equivalently, $\Phi=(K*u)$ with the causal kernel
\begin{equation}
    K(\tau)=\frac{1}{\tau_c}e^{-\tau/\tau_c}\Theta(\tau),
    \label{eq:single_RC_kernel}
\end{equation}
so $\tau_c$ sets the classical memory time of the control line.

In steady periodic operation, a sinusoidal command $u(t)=u_0\sin(\omega t)$
produces a single-harmonic realized response,
\begin{eqnarray}
    \Phi(t) &=& A\,\sin(\omega t-\delta), \nonumber \\
    A &=& \frac{u_0}{\sqrt{1+(\omega\tau_c)^2}},
    \qquad
    \delta=\arctan(\omega\tau_c),
    \label{eq:single_RC_sine_response}
\end{eqnarray}
as follows from the transfer function $G(i\omega)=1/(1+i\omega\tau_c)$.
Thus the channel memory has two experimentally accessible signatures: amplitude
roll-off $A/u_0$ and a phase lag $\delta$. Both will reappear below as the
frequency dependence of loop areas.

The ideal instantaneous-channel limit corresponds to the kernel concentrating at
$\tau=0$ so that $\Phi(t)\to u(t)$. For the single-pole model this is achieved by
$\tau_c\to 0$ at fixed static gain. Operationally, for band-limited commands with
characteristic frequency $\omega$, the condition $\omega\tau_c\ll 1$ implies
\begin{equation}
    \Phi(t)=u(t)+\mathcal O(\omega\tau_c),
    \label{eq:Phi_approx_u_small_omega_tau}
\end{equation}
and the classical control hysteresis vanishes in this limit.

After transients die out, the pair $(u(t),\Phi(t))$ traces a closed curve in the
$(u,\Phi)$ plane over one period $T=2\pi/\omega$. We quantify the resulting
classical control hysteresis by the oriented loop area
\begin{equation}
    \mathcal A_{u\Phi}
    :=
    \oint \Phi\,du
    =
    \int_0^T \Phi(t)\,\dot u(t)\,dt.
    \label{eq:A_uPhi_def_single_channel}
\end{equation}
For the single-pole response \eqref{eq:single_RC_sine_response}, direct
evaluation gives
\begin{equation}
    \mathcal A_{u\Phi}
    =
    -\pi u_0^2\,\frac{\omega\tau_c}{1+(\omega\tau_c)^2},
    \label{eq:A_uPhi_single_pole}
\end{equation}
i.e.\ $|\mathcal A_{u\Phi}|$ is maximized at $\omega\tau_c=1$ and decays for both
$\omega\tau_c\ll 1$ (negligible lag) and $\omega\tau_c\gg 1$ (strong attenuation).
This nonmonotone dependence is the classical ``memory resonance'' of a first-order
low-pass element.

For a measured observable $O(t)=\langle\psi(t)|\hat O|\psi(t)\rangle$ we define
\begin{eqnarray}
    \mathcal A_{uO}
    &:=&
    \oint O\,du
    =
    \int_0^T O(t)\,\dot u(t)\,dt,
    \nonumber \\
    \mathcal A_{\Phi O}
    &:=&
    \oint O\,d\Phi
    =
    \int_0^T O(t)\,\dot\Phi(t)\,dt.
    \label{eq:A_uO_A_PhiO_def_single_channel}
\end{eqnarray}
In an adiabatic regime with respect to the realized Hamiltonian path
$H(\Phi)=\sigma_z+\Phi\sigma_x$, and for preparations that follow an
instantaneous eigenspace (e.g.\ ground-state following), $O(t)$ becomes a
single-valued function of the realized field,
\begin{equation}
    O(t)\approx f(\Phi(t)).
    \label{eq:O_fPhi_single_channel}
\end{equation}
Consequently the $(\Phi,O)$ loop collapses to a curve and
\begin{equation}
    \mathcal A_{\Phi O}\approx 0.
    \label{eq:A_PhiO_zero_single_channel}
\end{equation}
This provides a clean operational separation: in the adiabatic regime any
hysteresis observed in the commanded plane $(u,O)$ is not evidence of intrinsic
quantum memory, but rather a deterministic consequence of the classical filter
$u\mapsto\Phi$ together with the static nonlinear map $\Phi\mapsto f(\Phi)$.

If the drive amplitude is small enough that $\Phi(t)$ remains near $0$, then
\begin{equation}
    f(\Phi)=f(0)+f'(0)\Phi+\mathcal O(\Phi^2).
    \label{eq:f_Taylor_single_channel}
\end{equation}
Since $\int_0^T \dot u(t)\,dt=0$, the constant term does not contribute to
$\mathcal A_{uO}$, and we obtain
\begin{equation}
    \mathcal A_{uO}
    =
    f'(0)\,\mathcal A_{u\Phi}
    +\mathcal O(\text{amplitude}^3).
    \label{eq:A_uO_fprime_A_uPhi_single_channel}
\end{equation}
In particular, for ground-state response with $\hat O=\sigma_x$ in the Hamiltonian of eqn. (\ref{eq:H_qubit_control_channel_section}), one has $f'(0)=-1$, so to leading order
\begin{equation}
    \mathcal A_{uO}\approx -\,\mathcal A_{u\Phi},
    \label{eq:A_uO_minus_A_uPhi_single_channel}
\end{equation}
and hence inherits the characteristic single-pole frequency dependence
\begin{equation}
    |\mathcal A_{uO}|
    \propto
    \frac{\omega\tau_c}{1+(\omega\tau_c)^2}
    \qquad
    \text{(adiabatic, weak-drive regime)}.
    \label{eq:A_uO_freq_fingerprint_single_channel}
\end{equation}

A convenient sufficient adiabaticity criterion compares the rate of change of
the Hamiltonian to the instantaneous gap. For $H(\Phi)=\sigma_z+\Phi\sigma_x$ the
gap is
\begin{equation}
    \Delta(\Phi)=2\sqrt{1+\Phi^2},
    \label{eq:gap_single_channel}
\end{equation}
and a typical condition has the form
\begin{equation}
    \frac{|\dot\Phi(t)|}{\Delta(\Phi(t))^2}\ll 1
    \quad\text{uniformly on the cycle}.
    \label{eq:adiabaticity_criterion_single_channel}
\end{equation}
For the steady-state filtered sinusoid \eqref{eq:single_RC_sine_response},
$|\dot\Phi|\le \omega A$, so low frequency and/or strong filtering (small $A$)
both promote adiabaticity. This observation will be used to organize the
adiabatic versus nonadiabatic regimes in later sections.

The single-exponential kernel \eqref{eq:single_RC_kernel} corresponds to the
impulse response of one lumped RC element, i.e.\ a single relaxational degree of
freedom in the control line. In realistic cryogenic wiring, however, the command
generated at room temperature traverses a \emph{distributed} passive network with
many internal charge-storage and dissipation mechanisms: distributed capacitance
to ground, resistive losses (including thermalized attenuators), bias tees and
low-pass sections, connectors and packaging parasitics, and on-chip filtering.
A standard circuit idealization of such a channel is an \emph{RC ladder} (a lossy
line discretized into many RC sections), which is still linear and passive but
typically high order.

For passive RC networks the transfer function $G(s)=\Phi(s)/u(s)$ is stable and,
for ladder families of interest, has poles on the negative real axis; in
particular, for a finite ladder $G$ is rational with real negative poles.
Consequently the impulse response is a superposition of decaying exponentials.
Appendix~\ref{app:rc_realizability} derives this explicitly for an $N$-stage ladder
and relates it to classical RC realizability/synthesis results
\cite{FialkowGerst1951,Magos1970}. In the time domain this motivates representing
the control kernel as a finite exponential-mode expansion,
\begin{equation}
    K(\tau)=\sum_{k=1}^{K_{\max}} c_k\,e^{-\nu_k\tau}\Theta(\tau),
    \qquad
    \nu_k>0,
    \label{eq:multi_exp_kernel_single_channel}
\end{equation}
which is \emph{exact} whenever the channel admits a finite RC-ladder realization
(and remains a controlled approximation for distributed lines).

Introducing auxiliary mode variables
\begin{eqnarray}
    \Phi_k(t):&=&\int_{-\infty}^{t} c_k e^{-\nu_k(t-s)}u(s)\,ds,
    \nonumber \\
\Phi(t)&=&\sum_{k=1}^{K_{\max}}\Phi_k(t),
\label{eq:Phi_modes_def_single_channel}
\end{eqnarray}
turns the convolutional memory into the time-local realization
\begin{equation}
    \dot\Phi_k(t)=-\nu_k\Phi_k(t)+c_k u(t),
    \qquad
    k=1,\dots,K_{\max}.
    \label{eq:Phi_modes_ODE_single_channel}
\end{equation}
For an actual finite ladder, the rates $\{\nu_k\}$ are precisely the relaxation
rates of the circuit modes (the poles of $G$), while the weights $\{c_k\}$
encode the input--output coupling of each mode (Appendix~\ref{app:rc_realizability}).

This representation also yields an explicit expression for the classical
control-channel hysteresis area under harmonic drive. For
$u(t)=u_0\sin(\omega t)$ in steady state, linearity implies that
$\Phi(t)$ is again purely sinusoidal at frequency $\omega$, with complex gain
\begin{equation}
    G(i\omega)=\frac{\Phi(i\omega)}{u(i\omega)}
    =\sum_{k=1}^{K_{\max}}\frac{c_k}{\nu_k+i\omega}.
    \label{eq:G_multi_modes}
\end{equation}
Equivalently, writing $G(i\omega)=|G(i\omega)|e^{-i\delta(\omega)}$ with
$\delta(\omega)\in[0,\pi/2)$, one has
$\Phi(t)=u_0|G(i\omega)|\sin(\omega t-\delta(\omega))$.
The oriented loop area in the $(u,\Phi)$ plane,
\begin{equation}
    \mathcal A_{u\Phi}=\oint \Phi\,du=\int_0^T \Phi(t)\,\dot u(t)\,dt,
    \qquad T=\frac{2\pi}{\omega},
\end{equation}
can then be evaluated in closed form:
\begin{equation}
    \mathcal A_{u\Phi}
    =
    -\pi u_0^2\,|G(i\omega)|\sin\delta(\omega)
    =
    \pi u_0^2\,\Im\,G(i\omega).
    \label{eq:AuPhi_ImG_general}
\end{equation}
We anticipate that the formula above, derived from a classical approximation to the drive channel, has a quantum analog derived in a later section.

Using \eqref{eq:G_multi_modes} this becomes an explicit sum over relaxation
channels,
\begin{equation}
    \mathcal A^{K_{max}}_{u\Phi}
    =
    -\pi u_0^2\,\omega\sum_{k=1}^{K_{\max}}\frac{c_k}{\nu_k^2+\omega^2}.
    \label{eq:AuPhi_multi_modes_explicit}
\end{equation}
In particular, each mode contributes a negative (lag-induced) area of magnitude
$\propto \omega/(\nu_k^2+\omega^2)$, peaking near $\omega\simeq \nu_k$ and
vanishing both for $\omega\ll \nu_k$ (negligible lag) and for $\omega\gg \nu_k$
(strong attenuation). Thus a multi-mode RC channel generically produces a
superposition of ``memory resonances'' across the hierarchy of relaxation rates. For the hysteresis of the quantum observable, we obtain then that in the adiabatic regime we have
\begin{eqnarray}
    \mathcal A_{uO}\approx f'(0) \mathcal A^{K_{max}}_{u\Phi}.
\end{eqnarray}

\subsection{Single-qubit case}
\label{sec:qubit_illustrations}

This section collects two analytically controlled qubit examples that we use
as reference points before turning to numerics. The first illustrates the
adiabatic regime for the transverse-drive model used throughout the paper,
where the qubit response becomes an (approximately) single-valued function of
the \emph{realized} field $\Phi(t)$. The second is an exactly solvable
(commuting) model in which the control enters longitudinally; it provides a
closed-form benchmark for how a classical control filter can generate a
hysteretic $(u,\Phi)$ loop even when the unitary evolution requires no time
ordering.

Throughout we adopt the notation of Secs.~\ref{sec:model} and
\ref{sec:hysteresis_definitions}: the commanded waveform is $u(t)$, the realized
field is $\Phi(t)=(K*u)(t)$ (or its time-local embedding in
Sec.~\ref{sec:embedding_scalar}), and hysteresis is quantified by loop areas
such as $\mathcal A_{u\Phi}=\oint \Phi\,du$ and
$\mathcal A_{uO}=\oint O\,du$.
\subsubsection{Single-qubit illustration and the exponential kernel}
\label{subsec:adiabatic_qubit_exp_kernel}

Let us compute the adiabatic map $f(\Phi)$ in a simple qubit example. Fix
$\Phi\in\mathbb R$ and consider the instantaneous Hamiltonian
\begin{equation}
    H(\Phi)=\sigma_z+\Phi\sigma_x=\bm h(\Phi)\cdot\bm\sigma,
    \qquad 
    \bm h(\Phi)=(\Phi,0,1),
\end{equation}
with eigenvalues $\pm\|\bm h(\Phi)\|$ and gap $2\|\bm h(\Phi)\|$, where
\begin{equation}
    \|\bm h(\Phi)\|=\sqrt{1+\Phi^2}.
\end{equation}
For any Pauli-vector Hamiltonian $H=\bm h\cdot\bm\sigma$, the ground state has
Bloch vector antiparallel to $\bm h$,
\begin{equation}
    \bm n_g(\Phi):=\big(\langle\sigma_x\rangle_{g(\Phi)},\langle\sigma_y\rangle_{g(\Phi)},\langle\sigma_z\rangle_{g(\Phi)}\big)
    =-\frac{\bm h(\Phi)}{\|\bm h(\Phi)\|}.
\end{equation}
Therefore
\begin{equation}
    \langle\sigma_x\rangle_{g(\Phi)}
    = -\frac{h_x(\Phi)}{\|\bm h(\Phi)\|}
    = -\frac{\Phi}{\sqrt{1+\Phi^2}}.
    \label{eq:sigmax_ground_Hz_plus_PhiSx}
\end{equation}

In the adiabatic-following regime (slow variation of $\Phi(t)$ relative to the
instantaneous gap), a state prepared in $\ket{g(\Phi(0))}$ remains close to
$\ket{g(\Phi(t))}$ up to phases, hence
\begin{equation}
    O(t)=\langle \sigma_x\rangle_t \approx \langle\sigma_x\rangle_{g(\Phi(t))}
    = -\frac{\Phi(t)}{\sqrt{1+\Phi(t)^2}}
    \equiv f(\Phi(t)),
\end{equation}
in agreement with the standard Bloch-sphere description of driven two-level
systems.\cite{NielsenChuang}
If $O(t)\approx f(\Phi(t))$ with single-valued $f$, then the $(\Phi,O)$ loop
collapses to a curve and
\begin{equation}
    \mathcal A_{\Phi O}
    =\oint O\,d\Phi
    \approx \oint f(\Phi)\,d\Phi
    =\oint dF(\Phi)
    =0,
\end{equation}
for $F'(\Phi)=f(\Phi)$ and cyclic protocols.

\subsubsection{Transverse filtered drive: instantaneous eigenbasis and adiabatic response}
\label{subsec:qubit_transverse_adiabatic}

We consider the transverse-drive qubit used in the main text,
\begin{equation}
    \hat H(t)=\Omega_z\,\sigma_z+\Phi(t)\,\sigma_x,
    \qquad
    \Phi(t)=(K*u)(t),
    \label{eq:qubit_transverse}
\end{equation}
with $\Omega_z>0$. Writing $\hat H(t)=\bm \hat H(t)\cdot\bm\sigma$ gives
\begin{equation}
    \bm \hat H(t)=\big(\Phi(t),\,0,\,\Omega_z\big),
    \qquad
    E(t)=\|\bm \hat H(t)\|=\sqrt{\Omega_z^2+\Phi(t)^2}.
\end{equation}
The instantaneous eigenvalues are $\pm E(t)$, and the instantaneous ground-state
Bloch direction is
\begin{equation}
    \bm n_g(t)=-\frac{\bm \hat H(t)}{\|\bm \hat H(t)\|}
    =
    -\frac{1}{\sqrt{\Omega_z^2+\Phi(t)^2}}\big(\Phi(t),\,0,\,\Omega_z\big).
    \label{eq:ng_def}
\end{equation}

\subsubsection{Transverse filtered drive: adiabatic reduction and small-amplitude loop areas}
\label{subsec:qubit_transverse_adiabatic_expanded}

Assume periodic steady state, $u(t+T)=u(t)$ and $\Phi(t+T)=\Phi(t)$. The loop
areas introduced in Sec.~\ref{sec:hysteresis_definitions} are
\begin{eqnarray}
    \mathcal A_{u\Phi}&=&\oint \Phi\,du=\int_0^T \Phi(t)\,\dot u(t)\,dt,\nonumber \\
    \mathcal A_{uO}&=&\oint O\,du=\int_0^T O(t)\,\dot u(t)\,dt,\nonumber \\
    \mathcal A_{\Phi O}&=&\oint O\,d\Phi=\int_0^T O(t)\,\dot\Phi(t)\,dt.
    \label{eq:areas_recalled}
\end{eqnarray}

In the adiabatic regime with respect to the realized trajectory $\Phi(t)$, and
for preparations diagonal in the instantaneous eigenbasis at $\Phi(0)$, the
expectation of any fixed observable $\hat O$ is (to leading adiabatic order) a
single-valued function of $\Phi$,
\begin{equation}
    O(t)=\langle \hat O\rangle_t \approx f(\Phi(t)).
    \label{eq:O_fPhi_again}
\end{equation}
This is the standard structure of the adiabatic theorem in gapped systems
\cite{Kato1950,AvronElgart1999,Teufel2003}. For ground-state following,
\begin{equation}
    \langle\sigma_x\rangle_t \approx -\frac{\Phi(t)}{\sqrt{\Omega_z^2+\Phi(t)^2}},
    \qquad
    \langle\sigma_z\rangle_t \approx -\frac{\Omega_z}{\sqrt{\Omega_z^2+\Phi(t)^2}}.
    \label{eq:adiabatic_sigmas_closed}
\end{equation}
In particular, \eqref{eq:O_fPhi_again} implies
\begin{equation}
    \mathcal A_{\Phi O}=\oint O\,d\Phi \approx \oint f(\Phi)\,d\Phi = 0,
    \label{eq:APhiO_zero_again}
\end{equation}
so any nonzero $\mathcal A_{\Phi O}$ indicates nonadiabaticity (or additional
degrees of freedom), not the classical filter alone.

\paragraph{Weak-drive expansion.}
Assume an amplitude scaling
\begin{equation}
    u(t)=\varepsilon\,u_1(t),\qquad 0<\varepsilon\ll 1,
    \label{eq:u_eps}
\end{equation}
with $u_1$ fixed, $T$-periodic, and not identically constant. By linearity,
\begin{equation}
    \Phi(t)=(K*u)(t)=\varepsilon\,\Phi_1(t),\qquad \Phi_1(t)=(K*u_1)(t),
    \label{eq:Phi_eps}
\end{equation}
so $\mathcal A_{u\Phi}=\varepsilon^2\,\mathcal A_{u_1\Phi_1}$. If $f$ is smooth at
$\Phi=0$,
\begin{equation}
    f(\Phi)=f(0)+f'(0)\Phi+\frac{1}{2}f''(0)\Phi^2+\mathcal O(\Phi^3).
    \label{eq:f_taylor}
\end{equation}
Using $\int_0^T \dot u(t)\,dt=0$, inserting into $\mathcal A_{uO}$ yields
\begin{eqnarray}
    \mathcal A_{uO}
    &=&
    f'(0)\int_0^T \Phi(t)\,\dot u(t)\,dt \nonumber \\
    &&+\frac{1}{2}f''(0)\int_0^T \Phi(t)^2\,\dot u(t)\,dt
    +\mathcal O(\varepsilon^4),
    \label{eq:AuO_expand_general}
\end{eqnarray}
hence
\begin{eqnarray}
    \mathcal A_{uO}
    &=&
    f'(0)\,\mathcal A_{u\Phi}
    +\mathcal O(\varepsilon^3).
    \label{eq:AuO_fprime_relation_expanded}
\end{eqnarray}

For ground-state following, define
\begin{equation}
    f_x(\Phi)=-\frac{\Phi}{\sqrt{\Omega_z^2+\Phi^2}},
    \qquad
    f_z(\Phi)=-\frac{\Omega_z}{\sqrt{\Omega_z^2+\Phi^2}},
\end{equation}
with small-$\Phi$ expansions
\begin{eqnarray}
    f_x(\Phi)&&= -\frac{1}{\Omega_z}\Phi + \frac{1}{2\Omega_z^3}\Phi^3+\mathcal O(\Phi^5),\nonumber \\
    f_z(\Phi)&&= -1 + \frac{1}{2\Omega_z^2}\Phi^2 + \mathcal O(\Phi^4).
    \label{eq:fx_fz_series}
\end{eqnarray}

\smallskip
\noindent
\emph{(i) Observable loop for $O=\sigma_x$.}
Since $f_x'(0)=-1/\Omega_z\neq 0$,
\begin{equation}
    \mathcal A_{u,\sigma_x}
    =
    -\frac{1}{\Omega_z}\,\mathcal A_{u\Phi}
    +\mathcal O(\varepsilon^3).
    \label{eq:AuOx_leading}
\end{equation}

\smallskip
\noindent
\emph{(ii) Observable loop for $O=\sigma_z$.}
Here $f_z'(0)=0$, so
\begin{equation}
    \mathcal A_{u,\sigma_z}
    =
    \frac{1}{2\Omega_z^2}\int_0^T \Phi(t)^2\,\dot u(t)\,dt
    +\mathcal O(\varepsilon^4)
    =
    \mathcal O(\varepsilon^3).
    \label{eq:AuOz_scaling}
\end{equation}

\paragraph{Single-pole kernel under harmonic command.}
Take the kernel
\begin{equation}
    K(\tau)=\alpha e^{-\alpha \tau}\,\Theta(\tau),
    \qquad \alpha>0,
    \label{eq:K_single_pole_again}
\end{equation}
and harmonic command
\begin{equation}
    u(t)=u_0\sin(\omega t),
    \qquad T=\frac{2\pi}{\omega}.
    \label{eq:u_sine_again}
\end{equation}
In steady state,
\begin{eqnarray}
    \Phi(t)&&=A\sin(\omega t-\delta),
    \qquad\nonumber \\
    A&&=u_0\frac{\alpha}{\sqrt{\alpha^2+\omega^2}},
    \qquad
    \delta=\arctan\!\Big(\frac{\omega}{\alpha}\Big).
    \label{eq:Phi_sine_again}
\end{eqnarray}
Then
\begin{eqnarray}
    \mathcal A_{u\Phi}
    &&=\int_0^T \Phi(t)\,\dot u(t)\,dt
    =-\pi u_0 A \sin\delta,
    \label{eq:AuPhi_ellipse_general}
\end{eqnarray}
so
\begin{equation}
    |\mathcal A_{u\Phi}|
    =
    \pi u_0^2\,\frac{\alpha\omega}{\alpha^2+\omega^2}.
    \label{eq:AuPhi_singlepole_closed}
\end{equation}
Combining \eqref{eq:AuOx_leading} and \eqref{eq:AuPhi_singlepole_closed} gives
\begin{eqnarray}
    |\mathcal A_{u,\sigma_x}|
    =
    \frac{\pi u_0^2}{\Omega_z}\,\frac{\alpha\omega}{\alpha^2+\omega^2}
    +\mathcal O(u_0^3).
    \label{eq:AuOx_explicit_singlepole}
\end{eqnarray}

These relations quantify how the classical memory scale (here $\alpha^{-1}$)
imprints itself onto commanded loop areas even when $\mathcal A_{\Phi O}\approx 0$.\cite{Teufel2003,DeGrandiPolkovnikov2010}
We next depart from the adiabatic reduction and treat the nonadiabatic regime,
where $\mathcal A_{\Phi O}$ becomes a nontrivial diagnostic of coherent
state-history dependence.

To contrast with the transverse-drive model~\eqref{eq:qubit_transverse}, consider
the commuting (longitudinal) control case
\begin{equation}
\hat H(t)=\big(\Omega_z+\Phi(t)\big)\sigma_z,
\qquad
\Phi(t)=(K*u)(t).
\label{eq:qubit_longitudinal}
\end{equation}
Since $[\hat H(t),H(t')]=0$, time ordering is unnecessary and
\begin{equation}
\hat U(t)=\exp\!\big(-i\sigma_z\,\theta(t)\big),
\qquad
\theta(t)=\int_0^t \big(\Omega_z+\Phi(\tau)\big)\,d\tau.
\label{eq:U_longitudinal}
\end{equation}
This is a standard exactly solvable setting.\cite{NielsenChuang}

\paragraph{Closed-form expectations from a simple preparation.}
Prepare $\ket{+}$ (the $+1$ eigenstate of $\sigma_x$). The Bloch vector rotates
about the $z$ axis by angle $2\theta(t)$, giving
\begin{equation}
\langle\sigma_x\rangle_t = \cos\!\big(2\theta(t)\big),
\qquad
\langle\sigma_y\rangle_t = -\sin\!\big(2\theta(t)\big),
\qquad
\langle\sigma_z\rangle_t = 0.
\label{eq:longitudinal_expectations}
\end{equation}
Thus, once the classical realized signal $\Phi(t)$ is known, the unitary
response is explicit for arbitrary protocols.

\paragraph{Single-pole kernel under harmonic command.}
For the single exponential kernel and $u(t)=u_0\sin(\omega t)$, the steady-cycle
realized field is
\begin{equation}
\Phi(t)=u_0\frac{\alpha}{\sqrt{\alpha^2+\omega^2}}\sin(\omega t-\delta),
\qquad
\delta=\arctan\!\Big(\frac{\omega}{\alpha}\Big),
\label{eq:Phi_harmonic_response}
\end{equation}
so
\begin{equation}
\theta(t)=\Omega_z t
-\frac{u_0\alpha}{\omega\sqrt{\alpha^2+\omega^2}}
\Big[\cos(\omega t-\delta)-\cos\delta\Big],
\label{eq:theta_closedform}
\end{equation}
and $\langle\sigma_x\rangle_t=\cos(2\theta(t))$ follows in closed form.

The transverse model~\eqref{eq:qubit_transverse} is the relevant setting for
separating control-channel memory from coherent nonadiabaticity using loop
measures: in the adiabatic regime $\mathcal A_{\Phi O}\approx 0$, while a nonzero
$\mathcal A_{uO}$ can persist as a consequence of $\mathcal A_{u\Phi}\neq 0$.
The longitudinal model~\eqref{eq:qubit_longitudinal} provides a clean analytic
benchmark where $(u,\Phi)$ hysteresis can coexist with fully explicit unitary
dynamics.

We now turn to numerical simulations (Sec.~\ref{sec:numerics}) that interpolate
between adiabatic and nonadiabatic regimes for the transverse model while
holding the same classical control filter fixed.

We now turn to numerical simulations (Sec.~\ref{sec:numerics}) that interpolate
between adiabatic and nonadiabatic regimes for the transverse model while
holding the same classical control filter fixed.
\begin{figure*}
    \centering
    \includegraphics[width=0.99\linewidth]{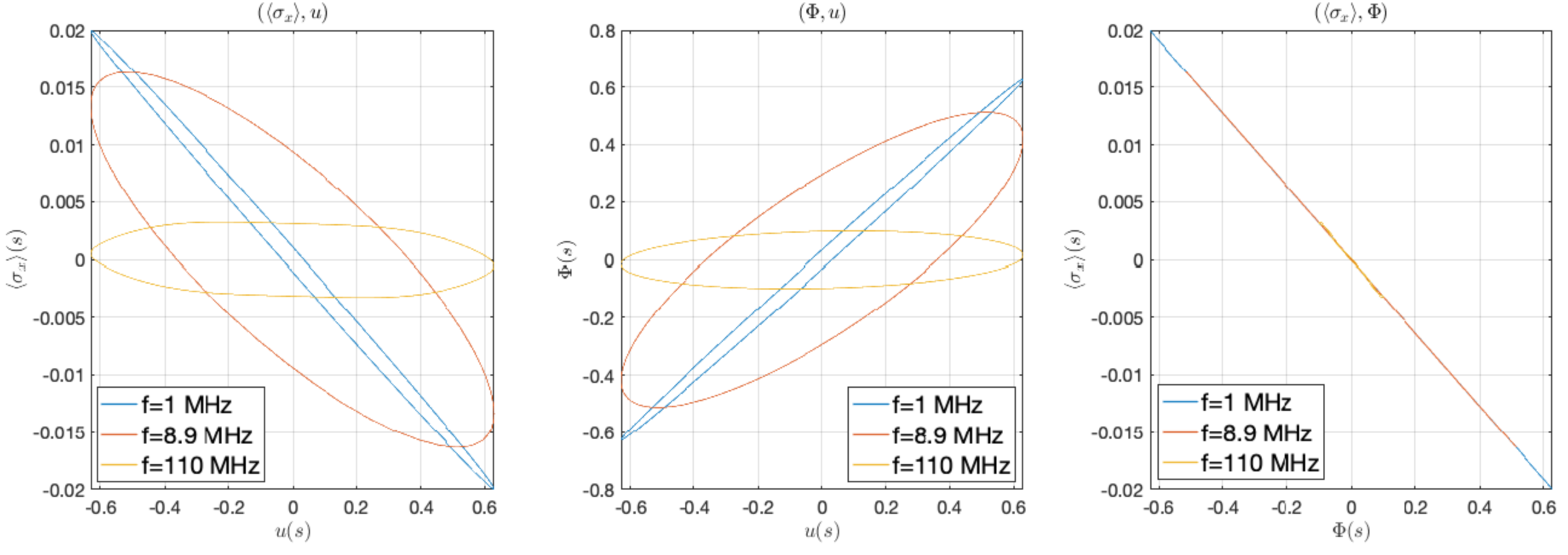}
    \caption{\textbf{Representative steady-cycle parametric loops for a filtered transverse drive.}
    A periodic command $u(t)$ (triangle protocol in the example shown) is applied for many cycles;
    after discarding transients, the final cycle is plotted parametrically as three loops:
    (left) the \emph{control-channel loop} $(u,\Phi)$ quantifying classical filtering and phase lag;
    (middle) the \emph{commanded observable loop} $(u,O)$ with $O(t)=\langle\sigma_z\rangle_t$;
    (right) the \emph{realized-drive loop} $(\Phi,O)$ which isolates state history dependence at fixed
    realized field. The $(u,\Phi)$ loop is nontrivial whenever the control channel has memory
    ($\mathcal A_{u\Phi}\neq 0$). In contrast, a nontrivial $(\Phi,O)$ loop indicates nonadiabatic
    quantum response under the realized drive; in the adiabatic-following regime it collapses toward
    a single-valued curve and $\mathcal A_{\Phi O}\approx 0$.
    }
    \label{fig:loops_threeplanes}
\end{figure*}

\begin{figure}
    \centering
    \includegraphics[width=0.99\linewidth]{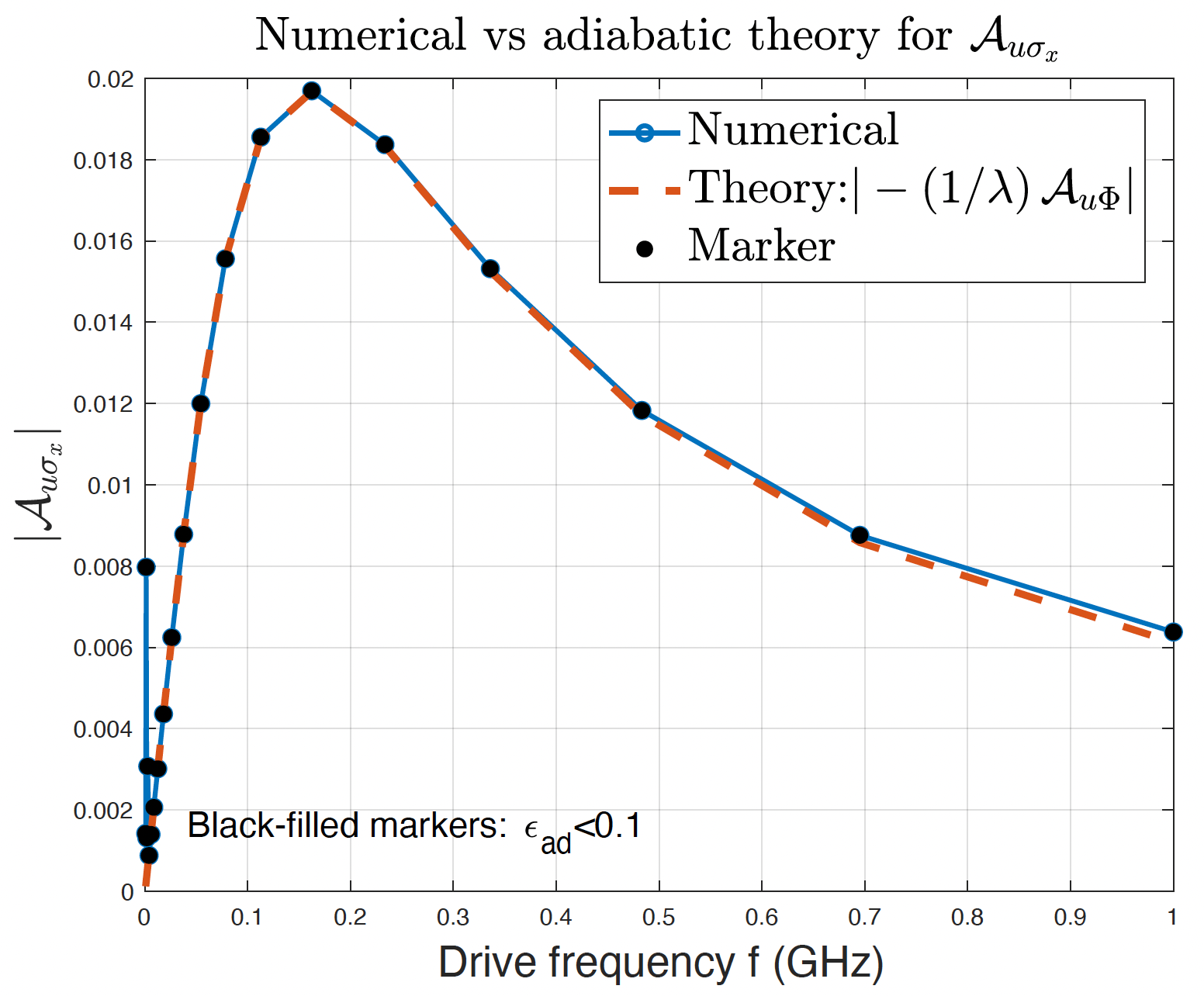}
    \caption{\textbf{Testing the adiabatic ``classical-memory imprint'' prediction.}
    The numerically extracted commanded-area $\tilde{\mathcal A}_{uO}$ (markers) is compared to the
    adiabatic small-amplitude prediction obtained from the control-channel area $\tilde{\mathcal A}_{u\Phi}$
    via $\tilde{\mathcal A}_{uO}\approx f'(0)\,\tilde{\mathcal A}_{u\Phi}$ (dashed curve), where
    $f(\Phi)$ is the adiabatic response function of the qubit observable.
    Points satisfying the adiabaticity diagnostic $\epsilon_{\rm ad}^{\max}<\epsilon_0$ are highlighted.
    Agreement holds in the regime where both (i) the realized field varies slowly compared to the instantaneous
    gap and (ii) the amplitude is sufficiently small for linearization of $f(\Phi)$ about $\Phi=0$.
    }
    \label{fig:adiabatic_area_test}
\end{figure}

\subsection{Numerical experiments}
\label{sec:numerics}

In this section we report time-domain simulations designed to separate the two
mechanisms that can generate cyclic hysteresis in commanded measurements.
The first is \emph{classical control-channel memory}: a finite-bandwidth drive
chain maps the commanded waveform $u(t)$ to a realized field $\Phi(t)$ with a
frequency-dependent lag, yielding a nonzero loop area $\mathcal A_{u\Phi}=\oint\Phi\,du$.
The second is \emph{coherent nonadiabatic response} of the qubit under the realized
drive, which can produce genuine state-history dependence at fixed $\Phi(t)$ and
hence a nonzero $\mathcal A_{\Phi O}=\oint O\,d\Phi$.
Our numerical observables are the oriented areas in the $(u,\Phi)$, $(u,O)$, and
$(\Phi,O)$ planes, evaluated after the dynamics reaches a steady cycle.

The protocol follows the standard experimental logic used in cyclic-drive studies
of superconducting qubits: a periodic control waveform is applied for many cycles,
initial transients are discarded, and the final cycle is analyzed through parametric
loops and their oriented areas.\cite{Krantz2019,Blais2021}
The distinctive modeling assumption here is that \emph{all} nonlocality is confined
to the classical drive chain, while the qubit itself evolves unitarily under the
realized field $\Phi(t)$.

We simulate a transverse-driven two-level system,
\begin{equation}
    \hat H(t)=\Omega_z\,\sigma_z+\Phi(t)\,\sigma_x,
\end{equation}
where $\Omega_z$ is the qubit splitting and $\Phi(t)$ is the realized control amplitude
delivered at the device node. The commanded waveform $u(t)$ is filtered by a causal,
dissipative control channel. For the numerical experiments we adopt the single-pole
(single-RC) embedding,
\begin{equation}
    \dot\Phi(t)=\frac{1}{\tau_c}\big(u(t)-\Phi(t)\big),
    \label{eq:numerics_rc_filter}
\end{equation}
which is the minimal time-local realization of a causal exponential kernel and makes the
control memory scale explicit (Sec.~\ref{sec:embedding_scalar}).

The qubit state is represented by its Bloch vector $r(t)\in\mathbb R^3$,
$r_j(t)=\langle\sigma_j\rangle_t$. For unitary evolution one has the precession equation
\begin{equation}
    \dot r(t)=\Omega(t)\times r(t),
    \qquad
    \Omega(t)=\big(\Phi(t),\,0,\,\Omega_z\big),
    \label{eq:numerics_bloch}
\end{equation}
which is equivalent to the Schr\"odinger equation for pure states.\cite{Abragam1961,Slichter1990}
Unless otherwise noted we initialize in the $|+\rangle$ state, $r(0)=(1,0,0)$, so that
$\langle\sigma_z\rangle$ is generated dynamically by the transverse drive.

To separate the filter timescale from qubit frequencies we work in dimensionless variables
\begin{eqnarray}
    s&=&\frac{t}{\tau_c},\qquad
    \tilde u(s)=\tau_c u(t),\nonumber \\ \tilde\Phi(s)&=&\tau_c\Phi(t),
    \qquad\lambda=\Omega_z\tau_c .
\end{eqnarray}
The coupled dynamics becomes
\begin{eqnarray}
    \frac{d\tilde\Phi}{ds}&=&-\tilde\Phi+\tilde u(s),\qquad
    \frac{dr}{ds}=\tilde\Omega(s)\times r,\nonumber \\
\tilde\Omega(s)&=&\big(\tilde\Phi(s),\,0,\,\lambda\big).
    \label{eq:numerics_dimless_odes}
\end{eqnarray}
The dimensionless drive frequency is $\tilde\omega=2\pi f\tau_c$.
In these units the control-channel response depends primarily on $\tilde\omega$,
while the degree of adiabaticity is controlled by $(\lambda,\tilde\Phi)$ through
the instantaneous gap and the realized slew rate.

We integrate \eqref{eq:numerics_dimless_odes} using a fixed-step fourth-order Runge--Kutta method
on a uniform grid in $s$, with an integer number of steps per drive period.
For each drive frequency we run many cycles, discard an initial transient window,
and compute loop areas from the final cycle. Because the filter equation is strictly stable,
this ``last-cycle'' procedure isolates the frequency-dependent steady response of the
classical channel and the coherent qubit dynamics.

We consider periodic commands $u(t)$ with zero mean, emphasizing sinusoidal and
piecewise-linear (triangle) protocols. For each run we record the commanded waveform
$u(t)$, the realized field $\Phi(t)$, and the qubit observable $O(t)$ (typically
$O(t)=\langle\sigma_z\rangle_t$). From the final cycle we compute the oriented areas
\begin{equation}
    \mathcal A_{u\Phi}=\oint \Phi\,du,\qquad
    \mathcal A_{uO}=\oint O\,du,\qquad
    \mathcal A_{\Phi O}=\oint O\,d\Phi,
\end{equation}
which quantify, respectively, hysteresis internal to the control channel, the experimentally
accessible commanded hysteresis, and the residual state-history dependence at fixed realized drive.
Numerically these are evaluated as discrete line integrals along the parametric curve using a
trapezoidal rule with explicit loop closure,
\begin{equation}
    \oint y\,dx \;\approx\; \sum_{n=1}^{N-1}\frac{y_{n+1}+y_n}{2}\,(x_{n+1}-x_n),
    \label{eq:numerics_trap_area}
\end{equation}
applied to $(x,y)=(u,\Phi)$, $(u,O)$, or $(\Phi,O)$ over one steady cycle. In the dimensionless
implementation we report $\tilde{\mathcal A}_{u\Phi}=\oint\tilde\Phi\,d\tilde u$ and
$\tilde{\mathcal A}_{uO}=\oint O\,d\tilde u$, with straightforward conversion to physical units
via $\tilde u=\tau_c u$ and $\tilde\Phi=\tau_c\Phi$.

For the transverse Hamiltonian $\hat H(t)=\Omega_z\sigma_z+\Phi(t)\sigma_x$ the instantaneous gap is
$\Delta(t)=2\sqrt{\Omega_z^2+\Phi(t)^2}$. A standard adiabatic parameter is
\begin{equation}
    \epsilon_{\rm ad}(t)    :=\frac{\big|\langle e(t)|\dot {\hat H}(t)|g(t)\rangle\big|}{\Delta(t)^2},
\end{equation}
with $\dot \hat H(t)=\dot\Phi(t)\sigma_x$.\cite{Kato1950,Teufel2003}
Evaluating the matrix element yields
\begin{equation}
    \epsilon_{\rm ad}(t)
    =
    \frac{|\dot\Phi(t)|\,\Omega_z}{4\big(\Omega_z^2+\Phi(t)^2\big)^{3/2}}.
    \label{eq:eps_ad_exact_main}
\end{equation}
In dimensionless variables, using $d\tilde\Phi/ds=-\tilde\Phi+\tilde u(s)$, this becomes
\begin{equation}
    \epsilon_{\rm ad}(s)
    =
    \frac{\big|\lambda\,d\tilde\Phi/ds\big|}{4\big(\lambda^2+\tilde\Phi(s)^2\big)^{3/2}}.
    \label{eq:eps_ad_dimless_used}
\end{equation}
We summarize adiabaticity for a given run by $\epsilon_{\rm ad}^{\max}$, the maximum over the
final cycle, and use a threshold $\epsilon_{\rm ad}^{\max}<\epsilon_0$ (with $\epsilon_0=0.1$)
to highlight the region where adiabatic-following predictions are expected to hold.

Figure~\ref{fig:loops_threeplanes} visualizes the operational separation between control-channel
memory and intrinsic state-history dependence. The $(u,\Phi)$ loop directly reports the phase lag
generated by the filter and therefore persists even when the qubit response is adiabatic.
In contrast, the $(\Phi,O)$ loop isolates history dependence of the observable at fixed realized
drive: in the adiabatic-following regime $O(t)\approx f(\Phi(t))$ is (approximately) single-valued,
so the loop collapses and $\mathcal A_{\Phi O}\approx 0$, whereas departures from this collapse
signal coherent nonadiabatic response under $\Phi(t)$.
The commanded loop $(u,O)$ combines both effects and is therefore generically nontrivial whenever
either $\mathcal A_{u\Phi}\neq 0$ or $\mathcal A_{\Phi O}\neq 0$.

A key analytic consequence of adiabatic following developed in Sec.~\ref{sec:adiabatic_regime} is
that, for small amplitudes, the commanded observable area is proportional to the control-channel area.
Figure~\ref{fig:adiabatic_area_test} compares the numerically extracted $\tilde{\mathcal A}_{uO}$
(markers) with the prediction obtained from $\tilde{\mathcal A}_{u\Phi}$ (dashed curve).
Agreement is observed precisely in the region labeled adiabatic by the independent diagnostic
$\epsilon_{\rm ad}^{\max}<\epsilon_0$ and breaks down once the realized field varies rapidly compared
to the instantaneous gap, consistent with the interpretation that deviations originate from coherent
nonadiabatic dynamics rather than from additional (open-system) memory.

To connect with a realistic control line, we take $\tau_c$ on the order of a nanosecond,
$\tau_c\sim 1~\mathrm{ns}$, representative of a well-engineered drive chain bandwidth.
The qubit splitting is chosen in the few-GHz range, e.g.\ $\Omega_z/2\pi\sim 5~\mathrm{GHz}$,
so that $\lambda=\Omega_z\tau_c=\mathcal O(10)$. Sweeping the drive frequency across several decades
therefore spans $\tilde\omega\ll 1$ (quasi-static tracking), $\tilde\omega\sim 1$ (maximal control lag),
and $\tilde\omega\gg 1$ (attenuated realized drive), enabling a direct visualization of the crossover
from classical memory-dominated hysteresis to regimes where coherent nonadiabatic response is prominent.

\begin{figure*}
    \centering
\includegraphics[width=0.9\linewidth]{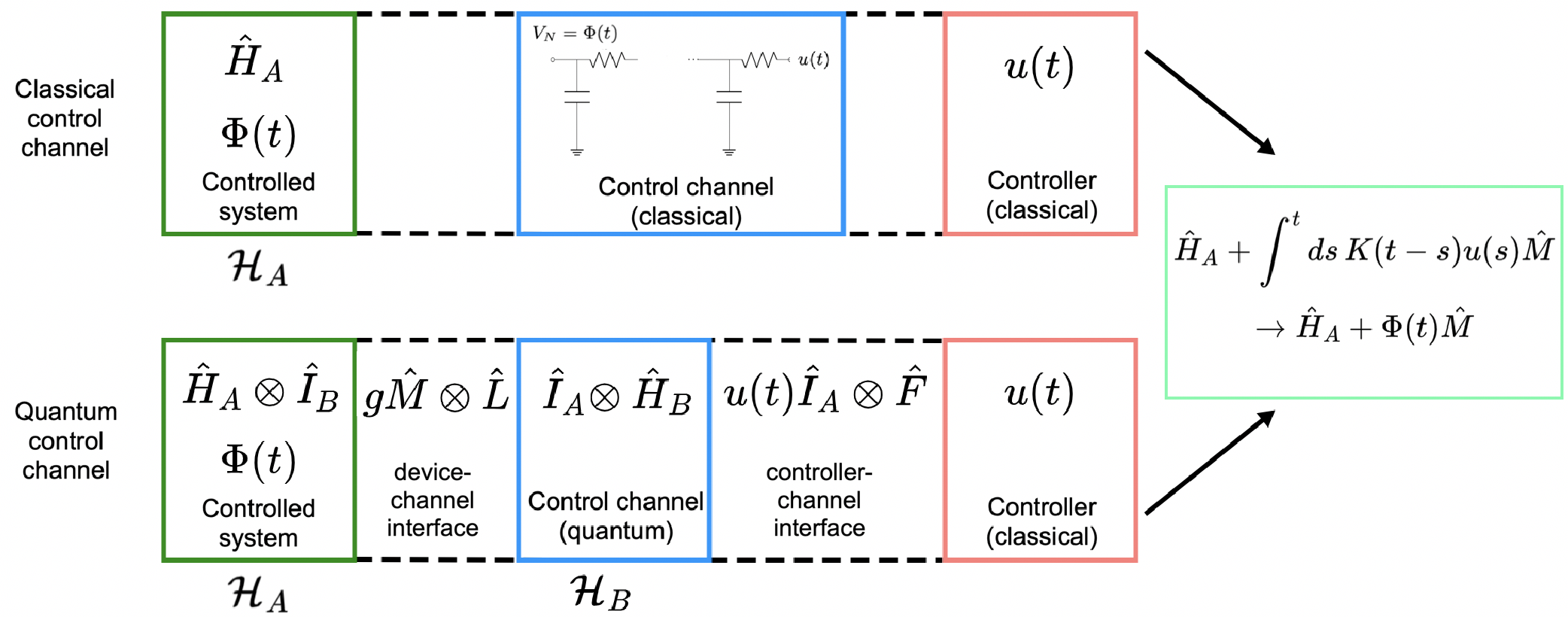}
    \caption{Schematic connection between a phenomenological \emph{classical} control filter and a microscopic \emph{quantum} control-channel model. \textbf{Top:} the room-temperature controller outputs a command $u(t)$, which is distorted by a passive control line (e.g.\ an RC ladder) into a realized in-situ field $\Phi(t)=(K*u)(t)$ that drives the device $A$ via $\hat H_A+\Phi(t)\hat M$. \textbf{Bottom:} the same architecture is modeled microscopically by a channel $B$ with Hamiltonian $\hat H_B$, driven at the input port by the classical source $u(t)$ through an operator $\hat F$ and coupled at the device port through $g\,\hat M\otimes\hat L$. In the weak-coupling regime (Born/Oppenheimer) limit and linear response (Kubo), the delivered field is the channel expectation $\Phi(t)=\langle \hat L\rangle_t=\int_{-\infty}^t \chi_{LF}(t-s)\,u(s)\,ds$, yielding the effective filtered Hamiltonian $\hat H_A+\int^t K(t-s)u(s)\,ds\,\hat M \equiv \hat H_A+\Phi(t)\hat M$ with $K$ identified as the retarded susceptibility of the control.}

    \label{fig:scheme}
\end{figure*}
\section{Physical basis for the memory in the control channel}
\label{sec:physics}
In this second part of the manuscript, we provide a more physical rationale why the kernel approach is grounded both classically and quantum mechanically, as shown in the diagram of Fig. \ref{fig:scheme}.
Specifically, this section motivates the exponential-mode control kernel introduced in
Sec.~\ref{sec:embedding_scalar} from the physics of a typical superconducting-qubit
control wiring and bias circuitry. The key point is that the waveform $u(t)$
generated by room-temperature electronics is not applied directly to the device.
Instead, it propagates through a passive, dissipative, approximately linear
network (for instance, attenuators, bias tees, wiring, packaging, and often deliberately
inserted low-pass elements), so that the device experiences a \emph{realized}
field $\Phi(t)$ that is a filtered version of the command.  This type of
finite-bandwidth and distortion modeling is standard in superconducting-qubit
hardware stacks and calibration workflows~\cite{Krantz2019,Kjaergaard2020,Roth2024}. This point of view can also be relaxed, as we show later, with the quantum treatment of the kernel $K$.

We consider, as an example, a single-qubit effective description (e.g.\ a two-level truncation of
a transmon or a flux qubit) driven by one near-resonant control quadrature~\cite{Koch2007,You2011}:
\begin{equation}
\hat H(t)=\hat H_A+\Phi(t)\hat M,
\qquad
\hat H_A=\sigma_z,
\qquad
\hat M=\sigma_x,
\label{eq:H_qubit_control_channel_section}
\end{equation}
where $\Phi(t)$ denotes the in situ drive amplitude at the device node. (For
flux-type control one may have $\hat M=\sigma_z$; nothing below depends on this
choice.) The classical filtering $u\mapsto\Phi$ is treated as an ordinary signal
processing stage. In particular, we assume it does not induce open-system
quantum memory in the sense of a non-Markovian bath: in our model the qubit
evolution remains unitary once $\Phi(t)$ is specified. This is an idealization,
but it is appropriate for isolating the role of \emph{classical} channel memory
in the hysteresis measures defined in Sec.~\ref{sec:hysteresis_definitions}.

\subsection{Classical treatment of $\Phi(t)$.} In the small-signal regime relevant for waveform transfer (and after selecting a
single effective quadrature), the control line can be modeled as a passive,
stable, linear time-invariant (LTI) two-port mapping the commanded waveform
$u(t)$ to the device-node field $\Phi(t)$. Causality motivates the convolution form
\begin{equation}
\Phi(t)=\int_{-\infty}^{t}K(t-s)\,u(s)\,ds,
\qquad
K(\tau)=0\ \ \text{for }\tau<0,
\label{eq:Phi_convolution_kernel_section}
\end{equation}
equivalently $\Phi(s)=G(s)u(s)$ in the Laplace domain, with transfer function
$G(s):=\Phi(s)/u(s)$ analytic for $\Re(s)>0$.

The distinctive feature of the experimental setting is that much of the relevant
filtering is well described by \emph{relaxational} (RC-type) dynamics rather than
high-$Q$ resonances: thermalized attenuators and resistive wiring introduce loss,
while distributed capacitances to ground (and intentional shunt capacitances in
bias tees and low-pass sections) provide storage. In such circumstances an RC
network is a natural effective description at the frequencies that shape the
envelope of typical control pulses.

Classical network-synthesis theory then severely restricts the analytic
structure of $G(s)$. For passive RC networks the transfer function is real,
stable, and (for ladder families of interest) has poles on the negative real
axis, generically simple; realizability/synthesis criteria and proofs for RC
ladders are classical~\cite{FialkowGerst1951,Magos1970}. Appendix~\ref{app:rc_realizability}
summarizes the relevant statements and derives the exponential-mode
representation directly for an RC ladder model.

These pole restrictions translate immediately into an exponential-mode impulse
response. If $G(s)$ is rational (finite ladder approximation), one can write
\begin{equation}
G(s)=g_\infty+\sum_{k=1}^{K}\frac{c_k}{s+\nu_k},
\qquad
\nu_k>0,
\label{eq:G_partial_fraction_section}
\end{equation}
so that the causal kernel is
\begin{equation}
K(\tau)=g_\infty\,\delta(\tau)+\sum_{k=1}^{K}c_k e^{-\nu_k\tau}\Theta(\tau).
\label{eq:K_sum_exp_section}
\end{equation}
For a distributed (diffusive) RC line, $G(s)$ is no longer rational, but it
remains passive and stable and admits an integral mixture
of exponentials; Appendix~\ref{app:rc_realizability} gives an explicit example
and the corresponding long-tailed kernel. In either case, the essential
consequence is that $K(\tau)$ is non-oscillatory and generated by a hierarchy of
relaxation rates $\{\nu_k\}$.

The exponential representation \eqref{eq:K_sum_exp_section} is equivalent to the
finite-dimensional time-local embedding used in Sec.~\ref{sec:embedding_scalar}.
Defining auxiliary modes
\begin{eqnarray}
\Phi_k(t):&=&\int_{-\infty}^{t}c_k e^{-\nu_k(t-s)}u(s)\,ds,\nonumber \\
\Phi(t)&=&g_\infty u(t)+\sum_{k=1}^{K}\Phi_k(t),
\label{eq:Phi_modes_def_section}
\end{eqnarray}
one obtains
\begin{equation}
\dot\Phi_k(t)=-\nu_k\Phi_k(t)+c_k u(t),
\qquad
k=1,\dots,K.
\label{eq:Phi_modes_ODE_section}
\end{equation}
In the ODE and time-local representation, the initial conditions are associated to the initial state of the shunted capacitances.
For a finite RC ladder this representation is not an approximation: it is simply
the modal form of the ladder state-space dynamics (Appendix~\ref{app:rc_realizability}).
Physically, the modes correspond to internal charge-storage degrees of freedom
of the line/bias network, and the slowest rate $\nu_{\min}$ sets the dominant
memory time $\tau_{\rm mem}\sim \nu_{\min}^{-1}$.

The simplest nontrivial instance is a single series resistance $R$ feeding a
device-side shunt capacitance $C$. Writing the device-node voltage as $\Phi(t)$,
Kirchhoff’s laws give
\begin{eqnarray}
&&C\,\dot\Phi(t)=\frac{u(t)-\Phi(t)}{R}
\hspace{0.1cm}\Longleftrightarrow\hspace{0.1cm}
\tau_c\dot\Phi(t)+\Phi(t)=u(t),
\nonumber \\
&&\tau_c:=RC,
\label{eq:single_RC_ODE_section}
\end{eqnarray}
with kernel
\begin{equation}
K(\tau)=\frac{1}{\tau_c}e^{-\tau/\tau_c}\Theta(\tau).
\label{eq:single_RC_kernel_section}
\end{equation}
This is precisely the single-pole truncation of \eqref{eq:K_sum_exp_section},
corresponding to retaining only the slowest relaxational channel. In the limit
$\tau_c\to 0$, the kernel concentrates at $\tau=0$ and $\Phi(t)\to u(t)$ for
protocols that do not vary on vanishing time scales, recovering the memoryless
control idealization.

The exponential-mode description is particularly convenient for the hysteresis
framework developed earlier in the manuscript. It provides an explicit classical memory channel
with a controlled hierarchy of time scales, and also a clean separation between
hysteresis inherited from the filter (visible in $\mathcal A_{u\Phi}$ and thus
potentially in $\mathcal A_{uO}$) and genuinely nonadiabatic/state-history
effects in the quantum dynamics (diagnosed by $\mathcal A_{\Phi O}$). This
separation was exploited in Secs.~\ref{subsec:adiabatic_qubit_exp_kernel}--%
\ref{subsec:qubit_transverse_adiabatic_expanded} and in the numerical study of
Sec.~\ref{sec:numerics}.

\subsection{Quantum Kubo treatment of $\Phi(t)$.} Throughout the paper we modelld the \emph{realized} control field $\Phi(t)$ seen by the
device as a causal convolution of the room-temperature command $u(t)$,
$\Phi(t)=(K*u)(t)$.  This might seem unjustified from a quantum perspective. Appendix~\ref{app:kubo_born_filtered_control} provides a
microscopic quantum derivation of this structure and clarifies when the kernel $K$
admits the exponential-mode forms used in Sec.~\ref{sec:embedding_scalar} and in the RC
examples (Appendix~\ref{app:rc_realizability}). Here we summarize the key steps and
consequences.
First, we treat the control hardware (wiring, attenuators, bias tees, filters, packaging)
as a \emph{physical dynamical system}, in the most microscopic description, a quantum
many-body ``channel'' $B$ with Hilbert space $\mathcal H_B$ and Hamiltonian $\hat H_B$.
The room-temperature electronics provide a prescribed waveform $u(t)$, which we model
as a classical (c-number) \emph{generalized force} applied at an \emph{input port} of the
channel. Microscopically, such forcing is represented by a Hamiltonian term linear in a
channel operator $\hat F$, i.e.\ a work term $-u(t)\hat F$, exactly as in the standard
Kubo formulation of linear response and the fluctuation--dissipation framework
\cite{Kubo1957Irreversible1,Kubo1966FDT,CallenWelton1951}. The device $A$ couples to the
channel at an \emph{output port} through a (generally distinct) channel operator $\hat L$,
which represents the local field/coordinate delivered to the device node. This separation
between input forcing and output readout is also the natural language of quantum network
and input--output theory for driven electromagnetic environments
\cite{YurkeDenker1984,GardinerCollett1985,Clerk2010QuantumNoise,Devoret1997LesHouches,VoolDevoret2017}.

We aim to show that a minimal bipartite Hamiltonian capturing these roles is
\begin{equation}
    \hat H(t)
    =
    \hat H_A \otimes \mathbb I_B
    +\mathbb I_A\otimes\hat H_B
    + g\,\hat M\otimes \hat L
    - u(t)\,\mathbb I_A\otimes \hat F,
    \label{eq:microscopic_H_recap}
\end{equation}
where $\hat M$ is the device operator that couples to the delivered field, and $g$ is a (an assumed)
small coupling parameter quantifying how strongly the device perturbs the control
channel. The structure \eqref{eq:microscopic_H_recap} is the direct analogue of standard
system--bath models in which a distinguished system couples to a large environment via a
bilinear interaction, with the environment itself subject to external driving \cite{
GardinerCollett1985InputOutput,
CollettGardiner1984SqueezingInputOutput,
Clerk2010QuantumNoiseRMP}. In
particular, eliminating (or coarse-graining) channel degrees of freedom in such models
generically produces retarded response functions and memory kernels, as in the classic
derivations of generalized Langevin dynamics for assemblies of coupled oscillators
\cite{FordKacMazur1965} and in the Caldeira--Leggett description of dissipative quantum
environments \cite{CaldeiraLeggett1981}. In our setting, however, we emphasize that it is related to the transmission conductance, closer to Kubo's \cite{Kubo1957Irreversible1}
\emph{two-port} viewpoint, i.e. $\hat F$ encodes how the commanded signal injects energy into
the control, while $\hat L$ encodes the field that is actually delivered at the
device node and therefore enters the effective device Hamiltonian.

This microscopic setup is precisely tailored to the kernel description used in the main
text: in linear response, the realized field $\Phi(t):=\langle \hat L(t)\rangle$ is a
causal convolution of the command $u(t)$ with the retarded susceptibility
$\chi_{LF}(\tau)$ \cite{Kubo1957Irreversible1,Kubo1966FDT}. Moreover, when $\hat F$ and
$\hat L$ are chosen as standard circuit port variables (e.g.\ charge/flux at a port and
current/voltage at a readout node), $\chi_{LF}(\omega)$ coincides with familiar circuit
response functions such as admittances or transfer admittances, and its dissipative
component is constrained by passivity and linked to equilibrium noise via
fluctuation--dissipation \cite{CallenWelton1951,Clerk2010QuantumNoise,Devoret1997LesHouches,VoolDevoret2017}.

As we show in the Appendix, the realized classical field is identified with the channel mean at the output port,
\begin{equation}
    \Phi(t):=\langle \hat L\rangle_t=\mathrm{Tr}_B\!\big(\hat\rho_B(t)\hat L\big),
    \label{eq:Phi_def_maintext}
\end{equation}
so $\Phi(t)$ is a c-number functional of the applied waveform $u(\cdot)$.
Differentiating $\hat\rho_A(t)=\mathrm{Tr}_B\hat\rho(t)$ yields an exact identity in which the
device couples to the operator-valued field
$\mathrm{Tr}_B[(\mathbb I_A\otimes\hat L)\hat\rho(t)]$
(Appendix~\ref{app:kubo_born_filtered_control}). In the weak-coupling regime one may neglect
device-induced disturbance of the channel to leading order, so that
$\hat\rho(t)\approx \hat\rho_A(t)\otimes \hat\rho_B^{(u)}(t)$ and the field factorizes as
$\mathrm{Tr}_B[(\mathbb I_A\otimes\hat L)\hat\rho(t)]\approx \Phi(t)\,\hat\rho_A(t)$.
This yields the effective reduced generator
\begin{equation}
    \dot{\hat\rho}_A(t)
    =
    -i\big[\hat H_A + g\,\Phi(t)\,\hat M,\ \hat\rho_A(t)\big]
    +\mathcal O(g^2),
    \label{eq:Heff_filtered_maintext}
\end{equation}
i.e.\ the device is driven \emph{coherently} by the classical realized field $\Phi(t)$,
while channel fluctuations and backaction enter only at higher order.

For a stationary reference state of the unforced channel, the response of $\Phi(t)$ to a weak
source $u(t)$ is given by a Kubo retarded susceptibility. In the closed (unitary) case,
\begin{eqnarray}
    \Phi(t)&=&\Phi_0+\int_{-\infty}^{t}\chi_{LF}(t-s)\,u(s)\,ds+\mathcal O(u^2),\nonumber \text{where}\\
    \chi_{LF}(\tau)&=& i\,\Theta(\tau)\,\langle[\hat L_I(\tau),\hat F_I(0)]\rangle_B,
    \label{eq:Phi_Kubo_maintext}
\end{eqnarray}
so the kernel used in the main text is a Kubo-like causal response function $K(\tau)\equiv \chi_{LF}(\tau)$ to leading order
(Appendix~\ref{app:kubo_born_filtered_control}), between the input and output port operators on $\mathcal H_B$. This identifies the classical filter as a
two-port quantum response function of the control.
We see that in this regime (and up to order $g^2)$, the evolution of the reduced density matrix is consistent with the time evolution of eqn.  (\ref{eq:Hkernel_general_rewrite}). Implicitly, this structure was also partly suggested in \cite{barrows} (see App. D) where, however, it was not identified as a convolution kernel \textit{a l\'a} Kubo as in here.

The structure of the kernel function $K$ depends on the physical assumptions.
First, we show that the time-domain structure of
$\chi_{LF}(\tau)$ depends qualitatively on whether the channel is lossless or dissipative.
For a finite-dimensional unitary channel, $\chi_{LF}(\tau)$ admits a Lehmann expansion as a
finite sum of oscillatory terms $e^{i\omega_{nm}\tau}$ and does not generically decay.
Thus, without dissipation (or a macroscopic continuum limit), one should not expect a
relaxational sum of decaying exponentials.
However, if the channel is an open system with a stable GKLS generator $\mathcal L_0$ (dissipation
\emph{on $B$}), the Kubo formula is obtained by replacing unitary evolution with the adjoint
semigroup $e^{\mathcal L_0^\dag\tau}$. This yields a mode expansion
$\chi_{LF}(\tau)=\Theta(\tau)\sum_k c_k e^{\lambda_k\tau}$ with $\mathrm{Re}(\lambda_k)\le 0$.
In the overdamped (RC-like) regime relevant for passive wiring and filters, the dominant
$\lambda_k$ are real and negative, $\lambda_k=-\nu_k$, giving
\begin{equation}
    K(\tau)\equiv \chi_{LF}(\tau)\;\approx\;\sum_{k=1}^{K_{\max}} c_k e^{-\nu_k\tau}\,\Theta(\tau),
    \qquad \nu_k>0,
    \label{eq:K_exponential_modes_maintext}
\end{equation}
which is consistent with the exponential-mode embedding used in Sec.~\ref{sec:embedding_scalar} and earlier in this section.
Thus, the same exponential structure follows from classical network synthesis: passive RC ladders have
transfer functions with poles on the negative real axis, implying impulse responses that are sums
(or mixtures) of decaying exponentials. Appendix~\ref{app:rc_realizability} makes this connection
explicit and provides the circuit-level interpretation of the rates $\nu_k$ as relaxation modes of
the line. 

A standard subtlety in Kubo linear response is that the assumption that the channel is thermal in the remote past
is inconsistent with applying a finite perturbation for all times: if the source $u(t)$ were nonzero forever,
the state at $t\to-\infty$ would not be an equilibrium state of $\hat H_B$.
The conventional resolution is to assume that the perturbation is switched on adiabatically from the far past,
either by sending the initial time $t_0\to-\infty$ together with a convergence factor, or equivalently by
replacing the perturbation by $-u(t)e^{\eta t}\hat F$ and taking $\eta\downarrow 0^+$ at the end
(the ``adiabatic switching'' or ``$i0^+$ retarded prescription'').\cite{Kubo1957Irreversible1,Mahan2000,Forster1975,FetterWalecka1971}
Concretely, in the derivation of Appendix~\ref{app:kubo_phi_derivation} one may keep the factor explicitly:
\begin{eqnarray}
\Phi(t)-\Phi_0
&=&
\int_{-\infty}^{t} ds\;\chi_{LF}^{(\eta)}(t-s)\,u(s),
\nonumber \\
\chi_{LF}^{(\eta)}(\tau)
&=&
i\,\Theta(\tau)\,e^{-\eta \tau}\,\langle[\hat L_I(\tau),\hat F]\rangle_B,
\qquad \eta>0,\nonumber \\
\end{eqnarray}
so that $\chi_{LF}(\tau)=\lim_{\eta\downarrow 0^+}\chi_{LF}^{(\eta)}(\tau)$ in the usual retarded sense.
In general, the representation above is well known to not lead to an RC-like kernel as it is quasi-periodic (see for instance App. \ref{app:closed_vs_open_kernel}).
In the fully closed case, the strictly closed thermal Lehmann sum representation for the kernel
\begin{equation}
    \chi_{LF}(\tau)
    =
    i\,\Theta(\tau)\sum_{m,n}(p_n-p_m)\,L_{nm}F_{mn}\,e^{i(E_n-E_m)\tau}.
    \label{eq:chi_lehmann_closed_main}
\end{equation}
is quasi-periodic in time and,
for a finite-dimensional $B$, does not generically decay as $\tau\to\infty$ (Appendix~\ref{app:thermal_lehmann_kernel}).

In frequency space, the same $\eta$ appears as the infinitesimal imaginary part selecting the retarded
analytic continuation:
\begin{eqnarray}
\chi_{LF}^R(\omega)
&=&
-i\int_0^\infty d\tau\;e^{i(\omega+i0^+)\tau}\,\langle[\hat L_I(\tau),\hat F]\rangle_B, \nonumber \\
\chi_{LF}^{(\eta)}(\omega)
&=&
-i\int_0^\infty d\tau\;e^{i(\omega+i\eta)\tau}\,\langle[\hat L_I(\tau),\hat F]\rangle_B.
\end{eqnarray}

Equivalently, in the Lehmann representation (Appendix~\ref{app:thermal_lehmann_kernel}), the regulator
shifts poles into the lower half-plane,
\begin{eqnarray}
\chi_{LF}^{(\eta)}(\omega)
&=&
\sum_{m,n}(p_n-p_m)\,L_{nm}F_{mn}\;
\frac{1}{\omega-\omega_{nm}+i\eta},\nonumber \\
&&\omega_{nm}:=E_n-E_m,
\label{eq:chi_lehmann_eta_maintext}
\end{eqnarray}
so that the spectral lines which are $\delta$-sharp in a strictly closed, finite-dimensional channel
are broadened into Lorentzians of width $\eta$:
$$
\Im\,\chi_{LF}^{(\eta)}(\omega)
=
-\sum_{m,n}(p_n-p_m)\,L_{nm}F_{mn}\;
\frac{\eta}{(\omega-\omega_{nm})^2+\eta^2}.
$$
From this viewpoint, the Kubo's $\eta$ is both (i) a mathematical convergence device that makes
$t_0\to-\infty$ well-defined and enforces causality, and (ii) a stand-in for physical level broadening
associated with finite lifetimes once the channel is coupled (even weakly) to uncontrolled degrees of freedom.
The same object also appears when one derives real-time response from imaginary-time (Matsubara) correlators:
one computes $\chi(i\omega_n)$ at discrete Matsubara frequencies $\omega_n$ and then analytically continues
$i\omega_n\mapsto \omega+i0^+$ to obtain the retarded $\chi^R(\omega)$.\cite{Mahan2000,FetterWalecka1971,Forster1975,Kamenev2011}

Introducing $\eta>0$ multiplies the time-domain commutator by $e^{-\eta\tau}$, producing an \emph{effective}
exponential envelope. However, a single constant $\eta$ only yields a single exponential scale and should be read
as a coarse phenomenology for dissipation.

A more physical route to genuine multi-timescale relaxational decay is precisely the open-channel description:
dissipation shifts the relevant poles by \emph{mode-dependent} decay rates. In the Markovian setting discussed above,
the adjoint generator $\mathcal L_0^\dag$ has eigenvalues $\lambda_k=-\nu_k+i\omega_k$ with $\nu_k\ge 0$, so that
the Kubo-like commutator form \eqref{eq:chi_open_commutator} produces
$$
\chi_{LF}(\tau)=\Theta(\tau)\sum_k c_k e^{\lambda_k \tau},
$$
and in the overdamped regime (relevant for passive RC stacks in-band) one has $\omega_k\approx 0$ and
$\chi_{LF}(\tau)\approx \Theta(\tau)\sum_k c_k e^{-\nu_k\tau}$, matching the RC-ladder impulse response.
In this sense, the classical RC picture and the ``$\eta$-broadened Kubo'' picture are two complementary
ways to encode the same physical idea: losses in the control stack move the response poles off the real axis,
broadening discrete spectral lines into continuous, dissipative features and producing decaying memory kernels.
This is exactly the regime in which the c-number filter model for $\Phi(t)$ is justified and in which
$\mathcal A_{u\Phi}$ isolates the out-of-phase (dissipative) quadrature via \eqref{eq:AuPhi_ImG_general},
thereby connecting our geometric hysteresis measures to standard conductance-like Kubo response functions
and their fluctuation--dissipation constraints (Appendix~\ref{app:thermal_lehmann_kernel}).\cite{Kubo1957Irreversible1,Forster1975,Mahan2000,BreuerPetruccione2002}

Taken together, these results justify modeling the control channel as a classical
memory kernel $K$ (generated either by an RC ladder or by an overdamped open-channel response),
while keeping the device dynamics unitary to leading order as in \eqref{eq:Heff_filtered_maintext}. 
Of course, a more careful treatment of the channel's properties is paramount to model the memory appropriately. Here, we simply focus on the simplest case in which the channel is stationary and thermal.
\section{Discussion and conclusions}
\label{sec:conclusions}

This work isolates a pragmatic and experimentally relevant source of apparent
``memory'' in driven quantum platforms: the field that enters the device
Hamiltonian is often not the commanded waveform $u(t)$, but a \emph{filtered}
realized signal $\Phi(t)$ generated by the control stack. The filtering is
classical, it reflects finite bandwidth, internal relaxation, and distortion of
the control path, yet it can generate pronounced hysteresis in drive-response
plots and can strongly shape the ``work-like'' cyclic integrals considered in
this paper.

The paper has two complementary parts. The first part develops a compact
\emph{hysteresis language} for separating control-channel memory from intrinsic
device dynamics using loop measures in the $(u,\Phi)$, $(u,O)$, and
$(\Phi,O)$ planes. The second part provides a \emph{modeling basis} for the
control channel itself, connecting the kernel picture $\Phi=(K*u)$ to passive
RC-type circuitry (and, more generally, to retarded susceptibilities of a driven
control), thereby justifying the exponential-mode embeddings used in the
main text.

The loop areas introduced in Sec.~\ref{sec:hysteresis_definitions} provide a
minimal operational vocabulary.
The control-channel area $\mathcal A_{u\Phi}=\oint \Phi\,du$ quantifies hysteresis
internal to the control path. The realized-drive area
$\mathcal A_{\Phi O}=\oint O\,d\Phi$ isolates history dependence of the observable
at fixed realized drive. Finally, the commanded area
$\mathcal A_{uO}=\oint O\,du$ is what is typically observed experimentally when
one plots an observable against the programmed waveform, and it conflates
classical control memory with device response.

A central message is that these loop measures, by themselves, do \emph{not}
certify quantum memory in the open-system sense (information backflow). In the
kernel-filtered model studied here, the only nonlocality resides in the
\emph{classical} map $u(\cdot)\mapsto \Phi(\cdot)$. Once the realized trajectory
$\Phi(t)$ is specified, the device evolves unitarily on $\mathcal H_A$. As a
consequence, for any two initial device states $\rho_1(0)$ and $\rho_2(0)$
evolved under the same Hamiltonian trajectory, the trace distance
\[
D(t)=\tfrac12\|\rho_1(t)-\rho_2(t)\|_1
\]
cannot exhibit revivals, since $\|\cdot\|_1$ is invariant under unitary
conjugation. \cite{BreuerLainePiilo2009,RivasHuelgaPlenio2014,deVegaAlonso2017}
This separation is experimentally actionable: if one reconstructs (or otherwise
controls) the realized $\Phi(t)$ and still observes trace-distance revivals for
some pair of preparations, then the reduced dynamics cannot be explained as
a classical filter and must involve additional
degrees of freedom.\cite{BreuerLainePiilo2009}

On the modeling side, the single-pole (single-RC) kernel provides a convenient
first parametrization of control memory: it is the minimal causal model that
turns the nonlocal relation $u(\cdot)\mapsto \Phi(\cdot)$ into a single auxiliary
state variable and introduces one time scale $\tau_c$ that can be compared
across devices and wiring stacks. Importantly, $\tau_c$ need not be introduced
as an abstract fit parameter. As summarized in
Appendix~\ref{app:bandwidth_tau_risetime}, when the small-signal transfer function
from $u$ to $\Phi$ is well approximated by a first-order low-pass response,
the cutoff frequency $f_{\mathrm{3dB}}$ and step-response rise time $t_r$
translate directly into $\tau_c$ via
$\tau_c \approx 1/2\pi f_{\mathrm{3dB}}$, leading to
$t_r \approx 0.3/B_W$,
with $B_W$ the (approximately) single-pole bandwidth in hertz. This connects
quantities routinely measured in situ (bandwidth, step response, impulse
response) to the effective memory scale used in the kernel phenomenology. This was also discussed in \cite{BornemanCory2012BandwidthLimited}, in which, for a bandwidth in MHz, they estimated a kernel memory of the order of $\mu$s. 

As a rough estimate with the modern control stack, an envelope bandwidth in the $0.1$--$1\,\mathrm{GHz}$
range corresponds to $\tau_c\sim 0.16$--$1.6\,\mathrm{ns}$, placing the most
pronounced kernel-induced hysteresis in the regime $\omega\tau_c\sim 1$, i.e.\
modulation frequencies from hundreds of MHz to a few GHz. Depending on the bandwidth of the control stack, the memory can be as small as $\approx 1 ns$ to a few $\mu$s. 

Second, the kernel model suggests a direct route to improved pulse design. If
the control path is parameterized by a small set of modes (single RC or a
multi-mode realization), one can \emph{embed} the filter dynamics as auxiliary
classical state variables and optimize pulses against the \emph{realized}
Hamiltonian rather than the commanded waveform. This viewpoint aligns with
distortion-aware optimal-control approaches in which hardware response is
explicitly incorporated into the forward model used for synthesis and
calibration.\cite{PhysRevApplied.4.024012,BornemanCory2012BandwidthLimited,Wittler2021C3Toolset,Glaser2015}
In gradient-based optimal control, this amounts to augmenting the forward model
by the filter ODE(s) and differentiating through them, in the spirit of
established quantum-control frameworks.\cite{ViolaKnillLloyd1999,Glaser2015}

Third, the kernel at the device is the result of of various filters of the stack, including both a classical and a quantum induced kernel, as in Fig. \ref{fig:classquant}. The action of these could be at different timescales. This implies that the filter on the $A$ system is the result of both a quantum $K_q$ and classical $K_c$ filter, e.g. $\mathcal K_{A}=K_{q}*K_{c}$, or alternatively
$\Phi(t)=(K_{q}*(K_c*u))(t)$, and thus the stages are modeled using different methods.  In practice, realistic wiring stacks
are rarely captured by a single pole: connectors, bias tees, impedance
mismatches, and distributed dissipation can produce multiple decays and long
tails. We provided a principled justification
for multi-exponential kernels based on passive RC ladder realizations and for
viewing the single-RC model as the leading (slowest-pole) truncation, while we also provided a microscopic justification in which the kernel is identified
with a retarded susceptibility of a driven control in linear response.
\begin{figure}
    \centering
\includegraphics[width=0.99\linewidth]{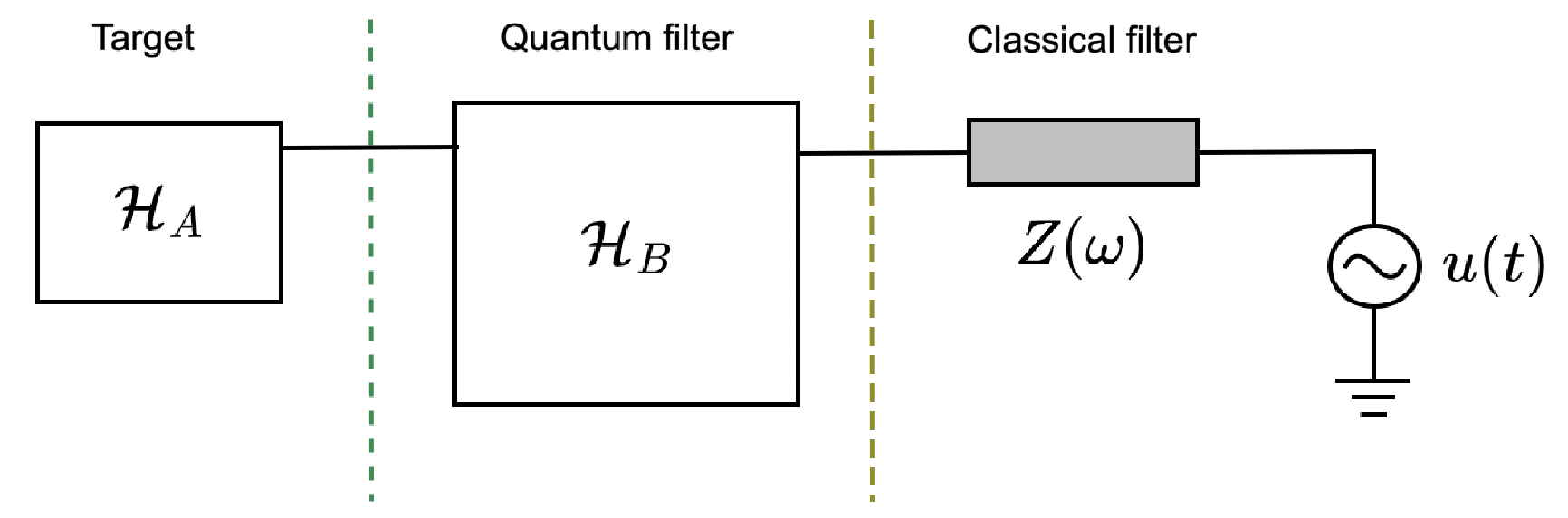}
    \caption{In practice, the filter at the target device is the result of various stages of the stack. }
    \label{fig:classquant}
\end{figure}

Two practical consequences follow.
First, hysteresis loops in the commanded plane $(u,O)$ should not be
over-interpreted: even for a closed, unitary device, classical filtering can
generate sizable loop areas through the lag $u\mapsto\Phi$, and coherent
nonadiabaticity can generate $\mathcal A_{\Phi O}\neq 0$ without invoking any
open-system memory. This motivates reporting control-channel diagnostics (at
minimum $\mathcal A_{u\Phi}$, or an equivalent transfer-function estimate)
alongside observable hysteresis plots.

Several extensions are natural. On the diagnostics side, the most informative
next step is a systematic numerical sweep of frequency and amplitude across the
boundary between adiabatic following and strongly nonadiabatic response, while
holding the \emph{control} transfer function fixed \cite{viola0,viola1,viola2,mont0,mont1,mont2,mont3}. This would map, within a
unified set of plots, how classical control memory (through $K$) and coherent
dynamics (through $H_A$) jointly determine $\mathcal A_{u\Phi}$,
$\mathcal A_{uO}$, and $\mathcal A_{\Phi O}$. On the modeling side, pushing
the microscopic derivation beyond leading order would clarify how channel
fluctuations and backaction modify the effective drive and when truly open-system
signatures become unavoidable. We expect the combined hysteresis diagnostics and
channel modeling developed here to be useful both for interpreting hysteretic
features in experimental data and for designing control strategies that
explicitly separate, diagnose, and mitigate classical control memory.

\begin{acknowledgements}
The author is indebted to F. Barrows, C. Nisoli, P. Sathe, A. Adebi, M. Polini, and G. Menichetti for comments regarding this manuscript.
This work was partly supported by
the U.S. Department of Energy through the Los Alamos National Laboratory. Los Alamos National Laboratory is
operated by Triad National Security, LLC, for the National Nuclear Security Administration of U.S. Department of
Energy (Contract No. 89233218CNA000001). The research presented in this article was supported by the Laboratory
Directed Research and Development program of Los Alamos National Laboratory under project number 20240032DR. The author is currently an employee of Planckian.
\end{acknowledgements}
\bibliographystyle{unsrt}
\bibliography{refs}

@article{ViolaKnillLloyd1999,
  title        = {Dynamical decoupling of open quantum systems},
  author       = {Viola, Lorenza and Knill, Emanuel and Lloyd, Seth},
  journal      = {Physical Review Letters},
  volume       = {82},
  number       = {12},
  pages        = {2417--2421},
  year         = {1999},
  doi          = {10.1103/PhysRevLett.82.2417}
}

@article{Glaser2015,
  title        = {Training Schr{\"o}dinger's cat: quantum optimal control},
  author       = {Glaser, Steffen J. and Boscain, Ugo and Calarco, Tommaso
                  and Koch, Christiane P. and K{\"o}ckenberger, Walter
                  and Kosloff, Ronnie and Kuprov, Ilya and Luy, Burkhard
                  and Schirmer, Sophie and Schulte-Herbr{\"u}ggen, Thomas
                  and Sugny, Dominique and Wilhelm, Frank K.},
  journal      = {The European Physical Journal D},
  volume       = {69},
  number       = {12},
  pages        = {279},
  year         = {2015},
  doi          = {10.1140/epjd/e2015-60464-1}
}

@article{Temme2017,
  title        = {Error mitigation for short-depth quantum circuits},
  author       = {Temme, Kristan and Bravyi, Sergey and Gambetta, Jay M.},
  journal      = {Physical Review Letters},
  volume       = {119},
  number       = {18},
  pages        = {180509},
  year         = {2017},
  doi          = {10.1103/PhysRevLett.119.180509}
}

@article{Krantz2019,
  title        = {A quantum engineer's guide to superconducting qubits},
  author       = {Krantz, Philip and Kjaergaard, Morten and Yan, Fei
                  and Orlando, Terry P. and Gustavsson, Simon
                  and Oliver, William D.},
  journal      = {Applied Physics Reviews},
  volume       = {6},
  number       = {2},
  pages        = {021318},
  year         = {2019},
  doi          = {10.1063/1.5089550}
}

@article{GardinerCollett1985InputOutput,
  title = {Input and output in damped quantum systems: Quantum stochastic differential equations and the master equation},
  author = {Gardiner, C. W. and Collett, M. J.},
  journal = {Physical Review A},
  volume = {31},
  pages = {3761--3774},
  year = {1985},
  doi = {10.1103/PhysRevA.31.3761}
}

@article{CollettGardiner1984SqueezingInputOutput,
  title = {Squeezing of intracavity and traveling-wave light fields produced in parametric amplification},
  author = {Collett, M. J. and Gardiner, C. W.},
  journal = {Physical Review A},
  volume = {30},
  pages = {1386--1391},
  year = {1984},
  doi = {10.1103/PhysRevA.30.1386}
}

@article{Clerk2010QuantumNoiseRMP,
  title = {Introduction to quantum noise, measurement, and amplification},
  author = {Clerk, A. A. and Devoret, M. H. and Girvin, S. M. and Marquardt, F. and Schoelkopf, R. J.},
  journal = {Reviews of Modern Physics},
  volume = {82},
  pages = {1155--1208},
  year = {2010},
  doi = {10.1103/RevModPhys.82.1155}
}

@article{Motzoi2009,
  title        = {Simple pulses for elimination of leakage in weakly anharmonic qubits},
  author       = {Motzoi, Felix and Gambetta, Jay M. and Rebentrost, Patrick
                  and Wilhelm, Frank K.},
  journal      = {Physical Review Letters},
  volume       = {103},
  number       = {11},
  pages        = {110501},
  year         = {2009},
  doi          = {10.1103/PhysRevLett.103.110501}
}

@article{BreuerLainePiilo2009,
  title        = {Measure for the Degree of Non-Markovian Behavior of Quantum Processes in Open Systems},
  author       = {Breuer, Heinz-Peter and Laine, E.-M. and Piilo, J.},
  journal      = {Physical Review Letters},
  volume       = {103},
  number       = {21},
  pages        = {210401},
  year         = {2009},
  doi          = {10.1103/PhysRevLett.103.210401}
}

@article{RivasHuelgaPlenio2014,
  title        = {Quantum Non-Markovianity: Characterization, Quantification and Detection},
  author       = {Rivas, {\'A}ngel and Huelga, Susana F. and Plenio, Martin B.},
  journal      = {Reports on Progress in Physics},
  volume       = {77},
  number       = {9},
  pages        = {094001},
  year         = {2014},
  doi          = {10.1088/0034-4885/77/9/094001}
}

@article{deVegaAlonso2017,
  title        = {Dynamics of Non-Markovian Open Quantum Systems},
  author       = {de Vega, In{\'e}s and Alonso, Daniel},
  journal      = {Reviews of Modern Physics},
  volume       = {89},
  number       = {1},
  pages        = {015001},
  year         = {2017},
  doi          = {10.1103/RevModPhys.89.015001}
}

@book{OppenheimWillsky1997,
  title        = {Signals and Systems},
  author       = {Oppenheim, Alan V. and Willsky, Alan S. and Nawab, S. Hamid},
  publisher    = {Prentice Hall},
  address      = {Upper Saddle River},
  year         = {1997},
  edition      = {2nd}
}

@book{AgarwalLang2005,
  title        = {Foundations of Analog and Digital Electronic Circuits},
  author       = {Agarwal, Anant and Lang, Jeffrey H.},
  publisher    = {Morgan Kaufmann},
  address      = {San Francisco},
  year         = {2005}
}

@book{Kato1995,
  title        = {Perturbation Theory for Linear Operators},
  author       = {Kato, Tosio},
  publisher    = {Springer},
  address      = {Berlin},
  year         = {1995},
  edition      = {2nd}
}

@book{ReedSimon1972,
  title        = {Methods of Modern Mathematical Physics, Vol.~I: Functional Analysis},
  author       = {Reed, Michael and Simon, Barry},
  publisher    = {Academic Press},
  address      = {New York},
  year         = {1972}
}

@article{Magos1970,
  title        = {Realization of Driving-Point and Transfer Functions by R--C Ladder Networks},
  author       = {Magos, L.},
  journal      = {International Journal of Electronics},
  volume       = {29},
  number       = {3},
  pages        = {263--276},
  year         = {1970},
  doi          = {10.1080/00207217008901268}
}

@book{Kailath1980,
  title        = {Linear Systems},
  author       = {Kailath, Thomas},
  publisher    = {Prentice Hall},
  address      = {Englewood Cliffs, NJ},
  year         = {1980}
}

@article{FialkowGerst1951,
  title        = {The Routh-Hurwitz Problem for Rational Functions of the Complex Variable},
  author       = {Fialkow, A. and Gerst, S.},
  journal      = {Transactions of the American Mathematical Society},
  volume       = {70},
  number       = {3},
  pages        = {381--420},
  year         = {1951},
  doi          = {10.1090/S0002-9947-1951-0043475-3}
}

@book{Lighthill1958,
  title        = {Introduction to Fourier Analysis and Generalized Functions},
  author       = {Lighthill, M. J.},
  publisher    = {Cambridge University Press},
  address      = {Cambridge},
  year         = {1958}
}

@book{Rudin1987,
  title        = {Real and Complex Analysis},
  author       = {Rudin, Walter},
  publisher    = {McGraw--Hill},
  address      = {New York},
  year         = {1987},
  edition      = {3rd}
}

@book{Bertotti1998,
  title        = {Hysteresis in Magnetism: For Physicists, Materials Scientists, and Engineers},
  author       = {Bertotti, Giorgio},
  publisher    = {Academic Press},
  address      = {San Diego},
  year         = {1998}
}

@book{Mayergoyz2003,
  title        = {Mathematical Models of Hysteresis and Their Applications},
  author       = {Mayergoyz, Isaak D.},
  publisher    = {Elsevier},
  address      = {Amsterdam},
  year         = {2003}
}

@article{Preisach1935,
  title        = {{\"U}ber die magnetische Nachwirkung},
  author       = {Preisach, F.},
  journal      = {Zeitschrift f{\"u}r Physik},
  volume       = {94},
  pages        = {277--302},
  year         = {1935}
}

@article{BoydChua1985,
  title        = {Fading Memory and the Problem of Approximating Nonlinear Operators with Volterra Series},
  author       = {Boyd, Stephen and Chua, Leon O.},
  journal      = {IEEE Transactions on Circuits and Systems},
  volume       = {32},
  number       = {11},
  pages        = {1150--1161},
  year         = {1985}
}

@article{BornFock1928,
  title        = {Beweis des Adiabatensatzes},
  author       = {Born, Max and Fock, Vladimir},
  journal      = {Zeitschrift f{\"u}r Physik},
  volume       = {51},
  pages        = {165--180},
  year         = {1928}
}

@article{AvronElgart1999,
  title        = {Adiabatic Theorem without a Gap Condition},
  author       = {Avron, J. E. and Elgart, A.},
  journal      = {Communications in Mathematical Physics},
  volume       = {203},
  number       = {2},
  pages        = {445--463},
  year         = {1999}
}

@article{JansenRuskaiSeiler2007,
  title        = {Bounds for the Adiabatic Approximation with Applications to Quantum Computation},
  author       = {Jansen, Sabine and Ruskai, Mary Beth and Seiler, Ruedi},
  journal      = {Journal of Mathematical Physics},
  volume       = {48},
  number       = {10},
  pages        = {102111},
  year         = {2007}
}

@article{RigolDunjkoOlshanii2008,
  title        = {Thermalization and Its Mechanism for Generic Isolated Quantum Systems},
  author       = {Rigol, Marcos and Dunjko, Vanja and Olshanii, Maxim},
  journal      = {Nature},
  volume       = {452},
  number       = {7189},
  pages        = {854--858},
  year         = {2008}
}

@article{Polkovnikov2011,
  title        = {Colloquium: Nonequilibrium Dynamics of Closed Interacting Quantum Systems},
  author       = {Polkovnikov, Anatoli},
  journal      = {Reviews of Modern Physics},
  volume       = {83},
  number       = {3},
  pages        = {863--883},
  year         = {2011}
}

@book{SakuraiNapolitano2017,
  title     = {Modern Quantum Mechanics},
  author    = {Sakurai, J. J. and Napolitano, Jim},
  edition   = {3},
  publisher = {Cambridge University Press},
  year      = {2017}
}

@article{Campisi2011,
  author    = {Michele Campisi and Peter H{\"a}nggi and Peter Talkner},
  title     = {Colloquium: Quantum fluctuation relations: Foundations and applications},
  journal   = {Rev. Mod. Phys.},
  volume    = {83},
  number    = {3},
  pages     = {771--791},
  year      = {2011},
  doi       = {10.1103/RevModPhys.83.771}
}

@article{Talkner2007,
  author    = {Peter Talkner and Eric Lutz and Peter H{\"a}nggi},
  title     = {Fluctuation theorems: Work is not an observable},
  journal   = {Phys. Rev. E},
  volume    = {75},
  pages     = {050102},
  year      = {2007},
  doi       = {10.1103/PhysRevE.75.050102}
}

@review{Breuer2016,
  author    = {Heinz-Peter Breuer and Elsi-Mari Laine and Jyrki Piilo and Bassano Vacchini},
  title     = {Colloquium: Non-Markovian dynamics in open quantum systems},
  journal   = {Rev. Mod. Phys.},
  volume    = {88},
  number    = {2},
  pages     = {021002},
  year      = {2016},
  doi       = {10.1103/RevModPhys.88.021002}
}

@article{Rivas2014,
  author    = {\'Angel Rivas and Susana F. Huelga and Martin B. Plenio},
  title     = {Quantum non-Markovianity: Characterization, quantification and detection},
  journal   = {Rep. Prog. Phys.},
  volume    = {77},
  number    = {9},
  pages     = {094001},
  year      = {2014},
  doi       = {10.1088/0034-4885/77/9/094001}
}

@article{Kjaergaard2020,
  author    = {Mikael Kjaergaard and Mollie E. Schwartz and Jochen Braum{\"u}ller and Philip Krantz and Joel I.-J. Wang and Simon Gustavsson and William D. Oliver},
  title     = {Superconducting qubits: Current state of play},
  journal   = {Annual Review of Condensed Matter Physics},
  volume    = {11},
  pages     = {369--395},
  year      = {2020},
  doi       = {10.1146/annurev-conmatphys-031119-050605}
}

@article{Roth2024,
  author    = {Markus Roth and Jonathan Herrmann and Markus F{\"o}rster and others},
  title     = {Cryogenic microwave engineering for scalable superconducting quantum processors},
  journal   = {PRX Quantum},
  year      = {2024},
  note      = {To appear; see also arXiv:2311.00000},
}

@article{Koch2007,
  author    = {Jens Koch and Terri M. Yu and Jay Gambetta and Alexandre A. Houck and D. I. Schuster and J. Majer and Alexandre Blais and M. H. Devoret and S. M. Girvin and R. J. Schoelkopf},
  title     = {Charge-insensitive qubit design derived from the Cooper pair box},
  journal   = {Phys. Rev. A},
  volume    = {76},
  pages     = {042319},
  year      = {2007},
  doi       = {10.1103/PhysRevA.76.042319}
}

@article{You2011,
  author    = {J. Q. You and Franco Nori},
  title     = {Atomic physics and quantum optics using superconducting circuits},
  journal   = {Nature},
  volume    = {474},
  pages     = {589--597},
  year      = {2011},
  doi       = {10.1038/nature10122}
}

@book{NielsenChuang,
  author    = {Michael A. Nielsen and Isaac L. Chuang},
  title     = {Quantum Computation and Quantum Information},
  publisher = {Cambridge University Press},
  year      = {2000}
}

@article{barrows,
  doi = {10.48550/ARXIV.2507.18079},
  url = {https://arxiv.org/abs/2507.18079},
  author = {Barrows,  Frank and Pelofske,  Elijah and Sathe,  Pratik and Caravelli,  Francesco and Nisoli,  Cristiano},
  title = {Magnetic Memory and Hysteresis from Quantum Transitions: Theory and Experiments on Quantum Annealers},
  publisher = {arXiv},
  year = {2025},
  copyright = {Creative Commons Attribution 4.0 International}
}

@article{CallenWelton1951,
  author  = {Callen, Herbert B. and Welton, Theodore A.},
  title   = {Irreversibility and Generalized Noise},
  journal = {Physical Review},
  volume  = {83},
  pages   = {34--40},
  year    = {1951},
  doi     = {10.1103/PhysRev.83.34}
}

@article{Cattaneo2021,
  title = {Engineering Dissipation with Resistive Elements in Circuit Quantum Electrodynamics},
  volume = {4},
  ISSN = {2511-9044},
  url = {http://dx.doi.org/10.1002/qute.202100054},
  DOI = {10.1002/qute.202100054},
  number = {11},
  journal = {Advanced Quantum Technologies},
  publisher = {Wiley},
  author = {Cattaneo,  Marco and Paraoanu,  Gheorghe Sorin},
  year = {2021},
  month = sep 
}

@book{FetterWalecka1971,
  author    = {Fetter, Alexander L. and Walecka, John Dirk},
  title     = {Quantum Theory of Many-Particle Systems},
  publisher = {McGraw-Hill},
  year      = {1971}
}

@book{Mahan2000,
  author    = {Mahan, Gerald D.},
  title     = {Many-Particle Physics},
  edition   = {3},
  publisher = {Springer},
  year      = {2000},
  isbn      = {9780306463389}
}

@book{Forster1975,
  author    = {Forster, Dieter},
  title     = {Hydrodynamic Fluctuations, Broken Symmetry, and Correlation Functions},
  publisher = {W. A. Benjamin},
  year      = {1975}
}

@book{Kamenev2011,
  author    = {Kamenev, Alex},
  title     = {Field Theory of Non-Equilibrium Systems},
  publisher = {Cambridge University Press},
  year      = {2011},
  isbn      = {9780521760829}
}

@book{BreuerPetruccione2002,
  author    = {Breuer, Heinz-Peter and Petruccione, Francesco},
  title     = {The Theory of Open Quantum Systems},
  publisher = {Oxford University Press},
  year      = {2002},
  isbn      = {9780198520634}
}

@article{GardinerCollett1985,
  author  = {Gardiner, C. W. and Collett, M. J.},
  title   = {Input and Output in Damped Quantum Systems: Quantum Stochastic Differential Equations and the Master Equation},
  journal = {Physical Review A},
  volume  = {31},
  number  = {6},
  pages   = {3761--3774},
  year    = {1985},
  doi     = {10.1103/PhysRevA.31.3761}
}

@article{PhysRevApplied.4.024012,
  title        = {Controlling Quantum Devices with Nonlinear Hardware},
  author       = {Hincks, I. N. and Granade, C. E. and Borneman, T. W. and Cory, D. G.},
  journal      = {Phys. Rev. Applied},
  volume       = {4},
  issue        = {2},
  pages        = {024012},
  year         = {2015},
  month        = {Aug},
  publisher    = {American Physical Society},
  doi          = {10.1103/PhysRevApplied.4.024012}
}

@article{Kubo1957Irreversible1,
  author  = {Kubo, R.},
  title   = {Statistical-Mechanical Theory of Irreversible Processes. I. General Theory and Simple Applications to Magnetic and Conduction Problems},
  journal = {Journal of the Physical Society of Japan},
  year    = {1957},
  volume  = {12},
  number  = {6},
  pages   = {570--586},
  doi     = {10.1143/JPSJ.12.570}
}

@article{Kubo1966FDT,
  author  = {Kubo, R.},
  title   = {The fluctuation-dissipation theorem},
  journal = {Reports on Progress in Physics},
  year    = {1966},
  volume  = {29},
  number  = {1},
  pages   = {255--284},
  doi     = {10.1088/0034-4885/29/1/306}
}

@article{FordKacMazur1965,
  author  = {Ford, G. W. and Kac, M. and Mazur, P.},
  title   = {Statistical Mechanics of Assemblies of Coupled Oscillators},
  journal = {Journal of Mathematical Physics},
  year    = {1965},
  volume  = {6},
  number  = {4},
  pages   = {504--515},
  doi     = {10.1063/1.1704304}
}

@article{CaldeiraLeggett1981,
  author  = {Caldeira, A. O. and Leggett, A. J.},
  title   = {Influence of Dissipation on Quantum Tunneling in Macroscopic Systems},
  journal = {Physical Review Letters},
  year    = {1981},
  volume  = {46},
  pages   = {211--214},
  doi     = {10.1103/PhysRevLett.46.211}
}

@article{YurkeDenker1984,
  author  = {Yurke, Bernard and Denker, John S.},
  title   = {Quantum network theory},
  journal = {Physical Review A},
  year    = {1984},
  volume  = {29},
  number  = {3},
  pages   = {1419--1437},
  doi     = {10.1103/PhysRevA.29.1419}
}

@article{Clerk2010QuantumNoise,
  author  = {Clerk, A. A. and Devoret, M. H. and Girvin, S. M. and Marquardt, F. and Schoelkopf, R. J.},
  title   = {Introduction to quantum noise, measurement, and amplification},
  journal = {Reviews of Modern Physics},
  year    = {2010},
  volume  = {82},
  number  = {2},
  pages   = {1155--1208},
  doi     = {10.1103/RevModPhys.82.1155}
}

@misc{caravelli2,
  doi = {10.48550/ARXIV.2509.05793},
  url = {https://arxiv.org/abs/2509.05793},
  author = {Caravelli,  Francesco},
  title = {Spectral Methods in Complex Systems},
  publisher = {arXiv},
  year = {2025},
  copyright = {Creative Commons Attribution 4.0 International}
}

@incollection{Devoret1997LesHouches,
  author    = {Devoret, Michel H.},
  title     = {Quantum Fluctuations in Electrical Circuits},
  booktitle = {Quantum Fluctuations},
  editor    = {Reynaud, Serge and Giacobino, Elisabeth and Zinn-Justin, Jean},
  series    = {Les Houches, Session LXIII (1995)},
  publisher = {Elsevier (North-Holland)},
  year      = {1997},
  pages     = {351--386}
}

@article{VoolDevoret2017,
  author  = {Vool, Uri and Devoret, Michel H.},
  title   = {Introduction to quantum electromagnetic circuits},
  journal = {International Journal of Circuit Theory and Applications},
  year    = {2017},
  volume  = {45},
  number  = {7},
  pages   = {897--934}
}

@article{BornemanCory2012BandwidthLimited,
  title        = {Bandwidth-Limited Control and Ringdown Suppression in High-Q Resonators},
  author       = {Borneman, Troy W. and Cory, David G.},
  journal      = {Journal of Magnetic Resonance},
  volume       = {225},
  pages        = {120--129},
  year         = {2012},
  doi          = {10.1016/j.jmr.2012.10.011},
  eprint       = {1207.1139},
  archivePrefix= {arXiv},
  primaryClass = {quant-ph}
}

@article{Wittler2021C3Toolset,
  title        = {Integrated Tool Set for Control, Calibration, and Characterization of Quantum Devices Applied to Superconducting Qubits},
  author       = {Wittler, Nicolas and Roy, Federico and Pack, Kevin and Werninghaus, Max and Roy, Anurag Saha and Egger, Daniel J. and Filipp, Stefan and Wilhelm, Frank K. and Machnes, Shai},
  journal      = {Phys. Rev. Applied},
  volume       = {15},
  pages        = {034080},
  year         = {2021},
  doi          = {10.1103/PhysRevApplied.15.034080},
  eprint       = {2009.09866},
  archivePrefix= {arXiv},
  primaryClass = {quant-ph}
}

@misc{TektronixBW035,
  author       = {{Tektronix}},
  title        = {Where does the formula {BW} = 0.35 / t\_{10\%-90\%} come from?},
  howpublished = {Tektronix Support FAQ},
  year         = {2025},
  note         = {Accessed: 2025-12-29. \url{https://www.tek.com/en/support/faqs/where-does-formula-bw-035-t10-90-come}}
}

@misc{KeysightScopeRiseTimeBandwidth,
  author       = {{Keysight Technologies}},
  title        = {Bandwidth and Rise Time Requirements for Making Accurate Oscilloscope Measurements},
  howpublished = {Application Note 5991-0662EN},
  year         = {2017},
  note         = {Published Dec. 1, 2017. Accessed: 2025-12-29. \url{https://www.acalbfi.com/media/pdf/app-note-Bandwidth-and-Rise-Time-Requirements-for-making-accurate-scope-measurements.5991-0662.pdf}}
}

@book{Mikusiski1978,
  title = {The Bochner Integral},
  ISBN = {9783034855679},
  url = {http://dx.doi.org/10.1007/978-3-0348-5567-9},
  DOI = {10.1007/978-3-0348-5567-9},
  publisher = {Birkh\"{a}user Basel},
  author = {Mikusiński,  Jan},
  year = {1978}
}

@article{pelofske,
  doi = {10.48550/ARXIV.2511.17779},
  url = {https://arxiv.org/abs/2511.17779},
  author = {Pelofske,  Elijah and Sathe,  Pratik and Nisoli,  Cristiano and Barrows,  Frank},
  title = {Probing Antiferromagnetic Hysteresis on Programmable Quantum Annealers},
  publisher = {arXiv},
  year = {2025},
  copyright = {arXiv.org perpetual,  non-exclusive license}
}

@article{Teufel2003,
  author  = {Teufel, Stefan},
  title   = {Adiabatic Perturbation Theory in Quantum Dynamics},
  journal = {Lecture Notes in Mathematics},
  volume  = {1821},
  publisher = {Springer},
  year    = {2003}
}

@article{DeGrandiPolkovnikov2010,
  author  = {De Grandi, Carlo and Polkovnikov, Anatoli},
  title   = {Adiabatic Perturbation Theory: From Landau–Zener Problem to Quenching Through a Quantum Critical Point},
  journal = {Lecture Notes in Physics},
  volume  = {802},
  pages   = {75--114},
  year    = {2010},
  publisher = {Springer}
}

@article{Blais2021,
  author  = {Blais, Alexandre and Grimsmo, Arne L. and Girvin, S. M. and Wallraff, Andreas},
  title   = {Circuit Quantum Electrodynamics},
  journal = {Reviews of Modern Physics},
  volume  = {93},
  number  = {2},
  pages   = {025005},
  year    = {2021}
}

@article{Kato1950,
  author  = {Kato, Tosio},
  title   = {On the Adiabatic Theorem of Quantum Mechanics},
  journal = {Journal of the Physical Society of Japan},
  volume  = {5},
  number  = {6},
  pages   = {435--439},
  year    = {1950}
}

@book{Abragam1961,
  author    = {Abragam, Anatole},
  title     = {The Principles of Nuclear Magnetism},
  publisher = {Oxford University Press},
  year      = {1961}
}

@book{Slichter1990,
  author    = {Slichter, Charles P.},
  title     = {Principles of Magnetic Resonance},
  series    = {Springer Series in Solid-State Sciences},
  volume    = {1},
  edition   = {3},
  publisher = {Springer},
  year      = {1990}
}

@article{viola0,
  title   = {Dynamical Decoupling of Open Quantum Systems},
  author  = {Viola, L. and Knill, E. and Lloyd, S.},
  journal = {Phys. Rev. Lett.},
  volume  = {82},
  pages   = {2417},
  year    = {1999},
  doi     = {10.1103/PhysRevLett.82.2417}
}

@article{viola1,
  title   = {Robust Dynamical Decoupling of Quantum Systems with Bounded Controls},
  author  = {Viola, L. and Knill, E.},
  journal = {Phys. Rev. Lett.},
  volume  = {90},
  pages   = {037901},
  year    = {2003},
  doi     = {10.1103/PhysRevLett.90.037901}
}

@article{viola2,
  title   = {Enhanced Convergence and Robust Performance of Randomized Dynamical Decoupling},
  author  = {Santos, L. F. and Viola, L.},
  journal = {Phys. Rev. Lett.},
  volume  = {97},
  pages   = {150501},
  year    = {2006},
  doi     = {10.1103/PhysRevLett.97.150501}
}

@article{mont0,
  title   = {Optimal Control at the Quantum Speed Limit},
  author  = {Caneva, T. and Murphy, M. and Calarco, T. and Fazio, R. and Montangero, S. and Giovannetti, V. and Santoro, G. E.},
  journal = {Phys. Rev. Lett.},
  volume  = {103},
  pages   = {240501},
  year    = {2009},
  doi     = {10.1103/PhysRevLett.103.240501}
}

@article{mont1,
  title   = {Communication at the quantum speed limit along a spin chain},
  author  = {Murphy, Michael and Montangero, Simone and Giovannetti, Vittorio and Calarco, Tommaso},
  journal = {Phys. Rev. A},
  volume  = {82},
  pages   = {022318},
  year    = {2010},
  doi     = {10.1103/PhysRevA.82.022318}
}

@article{mont2,
  title   = {Optimal Control Technique for Many-Body Quantum Dynamics},
  author  = {Doria, P. and Calarco, T. and Montangero, S.},
  journal = {Phys. Rev. Lett.},
  volume  = {106},
  pages   = {190501},
  year    = {2011},
  doi     = {10.1103/PhysRevLett.106.190501}
}

@article{mont3,
  title   = {Dressing the chopped-random-basis optimization: A bandwidth-limited access to the trap-free landscape},
  author  = {Rach, N. and M\"uller, M. M. and Calarco, T. and Montangero, S.},
  journal = {Phys. Rev. A},
  volume  = {92},
  pages   = {062343},
  year    = {2015},
  doi     = {10.1103/PhysRevA.92.062343}
}

@article{memristors,
  title = {Memristor-The missing circuit element},
  volume = {18},
  ISSN = {0018-9324},
  url = {http://dx.doi.org/10.1109/TCT.1971.1083337},
  DOI = {10.1109/tct.1971.1083337},
  number = {5},
  journal = {IEEE Transactions on Circuit Theory},
  publisher = {Institute of Electrical and Electronics Engineers (IEEE)},
  author = {Chua,  L.},
  year = {1971},
  pages = {507–519}
}

@article{caravelli,
  title = {Memristors for the Curious Outsiders},
  volume = {6},
  ISSN = {2227-7080},
  url = {http://dx.doi.org/10.3390/technologies6040118},
  DOI = {10.3390/technologies6040118},
  number = {4},
  journal = {Technologies},
  publisher = {MDPI AG},
  author = {Caravelli,  Francesco and Carbajal,  Juan Pablo},
  year = {2018},
  month = dec,
  pages = {118}
}

@article{diventra,
  title = {Circuit Elements With Memory: Memristors,  Memcapacitors,  and Meminductors},
  volume = {97},
  ISSN = {0018-9219},
  url = {http://dx.doi.org/10.1109/JPROC.2009.2021077},
  DOI = {10.1109/jproc.2009.2021077},
  number = {10},
  journal = {Proceedings of the IEEE},
  publisher = {Institute of Electrical and Electronics Engineers (IEEE)},
  author = {Di Ventra,  Massimiliano and Pershin,  Yuriy V. and Chua,  Leon O.},
  year = {2009},
  month = oct,
  pages = {1717–1724}
}
\clearpage
\onecolumngrid

\appendix
\section{Dyson expansion for kernel-filtered control with time-independent drift}
\label{app:dyson_kernel_time_independent_H0}

This appendix gives a self-contained derivation of the interaction-picture
representation and the associated Dyson series for the propagator in the
kernel-filtered control model. Throughout this appendix we assume that the
drift Hamiltonian $\hat H_A$ is time independent.

We start from the model used in the main text,
$$
\hat H(t)=\hat H_A+\hat H^\prime_c(t),
$$
where the filtered control Hamiltonian is defined by the causal convolution
$$
\hat H^\prime_c(t)=\int_{-\infty}^{t}K(t-s)\,\hat H_c(s)\,ds,
\qquad
K(\tau)=0 \ \text{for}\ \tau<0.
$$
When convenient one may take the lower limit to be $0$ by assuming that the
protocol begins at $t=0$ and that the drive is null for $t<0$.

The unitary propagator $\hat U(t)\equiv U(t,0)$ is defined by
$$
i\frac{d}{dt}\hat U(t)=\hat H(t)\hat U(t),
\qquad
U(0)=\mathbb 1,
$$
where we set $\hbar=1$.

Because $\hat H_A$ is time independent, the drift (free) propagator is
$$
\hat U_0(t)=e^{-i\hat H_A t},
\qquad
i\frac{d}{dt}\hat U_0(t)=\hat H_A \hat U_0(t),
\qquad
\hat U_0(0)=\mathbb 1.
$$
Define the interaction-picture propagator by the exact factorization
$$
\hat U(t)=\hat U_0(t)\,\hat U_I(t),
\qquad
\hat U_I(t)=\hat U_0^\dagger(t)\hat U(t)=e^{i\hat H_A t}\hat U(t).
$$
Differentiating $\hat U(t)=\hat U_0(t)\hat U_I(t)$ and using the evolution equations for
$U$ and $\hat U_0$ gives
$$
i\frac{d}{dt}\big(\hat U_0 \hat U_I\big)
=
\hat H_A \hat U_0 \hat U_I+\hat H^\prime_c \hat U_0 \hat U_I,
\qquad
i\frac{d}{dt}\big(\hat U_0 \hat U_I\big)
=
i\dot {\hat U}_0 \hat U_I+i{\hat U}_0\dot{\hat  U}_I.
$$
Since $i\dot \hat U_0=\hat H_A\hat U_0$, the $\hat H_A$ terms cancel and one obtains
$$
i\hat U_0(t)\,\dot \hat U_I(t)=\hat H^\prime_c(t)\,\hat U_0(t)\,\hat U_I(t).
$$
Multiplying on the left by $\hat U_0^\dagger(t)$ yields the interaction-picture
equation of motion
$$
i\frac{d}{dt}\hat U_I(t)=\tilde H_{1,I}(t)\,\hat U_I(t),
\qquad
\hat U_I(0)=\mathbb 1,
$$
where
$$
\tilde H_{1,I}(t)=\hat U_0^\dagger(t)\,\hat H^\prime_c(t)\,\hat U_0(t)
=
e^{i\hat H_A t}\,\hat H^\prime_c(t)\,e^{-i\hat H_A t}.
$$
This derivation shows that, under the standing assumption that $\hat H_A$ is time
independent, the interaction-picture definition used in the main text is
exactly equivalent to the Schr\"odinger-picture evolution.

The formal solution of the interaction-picture equation is
$$
\hat U_I(t)=\mathcal T\exp\!\left(-i\int_0^t \tilde H_{1,I}(\tau)\,d\tau\right),
$$
and expanding the time-ordered exponential yields the Dyson series
$$
\hat U_I(t)
=
\mathbb 1
+
\sum_{n=1}^{\infty}(-i)^n
\int_{0<\tau_n<\cdots<\tau_1<t}
d\tau_1\cdots d\tau_n\;
\tilde H_{1,I}(\tau_1)\cdots \tilde H_{1,I}(\tau_n).
$$
The full propagator is recovered as
$$
\hat U(t)=e^{-i\hat H_A t}\,\hat U_I(t).
$$

A common specialization is a scalar command multiplying a fixed operator,
$$
\hat H_c(s)=u(s)\,\hat M,
$$
with $u(s)\in\mathbb R$ and $\hat M=\hat M^\dagger$. In this case,
$$
\hat H^\prime_c(t)=\int_{-\infty}^{t}K(t-s)\,u(s)\,ds\;\hat M
=:w(t)\,\hat M,
$$
so the filter acts only on the scalar drive, producing the realized field
\begin{eqnarray}
w(t)=\int_{-\infty}^{t}K(t-s)\,u(s)\,ds.
\label{eq:w_def}   
\end{eqnarray}
The interaction-picture operator is then
$$
\tilde H_{1,I}(t)=w(t)\,\hat M_I(t),
\qquad
\hat M_I(t)=e^{i\hat H_A t}\hat M e^{-i\hat H_A t},
$$
and the Dyson series becomes
$$
\hat U_I(t)
=
\mathbb 1
+
\sum_{n=1}^{\infty}(-i)^n
\int_{0<\tau_n<\cdots<\tau_1<t}
d\tau_1\cdots d\tau_n\;
\prod_{j=1}^{n}\Big[w(\tau_j)\hat M_I(\tau_j)\Big].
$$

Because $w$ is a causal linear functional of the command $u$ through the
kernel $K$, substituting $w=K*u$ into the Dyson series expresses $\hat U_I(t)$ as
a multilinear functional expansion in the input waveform $u(\cdot)$. This is
structurally analogous to classical Volterra (Boyd--Chua) series expansions
for nonlinear dynamical systems with memory, where the output is represented
as a sum of multilinear functionals of the input with causal integration
domains. The present setting differs only in that the kernels are
operator-valued and the ordering is quantum (time ordering).

\subsection{Interaction picture and Dyson--Volterra expansion}
\label{subsec:dyson_volterra}

Let $\hat U(t)\equiv \hat U(t,0)$ solve $\dot \hat U(t) = -i \hat H(t)\hat U(t)$ (we set $\hbar=1$).
Write $\hat U(t)=e^{-i\hat H_A t}V(t)$. Then $V(0)=\mathbb{1}$ and
\begin{equation}
    \dot V(t) = -i\,w(t)\,\hat M_I(t)\,V(t),
    \qquad
    \hat M_I(t) := e^{i\hat H_A t}\hat M e^{-i\hat H_A t}.
    \label{eq:V_eqn}
\end{equation}
The formal solution is
\begin{equation}
    \hat V(t) = \mathcal{T}\exp\!\left(-i\int_0^t w(\tau)\,\hat M_I(\tau)\,d\tau\right),
    \label{eq:V_time_ordered}
\end{equation}
and expanding the time-ordered exponential gives a Dyson series in the
\emph{filtered} signal $w$:
\begin{equation}
\begin{aligned}
    \hat V(t)
    &= \mathbb{1}
    + \sum_{n=1}^{\infty}(-i)^n
    \int_{0<\tau_n<\cdots<\tau_1<t}
    d\tau_1\cdots d\tau_n\;
    \prod_{j=1}^{n}\qty[w(\tau_j)\hat M_I(\tau_j)].
\end{aligned}
\label{eq:Dyson_w}
\end{equation}

To obtain a Boyd--Chua / Volterra-type functional expansion directly in the
\emph{command} $u(\cdot)$, substitute $w(\tau_j)=\int_{-\infty}^{\tau_j}
K(\tau_j-s_j)u(s_j)\,ds_j$ into~\eqref{eq:Dyson_w}. One obtains a multilinear
functional series
\begin{equation}
    \hat V(t)
    = \mathbb{1}
    + \sum_{n=1}^{\infty} (-i)^n
    \int_{(-\infty,t]^n} ds_1\cdots ds_n\;
    \qty(\prod_{j=1}^n u(s_j))\;
    \hat G_n(t; s_1,\dots,s_n),
    \label{eq:V_volterra_u}
\end{equation}
where the \emph{operator-valued Volterra kernels} $\hat G_n$ are
\begin{equation}
    \hat G_n(t;s_1,\dots,s_n)
    :=\int_{\substack{0<\tau_n<\cdots<\tau_1<t\\ \tau_j \ge s_j}}
    d\tau_1\cdots d\tau_n\;
    \qty(\prod_{j=1}^n K(\tau_j-s_j))\;
    \hat M_I(\tau_1)\cdots \hat M_I(\tau_n).
\label{eq:Gn_general}
\end{equation}
The propagator is then
\begin{equation}
    \hat U(t) = e^{-i\hat H_A t}\hat V(t),
\end{equation}
with $\hat V(t)$ given by~\eqref{eq:V_volterra_u}.

The kernels admit a simple recursion that mirrors classical Volterra theory:
\begin{equation}
    \hat G_{n+1}(t;s_1,\dots,s_{n+1})
    =
    \int_{\max(0,s_1)}^{t} d\tau\;
    K(\tau-s_1)\,\hat M_I(\tau)\,
    \hat G_n(\tau;s_2,\dots,s_{n+1}),
\label{eq:G_recursion}
\end{equation}
with base case
\begin{equation}
    \hat G_1(t;s) = \int_{\max(0,s)}^{t} d\tau\;K(\tau-s)\,\hat M_I(\tau).
\label{eq:G1_general}
\end{equation}
For strictly causal kernels supported on $\tau\ge 0$, the integrands vanish
unless $s\le \tau$, making the causal structure explicit.

If $\comm{\hat H_A}{\hat M}=0$, then $\hat M_I(t)=\hat M$ and the time ordering
in~\eqref{eq:V_time_ordered} becomes irrelevant. In this special case,
\begin{equation}
    \hat U(t) = e^{-i\hat H_A t}\exp\!\left(-i\hat M \int_0^t w(\tau)\,d\tau\right),
    \label{eq:commuting_exact}
\end{equation}
with $w(\tau)$ given by~\eqref{eq:w_def}. The noncommuting case is the
generic situation, and is captured by the operator-valued kernels
$\hat G_n$ above.

\section{Time-local ODE embedding for scalar kernel-filtered Hamiltonians}
\label{app:embedding_derivation}

In this appendix we show that when the filtered Hamiltonian enters through a
scalar amplitude multiplying a fixed operator, and the kernel is a finite
sum of decaying exponentials, the kernel-filtered evolution is equivalent to
a time-local system consisting of the Schr\"odinger equation coupled to a
finite set of first-order ODEs for auxiliary filter variables.

Assume a drift-plus-control Hamiltonian of the form
$$
\hat H(t)=\hat H_A+\Phi(t)\,\hat M,
$$
with $\hat H_A$ time independent and $\hat M$ Hermitian. The realized control field
$\Phi(t)$ is obtained from a commanded waveform $u(t)$ by a causal kernel,
$$
\Phi(t)=\int_{-\infty}^{t}K(t-s)\,u(s)\,ds,
\qquad
K(\tau)=0\ \text{for}\ \tau<0.
$$
We take $u(t)$ real-valued. If $K(\tau)$ is real-valued, then $\Phi(t)$ is
real and $\hat H(t)$ is Hermitian.

Suppose the kernel admits a finite exponential representation
$$
K(\tau)=\sum_{k=1}^{K_{\max}} c_k e^{-\nu_k \tau}\,\Theta(\tau),
\qquad
\nu_k>0,\ \ c_k\in\mathbb R.
$$
Define auxiliary mode variables by
$$
\Phi_k(t):=\int_{-\infty}^{t} c_k e^{-\nu_k(t-s)}u(s)\,ds,
\qquad
\Phi(t)=\sum_{k=1}^{K_{\max}}\Phi_k(t).
$$
Causality is explicit: only values $s\le t$ contribute.

Differentiate $\Phi_k(t)$ under the integral sign. Write the integrand as
$c_k e^{-\nu_k(t-s)}u(s)$ and note that
$$
\frac{\partial}{\partial t}\,e^{-\nu_k(t-s)}=-\nu_k e^{-\nu_k(t-s)}.
$$
Using Leibniz' rule for differentiation of an integral with variable upper
limit gives
$$
\dot\Phi_k(t)
=
c_k e^{-\nu_k(t-t)}u(t)
+\int_{-\infty}^{t} c_k\frac{\partial}{\partial t}\Big(e^{-\nu_k(t-s)}\Big)u(s)\,ds.
$$
The boundary term is $c_k u(t)$. The remaining term evaluates to
$$
\int_{-\infty}^{t} c_k\Big(-\nu_k e^{-\nu_k(t-s)}\Big)u(s)\,ds
=
-\nu_k\int_{-\infty}^{t} c_k e^{-\nu_k(t-s)}u(s)\,ds
=
-\nu_k\Phi_k(t).
$$
Hence each mode satisfies the first-order ODE
$$
\dot\Phi_k(t)=-\nu_k\Phi_k(t)+c_k u(t).
$$
The corresponding initial condition depends on how the protocol is started.
If one assumes $u(t)=0$ for $t<0$ and begins the evolution at $t=0$, then
$$
\Phi_k(0)=\int_{-\infty}^{0} c_k e^{-\nu_k(0-s)}u(s)\,ds=0.
$$
More generally, prehistory in $u(s)$ for $s<0$ produces a nonzero $\Phi_k(0)$,
which is the correct way to encode nontrivial initial filter memory.

The quantum state obeys the Schr\"odinger equation
$$
i\frac{d}{dt}\ket{\psi(t)}=\Big(\hat H_A+\Phi(t)\hat M\Big)\ket{\psi(t)}
=\Big(\hat H_A+\Big[\sum_{k=1}^{K_{\max}}\Phi_k(t)\Big]\hat M\Big)\ket{\psi(t)}.
$$
Together with the mode equations, the kernel-filtered dynamics is therefore
equivalent to the time-local system
$$
\begin{cases}
i\dfrac{d}{dt}\ket{\psi(t)}
=\Big(\hat H_A+\Big[\sum_{k=1}^{K_{\max}}\Phi_k(t)\Big]\hat M\Big)\ket{\psi(t)},\\[8pt]
\dot\Phi_k(t)=-\nu_k\Phi_k(t)+c_k u(t),
\qquad k=1,\ldots,K_{\max}.
\end{cases}
$$
This embedding replaces the convolutional memory in the Hamiltonian by a
finite-dimensional classical state $(\Phi_1,\ldots,\Phi_{K_{\max}})$.
The resulting evolution remains unitary because the instantaneous Hamiltonian
acting on the quantum state is Hermitian for real $\Phi(t)$.

Let us now introduce the special case of a RC kernel.
For $K_{\max}=1$ one has $K(\tau)=c e^{-\nu \tau}\Theta(\tau)$ and
$\Phi(t)=\Phi_1(t)$, so the filter reduces to the familiar relaxation
equation
$$
\dot\Phi(t)=-\nu\Phi(t)+c\,u(t),
$$
coupled to the Schr\"odinger equation with $\hat H(t)=\hat H_A+\Phi(t)\hat M$.

Many kernels that arise in practice are not exactly finite sums of
exponentials but can be accurately approximated by such sums on a relevant
time window, either from physics (e.g. a Matsubara-like decompositions) or from
system-identification procedures. In that case, the above embedding provides
a controlled time-local approximation to the kernel-filtered dynamics, with
the approximation quality inherited from the kernel fit.

\section{Expression for $f'(0)$ in adiabatic case}
\label{app:expressionfprime0}
Starting from
\begin{equation}
f(\Phi)=\sum_i \frac{p_i}{\mathrm{Tr}\,\Pi_i(\Phi)}\,\mathrm{Tr}\!\big(\Pi_i(\Phi)\,\hat O\big),
\qquad
H(\Phi)=\hat H_A+\Phi\,\hat M,
\label{eq:fPhi_main_repeat}
\end{equation}
and assuming an iso-degenerate spectral path (so $\mathrm{Tr}\,\Pi_i(\Phi)$ is constant) and that $\hat O$ has no
explicit $\Phi$-dependence, one has
\begin{equation}
f'(0)=\sum_i \frac{p_i}{d_i}\,\mathrm{Tr}\!\Big( \Pi_i'(0)\,\hat O\Big),
\qquad d_i:=\mathrm{Tr}\,\Pi_i(0),
\label{eq:fprime_projector_repeat}
\end{equation}
where $\Pi_i'(\Phi)=\partial_\Phi \Pi_i(\Phi)$.

Assume the $i$-th spectral band of $H(0)$ is separated by a nonzero gap from the rest of the spectrum.
Let
\begin{equation}
Q_i(0):=\mathbb 1-\Pi_i(0),
\qquad
R_i(0):=Q_i(0)\,\big(H(0)-\varepsilon_i(0)\big)^{-1}Q_i(0),
\label{eq:reduced_resolvent_def}
\end{equation}
(the inverse is taken on the range of $Q_i(0)$). A standard Kato/analytic-perturbation identity gives
\begin{equation}
\Pi_i'(0)=R_i(0)\,\hat M\,\Pi_i(0)+\Pi_i(0)\,\hat M\,R_i(0).
\label{eq:Pi_prime_resolvent_again}
\end{equation}
Substituting into \eqref{eq:fprime_projector_repeat} yields the explicit resolvent form
\begin{equation}
f'(0)=\sum_i \frac{p_i}{d_i}\,
\mathrm{Tr}\!\Big(\big[R_i(0)\,\hat M\,\Pi_i(0)+\Pi_i(0)\,\hat M\,R_i(0)\big]\hat O\Big).
\label{eq:fprime_resolvent_form}
\end{equation}
Using cyclicity of the trace and Hermiticity ($R_i(0)=R_i(0)^\dagger$, $\hat M=\hat M^\dagger$, $\hat O=\hat O^\dagger$),
this can be written as a manifestly real expression,
\begin{equation}
f'(0)=2\sum_i \frac{p_i}{d_i}\,
\mathsf{Re}\,\mathrm{Tr}\!\Big(\Pi_i(0)\,\hat O\,R_i(0)\,\hat M\,\Pi_i(0)\Big).
\label{eq:fprime_resolvent_real}
\end{equation}
where $\mathsf{Re}$ implies the real part.

If $H(0)$ has a spectral decomposition into band projectors
\begin{equation}
H(0)=\sum_j \varepsilon_j(0)\,\Pi_j(0),
\label{eq:spectral_H0_projectors}
\end{equation}
then on the complement of band $i$ the reduced resolvent admits the expansion
\begin{equation}
R_i(0)=\sum_{j\neq i}\frac{1}{\varepsilon_j(0)-\varepsilon_i(0)}\,\Pi_j(0),
\label{eq:Ri_band_sum}
\end{equation}
which, inserted into \eqref{eq:Pi_prime_resolvent_again}, gives the familiar first-order perturbation-theory formula for
the derivative of the spectral projector:
\begin{equation}
\Pi_i'(0)=\sum_{j\neq i}\frac{\Pi_j(0)\,\hat M\,\Pi_i(0)+\Pi_i(0)\,\hat M\,\Pi_j(0)}{\varepsilon_i(0)-\varepsilon_j(0)}.
\label{eq:Pi_prime_band_sum}
\end{equation}
Substituting \eqref{eq:Pi_prime_band_sum} into \eqref{eq:fprime_projector_repeat} yields
\begin{equation}
f'(0)=\sum_i\frac{p_i}{d_i}\sum_{j\neq i}
\frac{\mathrm{Tr}\!\Big(\big[\Pi_j(0)\,\hat M\,\Pi_i(0)+\Pi_i(0)\,\hat M\,\Pi_j(0)\big]\hat O\Big)}
{\varepsilon_i(0)-\varepsilon_j(0)}.
\label{eq:fprime_projector_band_sum}
\end{equation}
A slightly cleaner (and again manifestly real) form is obtained by cyclicity of the trace:
\begin{equation}
f'(0)=2\sum_i\frac{p_i}{d_i}\sum_{j\neq i}
\frac{\mathsf{Re}\,\mathrm{Tr}\!\Big(\Pi_i(0)\,\hat O\,\Pi_j(0)\,\hat M\,\Pi_i(0)\Big)}
{\varepsilon_i(0)-\varepsilon_j(0)}.
\label{eq:fprime_projector_band_sum_real}
\end{equation}

If the spectrum is nondegenerate, $\Pi_i(0)=|i\rangle\langle i|$ and $d_i=1$,
then \eqref{eq:fprime_projector_band_sum_real} reduces to the standard first-order eigenstate perturbation formula:
\begin{equation}
f'(0)=2\sum_i p_i\sum_{j\neq i}
\frac{\mathsf{Re}\!\Big(\langle i|\hat O|j\rangle\,\langle j|\hat M|i\rangle\Big)}
{\varepsilon_i(0)-\varepsilon_j(0)}.
\label{eq:fprime_nondegenerate}
\end{equation}

Equations \eqref{eq:Pi_prime_resolvent_again}--\eqref{eq:fprime_nondegenerate} require a nonzero gap separating the chosen
band(s) at $\Phi=0$. At degeneracy crossings, $f'(0)$ need not be well-defined without specifying how the band(s) are
tracked through the crossing (and adiabaticity itself becomes subtle).
\section{Loop integrals, uniqueness, and bounds}\label{app:integralsbounds}

\subsection{Setup and definitions}

Let $u \in [0,T]$ be a cycle parameter (e.g.\ time or sweep parameter). Assume
\begin{equation}
\Phi:[0,T]\to\mathbb{R},\qquad O:[0,T]\to\mathbb{R}.
\end{equation}
We consider the (oriented) loop integral
\begin{equation}\label{eq:def_AuO}
A_{uO} \;:=\; \oint O\,d\Phi,
\end{equation}
interpreted as a Riemann--Stieltjes integral along the closed cycle. Throughout, we assume the cycle is closed in $\Phi$:
\begin{equation}\label{eq:closedPhi}
\Phi(0)=\Phi(T).
\end{equation}

A sufficient condition for reducing \eqref{eq:def_AuO} to an ordinary integral is absolute continuity of $\Phi$.

\begin{lemma}[Parameterization formula]\label{lem:param_formula}
If $\Phi$ is absolutely continuous on $[0,T]$, then $\Phi'(u)$ exists for almost every $u$ and is integrable, and
\begin{equation}\label{eq:param_integral}
A_{uO}=\oint O\,d\Phi=\int_0^T O(u)\,\Phi'(u)\,du.
\end{equation}
\end{lemma}

\begin{proof}
For absolutely continuous $\Phi$, the differential $d\Phi$ is represented by $\Phi'(u)\,du$ thus, the integral reduces to the integral:
\begin{equation}
\int_0^T O\,d\Phi=\int_0^T O(u)\,\Phi'(u)\,du.
\end{equation}
For a closed cycle, the loop integral $\oint O\,d\Phi$ is understood as the integral over one period, i.e.\ the same expression on $[0,T]$.
\end{proof}

We also define the \emph{total variation} (total sweep) of $\Phi$ with respect to $u$:
\begin{equation}\label{eq:variation}
V_{u\Phi} \;:=\;\int_0^T |\Phi'(u)|\,du,
\end{equation}
well-defined whenever $\Phi$ is absolutely continuous.

\begin{proposition}[Single-valued $O=f(\Phi)$ $\Rightarrow$ zero loop integral]\label{prop:unique_zero}
Assume $\Phi$ is absolutely continuous, \eqref{eq:closedPhi} holds, and there exists a \emph{single-valued} function $f:\mathbb{R}\to\mathbb{R}$ such that
\begin{equation}
O(u)=f(\Phi(u))\quad \text{for all }u\in[0,T].
\end{equation}
If $f$ is continuous (or merely integrable) so that an antiderivative $F$ exists with $F'(\Phi)=f(\Phi)$ almost everywhere, then
\begin{equation}
A_{uO}=\oint f(\Phi)\,d\Phi=0.
\end{equation}
\end{proposition}

\begin{proof}
Let $F$ be an antiderivative of $f$, i.e.\ $F'(\Phi)=f(\Phi)$ (a.e.). Along the trajectory,
\begin{equation}
dF = F'(\Phi)\,d\Phi = f(\Phi)\,d\Phi.
\end{equation}
Therefore,
\begin{equation}
\oint f(\Phi)\,d\Phi = \oint dF = F(\Phi(T)) - F(\Phi(0)).
\end{equation}
Using $\Phi(T)=\Phi(0)$ from \eqref{eq:closedPhi} yields $F(\Phi(T)) - F(\Phi(0))=0$, hence $A_{uO}=0$.
\end{proof}

\begin{remark}
Proposition~\ref{prop:unique_zero} formalizes the standard statement: a nonzero loop area $\oint O\,d\Phi$ requires that $O$ is \emph{not} a single-valued function of $\Phi$ along the cycle (i.e.\ there is hysteresis / multibranch behavior).
\end{remark}

\begin{theorem}[Variation bound]\label{thm:variation_bound}
Assume $\Phi$ is absolutely continuous, \eqref{eq:closedPhi} holds, and $O$ is bounded:
\begin{equation}
O_{\min}\le O(u)\le O_{\max}\quad \text{for all }u\in[0,T].
\end{equation}
Define $\Delta O:=O_{\max}-O_{\min}$ and $V_{u\Phi}$ as in \eqref{eq:variation}. Then
\begin{equation}\label{eq:variation_bound}
|A_{uO}|\;\le\;\frac{\Delta O}{2}\,V_{u\Phi}
\;=\;\frac{\Delta O}{2}\int_0^T |\Phi'(u)|\,du.
\end{equation}
Equivalently,
\begin{equation}
-\frac{\Delta O}{2}\,V_{u\Phi}\;\le\;A_{uO}\;\le\;\frac{\Delta O}{2}\,V_{u\Phi}.
\end{equation}
\end{theorem}

\begin{proof}
We set the midrange constant $O_c:=\tfrac{1}{2}(O_{\max}+O_{\min})$ and define $\widetilde O(u):=O(u)-O_c$.
Then $|\widetilde O(u)|\le \Delta O/2$ for all $u$.

Using Lemma~\ref{lem:param_formula},
\begin{equation}
A_{uO}=\int_0^T O(u)\,\Phi'(u)\,du
=\int_0^T (O_c+\widetilde O(u))\,\Phi'(u)\,du.
\end{equation}
The constant part vanishes because
\begin{equation}
\int_0^T O_c\,\Phi'(u)\,du
= O_c\big(\Phi(T)-\Phi(0)\big)=0
\end{equation}
by \eqref{eq:closedPhi}. Thus
\begin{equation}
A_{uO}=\int_0^T \widetilde O(u)\,\Phi'(u)\,du.
\end{equation}
Taking absolute values and applying the triangle inequality,
\begin{equation}
|A_{uO}|
\le \int_0^T |\widetilde O(u)|\,|\Phi'(u)|\,du
\le \frac{\Delta O}{2}\int_0^T |\Phi'(u)|\,du
=\frac{\Delta O}{2}\,V_{u\Phi}.
\end{equation}
This proves \eqref{eq:variation_bound}. The two-sided bound follows immediately.
\end{proof}

\begin{corollary}[Sup-norm and Cauchy--Schwarz bounds]\label{cor:norm_bounds}
Under the assumptions of Lemma~\ref{lem:param_formula},
\begin{equation}
|A_{uO}|\le \|O\|_{\infty}\int_0^T |\Phi'(u)|\,du,
\end{equation}
and, if additionally $O\in L^2(0,T)$ and $\Phi'\in L^2(0,T)$,
\begin{equation}
|A_{uO}|\le \|O\|_{2}\,\|\Phi'\|_{2}.
\end{equation}
\end{corollary}

\begin{proof}
From \eqref{eq:param_integral},
\begin{equation}
|A_{uO}| = \left|\int_0^T O\,\Phi'\,du\right|
\le \int_0^T |O|\,|\Phi'|\,du
\le \|O\|_{\infty}\int_0^T |\Phi'|\,du,
\end{equation}
proving the sup-norm bound. The $L^2$ bound is Cauchy--Schwarz:
\begin{equation}
\left|\int_0^T O\,\Phi'\,du\right|
\le \left(\int_0^T |O|^2\,du\right)^{1/2}
     \left(\int_0^T |\Phi'|^2\,du\right)^{1/2}.
\end{equation}
\end{proof}

The loop-area becomes particularly transparent when the cycle consists of an ``up-sweep'' in $\Phi$ followed by a ``down-sweep''.

\begin{lemma}[Two-branch formula for a single turning point]\label{lem:branch_formula}
Assume $\Phi$ is absolutely continuous and there exists $u_\ast\in(0,T)$ such that:
\begin{equation}
\Phi \text{ is strictly increasing on } [0,u_\ast],\qquad
\Phi \text{ is strictly decreasing on } [u_\ast,T],
\end{equation}
and $\Phi(0)=\Phi(T)=\Phi_{\min}$, $\Phi(u_\ast)=\Phi_{\max}$.
Define the two branches as single-valued functions of $\Phi$ on $[\Phi_{\min},\Phi_{\max}]$ by
\begin{equation}
O_\uparrow(\varphi) := O\big(u_\uparrow(\varphi)\big),
\quad
O_\downarrow(\varphi) := O\big(u_\downarrow(\varphi)\big),
\end{equation}
where $u_\uparrow$ is the inverse of $\Phi|_{[0,u_\ast]}$ and $u_\downarrow$ is the inverse of $\Phi|_{[u_\ast,T]}$.
Then
\begin{equation}\label{eq:branch_area}
A_{uO}=\oint O\,d\Phi
=\int_{\Phi_{\min}}^{\Phi_{\max}}\bigl(O_\uparrow(\varphi)-O_\downarrow(\varphi)\bigr)\,d\varphi.
\end{equation}
\end{lemma}

\begin{proof}
Split the loop integral into two parts and use Lemma~\ref{lem:param_formula}:
\begin{equation}
A_{uO}=\int_0^{u_\ast} O(u)\,\Phi'(u)\,du + \int_{u_\ast}^{T} O(u)\,\Phi'(u)\,du.
\end{equation}
On $[0,u_\ast]$, $\Phi$ is strictly increasing, so the change of variables $\varphi=\Phi(u)$ is valid and $d\varphi=\Phi'(u)\,du$:
\begin{equation}
\int_0^{u_\ast} O(u)\,\Phi'(u)\,du
=\int_{\Phi_{\min}}^{\Phi_{\max}} O\big(u_\uparrow(\varphi)\big)\,d\varphi
=\int_{\Phi_{\min}}^{\Phi_{\max}} O_\uparrow(\varphi)\,d\varphi.
\end{equation}
On $[u_\ast,T]$, $\Phi$ is strictly decreasing, so the same substitution yields reversed limits:
\begin{equation}
\int_{u_\ast}^{T} O(u)\,\Phi'(u)\,du
=\int_{\Phi_{\max}}^{\Phi_{\min}} O\big(u_\downarrow(\varphi)\big)\,d\varphi
=-\int_{\Phi_{\min}}^{\Phi_{\max}} O_\downarrow(\varphi)\,d\varphi.
\end{equation}
Adding both expressions gives \eqref{eq:branch_area}.
\end{proof}

\begin{proposition}[Lower bound from uniform branch separation]\label{prop:lower_bound_gap}
Under the assumptions of Lemma~\ref{lem:branch_formula}, suppose there exists $\delta>0$ such that
\begin{equation}
O_\uparrow(\varphi)-O_\downarrow(\varphi)\ge \delta
\quad \text{for all }\varphi\in[\Phi_{\min},\Phi_{\max}].
\end{equation}
Then
\begin{equation}\label{eq:lower_bound_gap}
A_{uO}\ge \delta\,(\Phi_{\max}-\Phi_{\min}).
\end{equation}
If instead $O_\downarrow(\varphi)-O_\uparrow(\varphi)\ge \delta$, then
$A_{uO}\le -\delta(\Phi_{\max}-\Phi_{\min})$.
\end{proposition}

\begin{proof}
From \eqref{eq:branch_area},
\begin{equation}
A_{uO}=\int_{\Phi_{\min}}^{\Phi_{\max}}\bigl(O_\uparrow(\varphi)-O_\downarrow(\varphi)\bigr)\,d\varphi
\ge \int_{\Phi_{\min}}^{\Phi_{\max}} \delta\,d\varphi
= \delta(\Phi_{\max}-\Phi_{\min}),
\end{equation}
which proves \eqref{eq:lower_bound_gap}. The second case is identical after swapping the branches.
\end{proof}

\begin{remark}[Why lower bounds need extra assumptions]
Without a condition such as a sign-definite branch difference $O_\uparrow-O_\downarrow$,
cancellations may occur and one cannot guarantee any strictly positive lower bound on $|A_{uO}|$
from bounds on $|O|$ and $|\Phi'|$ alone; the universal lower bound is then $0$.
\end{remark}

\section{Nonadiabatic cyclic integrals for filtered control: identities and bounds}
\label{app:nonadiabatic_cyclic_integrals}

This appendix provides the derivations underlying the bounds for $I$ in the main text.
The only assumption is unitary evolution
generated by the time-dependent Hamiltonian.

Let $\hat O$ be time independent and define
\begin{equation}
    I \;=\; \oint dt\;\mathrm{Tr}\!\big(\rho(t)\hat O\big)\,\dot u(t)
\end{equation}
over one cycle of duration $T$ in steady periodic response,
$u(T)=u(0)$ and $\rho(T)=\rho(0).$ Integration by parts gives
\begin{equation}
    I
    =\Big[u(t)\,\mathrm{Tr}\!\big(\rho(t)\hat O\big)\Big]_0^T
    -\int_0^T dt\;u(t)\,\frac{d}{dt}\mathrm{Tr}\!\big(\rho(t)\hat O\big)
    =-\oint dt\;u(t)\,\mathrm{Tr}\!\big(\dot\rho(t)\hat O\big).
\end{equation}
Under unitary dynamics, $$\dot\rho(t)=-i[\hat H(t),\rho(t)],$$ and by cyclicity of
the trace,
\begin{equation}
    \mathrm{Tr}\!\big(\dot\rho(t)\hat O\big)
    =-i\,\mathrm{Tr}\!\big(\rho(t)[\hat O,\hat H(t)]\big)
    =-i\,\big\langle[\hat O,\hat H(t)]\big\rangle_t.
\end{equation}
Substituting yields
\begin{equation}
    I \;=\; i\oint dt\;u(t)\,\big\langle[\hat O,\hat H(t)]\big\rangle_t.
\end{equation}

Let us now assume the scalar filtered-control model
\begin{equation}
    \hat H(t)=\hat H_A+\Phi(t)\hat M,
    \qquad
    \Phi(t)=\int_0^t K(t-s)\,u(s)\,ds,
\end{equation}
with real causal kernel $K$. Define the commutator weights
\begin{equation}
    a(t)=\big\langle[\hat O,\hat H_A]\big\rangle_t,
    \qquad
    b(t)=\big\langle[\hat O,\hat M]\big\rangle_t.
\end{equation}
Then
$$\big\langle[\hat O,\hat H(t)]\big\rangle_t=a(t)+\Phi(t)b(t),$$
and the cyclic functional splits exactly as
\begin{equation}
    I
    =
    i\oint_0^T dt\;u(t)\,a(t)
    +
    i\oint_0^T dt\;u(t)\,\Phi(t)\,b(t).
\end{equation}
The first term depends on correlations between the drive and $a(t)$ and is
present even when $K(t)$ approaches a delta function. The second term is the
contribution in which the control history enters explicitly through the
convolution; it is the term that reduces to the quadratic functional studied
in the main text, when $b(t)$ may be treated
as constant.

To be concrete, we write the memory contribution as
\begin{equation}
    I_{\rm mem}
    = i\int_0^T dt\;b(t)\,u(t)\,\Phi(t)
    = i\int_0^T dt\;b(t)\,u(t)\int_0^t ds\;K(t-s)\,u(s).
\end{equation}
Introducing $$v(t)=b(t)u(t)$$ gives the triangular bilinear form
\begin{equation}
    I_{\rm mem}
    = i\int_0^T dt\int_0^t ds\; v(t)\,K(t-s)\,u(s),
\end{equation}
which makes the temporal nonlocality explicit: values of the drive at earlier
times $s<t$ contribute to the response at time $t$ with weight $K(t-s).$
When $b(t)$ is constant, $v(t)\propto u(t)$ and the expression becomes a
genuine quadratic functional of the command waveform.

We now bound the norm of the memory content.
The bias term obeys Cauchy--Schwarz:
\begin{equation}
    |I_{\rm bias}|
    =\left|\int_0^T dt\;u(t)\,a(t)\right|
    \le \|u\|_{L^2(0,T)}\,\|a\|_{L^2(0,T)}.
\end{equation}
For the memory term, first bound $$|b(t)|\le \|b\|_\infty$$ to obtain
\begin{equation}
    |I_{\rm mem}|
    \le \|b\|_\infty \int_0^T dt\;|u(t)\Phi(t)|
    \le \|b\|_\infty\,\|u\|_{L^2(0,T)}\,\|\Phi\|_{L^2(0,T)}.
\end{equation}
Since $\Phi=K_c*u$ with $K_c(\tau)=K(\tau)\Theta(\tau)\in L^1(\mathbb R_+)$,
Young's inequality yields
\begin{equation}
    \|\Phi\|_{L^2(0,T)}
    \le \|K_c\|_{L^1(\mathbb R_+)}\,\|u\|_{L^2(0,T)}.
\end{equation}
Combining these two gives
\begin{equation}
    |I_{\rm mem}|
    \le \|b\|_\infty\,\|K_c\|_{L^1(\mathbb R_+)}\,\|u\|_{L^2(0,T)}^2,
\end{equation}
and therefore
\begin{equation}
    |I|
    \le \|u\|_{L^2(0,T)}\,\|a\|_{L^2(0,T)}
    +  \|b\|_\infty\,\|K_c\|_{L^1(\mathbb R_+)}\,\|u\|_{L^2(0,T)}^2,
\end{equation}
which is the inequality we report in the paper.

In the instantaneous-control limit $K(\tau)\to\delta(\tau)$,
$\Phi(t)\to u(t)$ and the memory contribution becomes local,
\begin{equation}
    I_{\rm mem}\to i\int_0^T dt\;b(t)\,u(t)^2.
\end{equation}
Departures from this local reduction quantify the role of the control-channel
memory kernel.

For completeness, we provide exact formulae in the case of sinusoidal control and an exponential kernel function; consider the  kernel
\begin{equation}
    K(\tau)=\frac{1}{\tau_c}\,e^{-\tau/\tau_c}\,\Theta(\tau),
\end{equation}
so that $\Phi$ satisfies $\tau_c\,\dot\Phi+\Phi=u$. For $u(t)=A\cos(\omega t)$,
the solution with $\Phi(0)=0$ is
\begin{equation}
    \Phi(t)
    =
    \frac{A}{1+(\omega\tau_c)^2}\Big[\cos(\omega t)+\omega\tau_c\sin(\omega t)-e^{-t/\tau_c}\Big].
\end{equation}
The quadratic functional $J(T)=\int_0^T u(t)\Phi(t)\,dt$ and its special cases
(integer numbers of periods and long-time limit) follow by direct integration,
recovering the expressions stated in the main text when $b(t)$ is treated as
constant and the bias term is suppressed by symmetry or design.

\section{Classical control channel: RC synthesis, pole structure, and exponential-mode kernels}
\label{app:rc_realizability}

\begin{figure}[t]
\centering
\begin{circuitikz}[american]
  \draw (0,0) node[left] {$u(t)$} to[short, o-] (0,0);

  \draw (0,0) to[R=$R_1$] (2,0) coordinate (v1);
  \draw (v1) node[above] {$V_1$} to[C=$C_1$] (2,-2) node[ground]{};

  \draw (2,0) to[short] (3,0);
  \draw (3.2,0) node {$\cdots$};

  \draw (4,0) node[above] {$V_j$};

  \draw (5,0) to[R=$R_N$] (7,0) coordinate (vN);
  \draw (vN) node[above] {\hspace{1cm}$V_N=\Phi(t)$} to[C=$C_N$] (7,-2) node[ground]{};
  \draw (7,0) to[short, -o] (7.8,0);
\end{circuitikz}
\caption{Discrete RC ladder (lossy line) model of the control channel. The commanded input voltage $u(t)$ drives a chain of series resistors with shunt capacitors to ground. The delivered device-node voltage is $\Phi(t)=V_N(t)$.}
\label{fig:rc_ladder}
\end{figure}
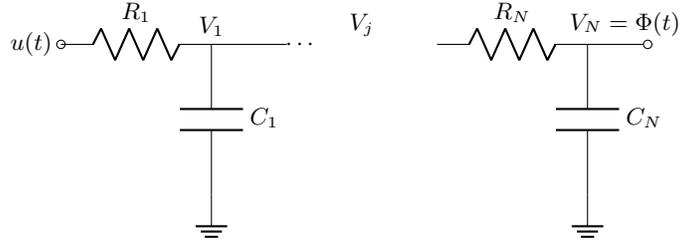

This appendix motivates the use of RC ladder models for the control channel and
connects the associated transfer function to the exponential-mode representation
used in the main text. The essential point is that the class of transfer
functions produced by passive RC networks has a highly constrained analytic
structure (stability, real coefficients, and pole/zero restrictions). That same
structure implies that the corresponding impulse response is a nonnegative
mixture of decaying exponentials (discrete for finite ladders, continuous in the
distributed limit). This is precisely the statement that the realized field
$
\Phi(t) = (K * u)(t)
$
can be represented (exactly or to controlled accuracy) by a finite set of
auxiliary modes obeying first-order ODEs.

We model the control line and bias circuitry as a passive linear time-invariant
two-port that maps an input voltage waveform $u(t)$ to the voltage $\Phi(t)$ at
the device node. In the Laplace domain (with $s$ in the right half-plane),
this defines a transfer function
$$
G(s) := \frac{\Phi(s)}{u(s)}.
$$
Passivity and causality imply that the network is stable (no poles in
$\Re(s) > 0$) and that $G(s)$ is analytic there with real coefficients. For RC
networks (resistors and capacitors only), the pole structure is even more
restricted: the transfer function of a passive RC two-port is a real rational
function whose poles lie on the negative real axis (and, in the generic case,
are simple). A detailed characterization of the transfer functions of general
RC ladder topologies is classical; for the ladder class relevant here, these
restrictions are derived and discussed in \cite{FialkowGerst1951}.

For synthesis, one often specifies a desired voltage-ratio function $G(s)$ and
asks when it is realizable by an RC ladder. A standard criterion is that
voltage-ratio functions with negative poles and nonpositive zeros fall into the
RC ladder realizability class; one constructive route proceeds by selecting an
appropriate driving-point function (an impedance or admittance) whose numerator
and denominator are matched to the prescribed transfer denominator and then
performing a ladder extraction step-by-step \cite{Magos1970}. In other words,
the analytic restrictions are not merely consequences of RC ladders; they are
also the design constraints that ensure a ladder realization exists.

In the present work we use these results in the forward direction: if the
control channel is well described by a passive RC network (a common and
physically motivated assumption for wiring, filters, bias tees, and distributed
lossy lines at the relevant frequencies), then $G(s)$ belongs to the above
class, and therefore its time-domain impulse response necessarily has an
exponential-mode representation.

Assume first that $G(s)$ is rational and strictly proper or proper. Under the
RC restrictions, all poles lie on the negative real axis. For a finite ladder,
there are finitely many poles, and one may write a partial-fraction expansion
of the form
$$
G(s) = g_\infty + \sum_{k=1}^{K} \frac{c_k}{s+\nu_k},
\qquad
\nu_k > 0,
$$
where $g_\infty$ accounts for any instantaneous feedthrough (if present). The
corresponding causal impulse response is
$$
K(t) = g_\infty\,\delta(t) + \sum_{k=1}^{K} c_k e^{-\nu_k t}\,\Theta(t).
$$
Consequently,
$$
\Phi(t) = (K*u)(t)
       = g_\infty\,u(t) + \sum_{k=1}^{K} c_k \int_{-\infty}^{t} e^{-\nu_k(t-s)}u(s)\,ds.
$$
Introducing auxiliary coordinates
$$
\Phi_k(t) := \int_{-\infty}^{t} c_k e^{-\nu_k(t-s)}u(s)\,ds,
\qquad
\Phi(t) = g_\infty\,u(t) + \sum_{k=1}^{K}\Phi_k(t),
$$
one obtains the time-local realization
$$
\dot \Phi_k(t) = -\nu_k \Phi_k(t) + c_k u(t),
\qquad
k=1,\dots,K.
$$
This is the ODE embedding used in the main text. For a finite RC ladder, it is
not an approximation: it is simply the state-space realization implied by the
pole expansion of $G(s)$.

The memoryless (instantaneous) limit corresponds to $K(t) \propto \delta(t)$, i.e.
$G(s) \to const$ (up to a static gain convention). In the exponential-mode picture,
a convenient sufficient condition for this limit is:
(i) the total weight satisfies $g_\infty + \sum_k c_k = 1$, and
(ii) all time scales collapse, in the sense that $\max_k \nu_k^{-1} \to 0$.
Then $K$ concentrates at $t=0$ and $\Phi(t)$ approaches $u(t)$ uniformly on
protocols that do not vary on vanishing time scales.

Conversely, the longest memory time is controlled by the slowest mode
$$
\tau_{\rm mem} \sim \nu_{\min}^{-1},
\qquad
\nu_{\min} := \min_k \nu_k.
$$
Retaining only this slowest pole is precisely the single-RC approximation:
it keeps the dominant long-time relaxation channel and discards faster
components of the line response.

To make the preceding discussion concrete, we now write explicitly the dynamics
of a discrete RC ladder (a lumped-element model of a lossy line). Consider the
$N$-stage network sketched in Fig.~\ref{fig:rc_ladder}: the commanded input
voltage $u(t)$ drives a chain of series resistors, and each intermediate node is
shunted to ground by a capacitor. Let $V_j(t)$ denote the voltage at node $j$
for $j=1,\dots,N$, with the boundary condition $V_0(t)=u(t)$, and identify the
delivered device-node voltage as
$$
\Phi(t)=V_N(t).
$$
For clarity we allow the component values to vary by stage ($R_j$, $C_j$), and
recover the uniform case by setting $R_j\equiv R$, $C_j\equiv C$.

Kirchhoff's current law at node $j$ equates the capacitor current to the net
resistive current flowing into that node. For the first node one obtains
$$
C_1\,\dot V_1(t)
=
\frac{V_0(t)-V_1(t)}{R_1}
-
\frac{V_1(t)-V_2(t)}{R_2},
$$
for interior nodes $j=2,\dots,N-1$ one has
\begin{eqnarray}
C_j\,\dot V_j(t)
&=& i_j-i_{j+1}\nonumber \\
&=&\frac{V_{j-1}(t)-V_j(t)}{R_j}
-
\frac{V_j(t)-V_{j+1}(t)}{R_{j+1}},
\end{eqnarray}

and for the terminal node (open end beyond node $N$) one finds
$$
C_N\,\dot V_N(t)
=
\frac{V_{N-1}(t)-V_N(t)}{R_N}.
$$
In this model, thus, at equilibrium the currents $i_j=i_{j+1}$ and the shunted capacitances do not leak currents into the ground.
In the uniform ladder case ($R_j\equiv R$, $C_j\equiv C$) these equations reduce to the familiar
discrete-diffusion form
$$
C\,\dot V_j(t) = \frac{1}{R}\big(V_{j-1}(t)-2V_j(t)+V_{j+1}(t)\big),\qquad
j=1,\dots,N-1,
$$
together with the terminal condition
$$
C\,\dot V_N(t)=\frac{1}{R}\big(V_{N-1}(t)-V_N(t)\big).
$$

Writing $\vec V(t)=(V_1(t),\dots,V_N(t))^{\mathsf T}$, the ladder is a finite-dimensional
linear time-invariant system
$$
\dot {\vec V}(t)=A\,\vec V(t)+\vec b\,u(t),
\qquad
\Phi(t)=\vec e_N^{\mathsf T}\vec V(t),
$$
where $\vec e_N$ is the $N$th standard basis vector, $\vec b=(1/(R_1C_1),0,\dots,0)^{\mathsf T}$,
and $A$ is tridiagonal with entries determined by $\{R_j,C_j\}$ (explicitly:
$A_{11}=-(1/(R_1C_1)+1/(R_2C_1))$, $A_{12}=1/(R_2C_1)$; for $2\le j\le N-1$,
$A_{j,j-1}=1/(R_jC_j)$, $A_{j,j}=-(1/(R_jC_j)+1/(R_{j+1}C_j))$, $A_{j,j+1}=1/(R_{j+1}C_j)$;
and $A_{N,N-1}=1/(R_NC_N)$, $A_{N,N}=-1/(R_NC_N)$).

At this point, it is obvious to see that for any locally integrable input $u$ and any initial condition $V(0)$, the unique
solution is given by variation of constants \cite{caravelli2},
$$
V(t)=e^{At}V(0)+\int_{0}^{t}e^{A(t-s)}\,b\,u(s)\,ds,
$$
and therefore the readout node satisfies
$$
\Phi(t)=\vec e_N^{\mathsf T}e^{At}\vec V(0)+\int_{0}^{t}K_N(t-s)\,u(s)\,ds,
\qquad
K_N(t):=\vec e_N^{\mathsf T}e^{At}\vec b.
$$
Thus the input--output map $u\mapsto \Phi$ is a causal convolution with kernel
$K_N$, up to an exponentially decaying transient term determined by $V(0)$. In
particular, if one initializes the ladder in its steady state for $u\equiv 0$
(i.e.\ $V(0)=0$), then $\Phi(t)=(K_N*u)(t)$ exactly.

Taking the Laplace transform for $\Re(s)>0$ gives
$$
(sI-A)V(s)=V(0)+b\,u(s),
\qquad
\Phi(s)=\vec e_N^{\mathsf T}\vec V(s),
$$
hence
$$
\Phi(s)=e_N^{\mathsf T}(sI-A)^{-1}\vec V(0)+G_N(s)\,u(s),
\qquad
G_N(s):=\frac{\Phi(s)}{u(s)}=\vec e_N^{\mathsf T}(sI-A)^{-1}\vec b.
$$
The scalar transfer function $G_N(s)$ is a real rational function whose poles are
the eigenvalues of $A$. Since the ladder is passive, all eigenvalues satisfy
$\Re(\lambda)<0$, so $G_N$ is stable (analytic for $\Re(s)>0$). For RC ladders,
one in fact has the stronger property that the poles lie on the negative real
axis (and generically are simple) \cite{FialkowGerst1951,Magos1970}. Equivalently,
the eigenvalues of $-A$ are real and strictly positive.

Assume for simplicity that $A$ is diagonalizable with distinct eigenvalues
$\lambda_k\in\mathbb R$ ($k=1,\dots,N$), with $\lambda_k<0$. Writing
$\nu_k:=-\lambda_k>0$, the transfer function admits a partial-fraction expansion
$$
G_N(s)=\sum_{k=1}^{N}\frac{\alpha_k}{s+\nu_k},
$$
where the residues $\alpha_k$ depend on the overlaps of the input and readout
vectors with the eigenmodes of $A$. Inverting the Laplace transform yields the
impulse response as a finite sum of decaying exponentials,
$$
K_N(t)=\sum_{k=1}^{N}\alpha_k\,e^{-\nu_k t}\,\Theta(t),
$$
so the realized field is
$$
\Phi(t)=\int_{0}^{t}\Big(\sum_{k=1}^{N}\alpha_k e^{-\nu_k(t-s)}\Big)u(s)\,ds.
$$
This is the precise sense in which a finite RC ladder produces an exponential-mode
kernel: its poles $\{-\nu_k\}$ are real and negative, and each pole contributes one
relaxation mode. In the language of the main text, one can define auxiliary modes
$$
\Phi_k(t):=\int_{0}^{t}\alpha_k e^{-\nu_k(t-s)}u(s)\,ds,
\qquad
\Phi(t)=\sum_{k=1}^{N}\Phi_k(t),
$$
which satisfy the time-local realization
$$
\dot\Phi_k(t)=-\nu_k\Phi_k(t)+\alpha_k\,u(t),
\qquad
k=1,\dots,N.
$$

The same result follows without partial fractions by expanding the matrix exponential.
If $A=Q\Lambda Q^{-1}$ with $\Lambda=\mathrm{diag}(\lambda_1,\dots,\lambda_N)$, then
$$
K_N(t)=\vec e_N^{\mathsf T}Q\,e^{\Lambda t}\,Q^{-1}\vec b
      =\sum_{k=1}^{N}\Big(\vec e_N^{\mathsf T}q_k\Big)\Big(\tilde q_k^{\mathsf T}\vec b\Big)
        e^{\lambda_k t}\,\Theta(t),
$$
where $q_k$ and $\tilde q_k$ are right and left eigenvectors. For symmetric ladders
(e.g.\ uniform $R,C$), one can take $Q$ orthogonal so that $\tilde q_k=q_k$.

On the imaginary axis $s=i\omega$, the steady-state sinusoidal response is governed by
$G_N(i\omega)$, whose magnitude and phase encode attenuation and lag. The negative real
pole structure implies a monotone low-pass character and guarantees that $K_N(t)$ is a
linear combination of decaying exponentials. This pole/impulse-response correspondence
is the direct bridge from circuit realizability to the exponential-mode kernels used to
embed the control memory as a finite set of auxiliary ODEs in the main text.

The input--output
transfer function is
$$
G_N(s):=\frac{\Phi(s)}{u(s)}=e_N^{\mathsf T}(sI-A)^{-1}b.
$$
Hence the poles of $G_N$ are the eigenvalues of $A$. For a passive RC ladder the
dynamics is stable, so these poles lie in $\mathrm{Re}(s)<0$, and for the ladder
families of interest they occur on the negative real axis; the associated
realizability and synthesis constraints are classical and are treated in detail
in \cite{FialkowGerst1951} and from a constructive synthesis viewpoint in
\cite{Magos1970}.

The corresponding causal impulse response kernel is
$$
K_N(t)=e_N^{\mathsf T}e^{At}b\,\Theta(t),
$$
and diagonalizing $A$ yields an explicit exponential-mode decomposition,
$$
K_N(t)=\sum_{k=1}^{N}\alpha_k\,e^{-\nu_k t}\,\Theta(t),
\qquad \nu_k>0,
$$
where $\{-\nu_k\}$ are the eigenvalues of $A$ and the weights $\alpha_k$ are fixed by
the overlaps of the input vector $b$ and the readout vector $e_N$ with the
corresponding eigenvectors. Note that $\alpha_k=\Big(e_N^{\mathsf T}q_k\Big)\Big(\tilde q_k^{\mathsf T}b\Big)$. Consequently, the delivered voltage is a causal
convolution $\Phi(t)=(K_N*u)(t)$ and therefore admits the finite-dimensional
time-local realization used in the main text:
$$
\dot\Phi_k(t)=-\nu_k\Phi_k(t)+\alpha_k u(t),
\qquad
\Phi(t)=\sum_{k=1}^{N}\Phi_k(t).
$$
For a finite ladder this representation is exact: it is simply the modal form of
the ladder state-space dynamics.

The longest memory time is governed by the slowest decay rate,
$$
\tau_{\mathrm{mem}}\sim \nu_{\min}^{-1},
\qquad
\nu_{\min}:=\min_k \nu_k.
$$
For the uniform ladder one recovers diffusive scaling of the slowest mode,
$\nu_{\min}\propto (RC)^{-1}N^{-2}$ for large $N$, so adding ladder stages increases
the longest memory time quadratically in $N$. In this precise sense, a long line
naturally produces a hierarchy of relaxation rates rather than a single RC constant,
and the single-pole (single-RC) model corresponds to retaining only the slowest
relaxation channel.

Let us now discuss the diffusive regime. A long ladder with many small sections approaches a distributed RC line. In that
limit, voltage propagation along the line is governed by a diffusion equation,
and the transfer function from the input to a point a distance $x$ away is no
longer rational. Nevertheless, it remains passive and stable, and its impulse
response is still completely monotone: it can be written as a nonnegative mixture
of decaying exponentials.

Concretely, for a distributed line with resistance $r$ and capacitance $c$ per
unit length, the voltage $V(x,t)$ satisfies
\begin{eqnarray}
\partial_t V(x,t)=D\,\partial_x^2 V(x,t),
\qquad
D:=\frac{1}{rc}.\label{eq:Vreadout}
\end{eqnarray}
For boundary drive $V(0,t)=u(t)$ and a readout at position $x=L$, the input--output
relation is again a causal convolution
$\Phi(t)=\int_{-\infty}^t K_L(t-s)u(s)\,ds$.
In the Laplace domain one finds
$$
G_L(s)=\frac{\Phi(s)}{u(s)}=\exp\!\big(-L\sqrt{s/D}\big),
$$
which has no right-half-plane singularities and is strictly decaying with frequency.
The corresponding time-domain kernel has a long tail (slower than an exponential),
reflecting the continuum of diffusive time scales in the line. Equivalently,
$G_L(s)$ admits a Laplace--Stieltjes representation
$$
G_L(s)=\int_{0}^{\infty}\frac{d\mu_L(\nu)}{s+\nu},
\qquad d\mu_L(\nu)\ge 0,
$$
and therefore
$$
K_L(t)=\int_{0}^{\infty}e^{-\nu t}\,d\mu_L(\nu).
$$
This shows that the ``sum of exponentials'' representation used in the main text
is the discrete counterpart of a more general statement: passive RC channels
produce kernels that are nonnegative mixtures of exponentials. A finite auxiliary-mode
model corresponds to approximating the measure $\mu_L$ by a finite quadrature rule,
i.e.\ replacing the continuum by finitely many effective relaxation rates $\{\nu_k\}$
with weights $\{c_k\}$.

The control-channel model in the main text assumes that the realized field is generated
by a passive, stable filter with a hierarchy of relaxation rates,
$$
K(t)\approx \sum_{k=1}^{K_{\max}} c_k e^{-\nu_k t}\,\Theta(t).
$$
For a finite ladder, this is exact after lumped-element modeling; for a long distributed
line, it is a controlled approximation obtained by rational (modal) approximation of the
transfer function. The single-RC case corresponds to retaining only the slowest relaxation
rate and is therefore the minimal model that captures a nontrivial memory time. Finally,
the nomenclature ``mode'' is literal in the ladder picture: the rates $\nu_k$ are the decay
rates of the ladder's internal linear modes, i.e.\ the channels through which the line stores
and releases past drive history \cite{FialkowGerst1951,Magos1970}.

It is worth noting that the same passive RC-line model also applies to the
\emph{readout} (or more generally the wiring-mediated observation) of signals
generated at a quantum device node. In the distributed (diffusive) limit, the
voltage along a lossy line with resistance $r$ and capacitance $c$ per unit
length satisfies eqn.  \eqref{eq:Vreadout}
with a boundary drive (or device-generated signal) imposed at $x=0$ and a
readout taken at $x=l$. This defines a linear, causal input--output relation
between the waveform $u(t):=V(0,t)$ and the measured voltage
$\Phi(t):=V(l,t)$.

Passing to the Laplace domain (with $\Re(s)>0$), the diffusion equation gives
$$
V(x,s)=V(0,s)\,\exp\!\Big(-x\sqrt{\frac{s}{D}}\Big),
\qquad\Rightarrow\qquad
G_l(s):=\frac{\Phi(s)}{u(s)}=\exp\!\Big(-l\sqrt{\frac{s}{D}}\Big).
$$
Therefore the readout is a convolution $\Phi=K_l*u$ with an impulse response
$K_l(t)$ defined by the inverse Laplace transform of $G_l(s)$. Using the
standard inversion formula for $\exp(-a\sqrt{s})$, one obtains the explicit
kernel
$$
K_l(t)=\frac{l}{2\sqrt{\pi D}}\;t^{-3/2}\exp\!\Big(-\frac{l^2}{4Dt}\Big)\,\Theta(t).
$$
This kernel is strictly causal and normalized (it has unit DC gain), but it is
\emph{not} concentrated at a single time: the diffusive line does not behave
as an ideal delay element. Instead, it realizes a broad distribution of
effective delays, reflecting the continuum of internal relaxation time scales.

A convenient operational notion of the readout time is obtained from the step
response. For a step input $u(t)=u_0\,\Theta(t)$ one finds
$$
\Phi(t)=u_0\,\mathrm{erfc}\!\Big(\frac{l}{2\sqrt{Dt}}\Big),
$$
so the time to reach a fraction $\eta\in(0,1)$ of the final value satisfies
$$
\Phi(\tau_\eta)=\eta u_0
\qquad\Longleftrightarrow\qquad
\tau_\eta=\frac{l^2}{4D\,\big(\mathrm{erfc}^{-1}(\eta)\big)^2}.
$$
In particular, any reasonable “readout/settling time” definition scales as
$$
\tau_{\rm ro}\sim \frac{l^2}{D}=rc\,l^2=(rl)(cl)=R_{\rm tot}C_{\rm tot},
$$
up to an $O(1)$ factor that depends on the chosen threshold (e.g.\ $50\%$,
$90\%$, or a $10$--$90\%$ rise time).

Finally, it is sometimes useful (but only as an approximation) to describe the
diffusive kernel as an effective delay for band-limited waveforms. Writing the
frequency response $G_l(i\omega)$, its phase implies a frequency-dependent
group delay
$$
\tau_{\rm g}(\omega):=-\frac{d}{d\omega}\arg G_l(i\omega)
= \frac{l}{2\sqrt{2D\,\omega}},
$$
so around a narrow carrier band one may heuristically view the readout as an
attenuated, delayed version of the drive together with additional smoothing.
The exact time-domain description, however, remains the causal convolution with
$K_l(t)$ above rather than a delta-function delay.

\section{Single-pole bandwidth, time constant, and rise-time conversions}
\label{app:bandwidth_tau_risetime}
Let us briefly describe how to model the controller using the specifications from the classical controller device sitting in the laboratory - e.g. how to turn this manuscript into a practical rule. We assume for simplicity here that this is classical, i.e. finding the specifications of the RC ladder channel and the dominant contribution, in which the kernel can be written as $K(\tau)=\tau_c^{-1} e^{-\tau/\tau_c}$. Here we use the make the assumption that there is a unique RC channel. In this case, the engineering literature comes to help to estimate the time $\tau_c$ \cite{TektronixBW035,KeysightScopeRiseTimeBandwidth,KeysightScopeRiseTimeBandwidth}
. 

In practice the control channel is often characterized \emph{in situ} by a small-signal
frequency response (e.g.\ a measured transfer function magnitude or a $3\,\mathrm{dB}$
bandwidth), whereas in our model the same channel is parameterized in the time domain by
a memory scale $\tau_c$ (for instance, via an exponential kernel or an equivalent
first-order ODE). When the small-signal map from the commanded waveform $u(t)$ to the
realized field $\Phi(t)$ is well described by a \emph{dominant single pole}, these
characterizations are equivalent and one can translate between them without ambiguity.

A first-order low-pass response with unit DC gain has transfer function
$$
H(i\omega)=\frac{1}{1+i\omega\tau_c},
$$
where $\tau_c>0$ is the characteristic time constant. The magnitude is
$$
|H(i\omega)|=\frac{1}{\sqrt{1+(\omega\tau_c)^2}}.
$$
The $3\,\mathrm{dB}$ cutoff is defined as the frequency where the \emph{power} has dropped
by a factor of two relative to DC. Equivalently, the magnitude has dropped to
$1/\sqrt{2}$ of its DC value:
$$
|H(i\omega_{3\mathrm{dB}})|=\frac{1}{\sqrt{2}}.
$$
Substituting the magnitude formula gives
$$
\frac{1}{\sqrt{1+(\omega_{3\mathrm{dB}}\tau_c)^2}}=\frac{1}{\sqrt{2}}
\quad\Longrightarrow\quad
1+(\omega_{3\mathrm{dB}}\tau_c)^2=2
\quad\Longrightarrow\quad
\omega_{3\mathrm{dB}}\tau_c=1.
$$
Thus
$$
\omega_{3\mathrm{dB}}=\frac{1}{\tau_c},
\qquad
f_{\mathrm{3dB}}=\frac{\omega_{3\mathrm{dB}}}{2\pi}=\frac{1}{2\pi\tau_c},
\qquad
\tau_c=\frac{1}{2\pi f_{\mathrm{3dB}}}.
$$
This identity is exact for the single-pole model. In more complicated lines with multiple
poles (or weak resonances), one may still report an \emph{effective} $\tau_c$ by fitting a
dominant-pole approximation over the relevant bandwidth.

The same first-order system has a standard time-domain interpretation. Let the input be
a unit step, and denote by $y(t)$ the normalized output (so that $y(\infty)=1$). The
first-order low-pass step response is
$$
y(t)=1-e^{-t/\tau_c}.
$$
As stated in the text, we define the $10\%$ and $90\%$ times, $t_{10}$ and $t_{90}$, by
$y(t_{10})=0.1$ and $y(t_{90})=0.9$. Solving gives
$$
0.1=1-e^{-t_{10}/\tau_c}
\quad\Longrightarrow\quad
e^{-t_{10}/\tau_c}=0.9
\quad\Longrightarrow\quad
t_{10}=\tau_c\ln\!\Big(\frac{1}{0.9}\Big),
$$
and
$$
0.9=1-e^{-t_{90}/\tau_c}
\quad\Longrightarrow\quad
e^{-t_{90}/\tau_c}=0.1
\quad\Longrightarrow\quad
t_{90}=\tau_c\ln\!\Big(\frac{1}{0.1}\Big).
$$
Therefore the $10\%$--$90\%$ rise time is
$$
t_r=t_{90}-t_{10}
=\tau_c\left[\ln\!\Big(\frac{1}{0.1}\Big)-\ln\!\Big(\frac{1}{0.9}\Big)\right]
=\tau_c\ln\!\Big(\frac{0.9}{0.1}\Big)
=\tau_c\ln 9
\approx 2.197\,\tau_c.
$$

If the channel is approximately single-pole, it is standard to identify the bandwidth
${\mathcal B}_w$ (in hertz) with the $3\,\mathrm{dB}$ cutoff,
$$
{\mathcal B}_w \approx f_{\mathrm{3dB}}=\frac{1}{2\pi\tau_c}.
$$
Combining this with $t_r\approx (\ln 9)\tau_c$ yields
$$
t_r \approx (\ln 9)\tau_c
= \frac{\ln 9}{2\pi}\,\frac{1}{{\mathcal B}_w}
\approx \frac{2.197}{2\pi}\,\frac{1}{{\mathcal B}_w}
\approx \frac{0.349}{{\mathcal B}_w}
\approx \frac{0.35}{{\mathcal B}_w}.
$$
This is the origin of the widely used bandwidth--rise-time conversion for first-order
limited responses in the engineering context, which we also use in the text. In our context it provides a direct translation between measurable
small-signal characterizations of the line (a $3\,\mathrm{dB}$ bandwidth or step/impulse
response) and the effective memory scale $\tau_c$ used to parameterize the kernel or its
ODE embedding in the main text.

The relations above are exact for a true first-order response. For transfer functions
with multiple comparable poles/zeros, the $3\,\mathrm{dB}$ point and the $10\%$--$90\%$
rise time generally do not correspond to a unique $\tau_c$; nevertheless, over a limited
frequency range one can often fit a dominant pole and use the same conversions to report
an effective memory scale appropriate for the drive waveforms considered in this work.
The specifications provided in the literature suggest that 
\cite{TektronixBW035,KeysightScopeRiseTimeBandwidth,KeysightScopeRiseTimeBandwidth} $\tau_c$ is in the range of $500ps-1500$ ps for a scope in the Ghz regime. 

\section{Quantum control channel: Kubo--Born--Oppenheimer derivation of a filtered control Hamiltonian (with weak dissipation)}
\label{app:kubo_born_filtered_control}

This appendix derives, from a microscopic quantum model, the effective description used
throughout the paper in which the device is driven by a \emph{realized} classical field
$\Phi(t)$ that is a causal convolution of the commanded waveform $u(t)$.
The key steps are:
(i) model the control hardware as a (large) quantum ``channel'' driven at an input port by
a classical source $u(t)$;
(ii) identify the delivered device-node field with the expectation value of an output-port
operator $\hat B$;
(iii) use linear response (Kubo) to obtain $\Phi(t)=(K*u)(t)$ with $K$ a retarded
susceptibility; and
(iv) under weak device--channel coupling, show that the device evolves under an effective
Hamiltonian $\hat H_{\rm eff}(t)=\hat H_A+\Phi(t)\hat M$, with any Lindblad contribution
appearing as a separate (weak) correction generated by channel fluctuations.

\subsection{Closed system case (no dissipation)}
\label{app:kubo_phi_derivation}
Let $\mathcal H_A$ be the device Hilbert space (a qubit in the main text) and
$\mathcal H_B$ the Hilbert space of the control channel (wiring/filter stack). We assume a
classical command $u(t)$ acts at the \emph{input port} of the channel as a c-number source,
and the device couples to the \emph{output port} field at the chip.

We take the total Hamiltonian (setting $\hbar=1$) to be
\begin{equation}
    \hat H(t)
    =
    \hat H_A
    +
    \hat H_B
    +
    g\,\hat M\otimes \hat L
    \;-\;
    u(t)\,\mathbb I_A\otimes \hat F,
    \label{eq:H_total_kubo_born}
\end{equation}
where:
 $\hat F$ is the channel operator conjugate to the applied input waveform $u(t)$
(generalized force at the boundary);
 $\hat L$ is the channel operator representing the delivered field at the device node
(generalized coordinate at the output port);
$\hat M$ is the device coupling operator (e.g.\ $\sigma_x$ or $\sigma_z$); and
 $g$ quantifies device loading of the channel.

We assume a bipartite Hilbert space $\mathcal H=\mathcal H_A\otimes\mathcal H_B$ and an
initial product state at time $t_0$,
\begin{equation}
    \hat\rho(t_0)=\hat\rho_A(t_0)\otimes \hat\rho_B,
    \qquad [\hat\rho_B,\hat H_B]=0,
    \label{eq:rho0_factorized}
\end{equation}
with exact unitary evolution
\begin{equation}
    \hat\rho(t)=\hat U(t,t_0)\,\hat\rho(t_0)\,\hat U^\dagger(t,t_0),
    \qquad
    \hat U(t,t_0)=\mathcal T\exp\!\Big(-i\!\int_{t_0}^{t}\hat H(s)\,ds\Big).
    \label{eq:global_unitary}
\end{equation}
The reduced state is $\hat\rho_A(t):=\mathrm{Tr}_B[\hat\rho(t)]$.

\smallskip
\noindent
To derive $\dot{\hat\rho}_A(t)$, we will use only standard identities for the partial trace.
Fix an orthonormal basis $\{|n\rangle\}$ of $\mathcal H_B$. By definition,
\begin{equation}
    \mathrm{Tr}_B[\hat X]
    :=
    \sum_n (\mathbb I_A\otimes\langle n|)\,\hat X\,(\mathbb I_A\otimes |n\rangle),
    \qquad \hat X\in\mathcal B(\mathcal H_A\otimes\mathcal H_B).
    \label{eq:partial_trace_def}
\end{equation}
From \eqref{eq:partial_trace_def} one immediately checks the following identities
(for $\hat A$ acting on $\mathcal H_A$, $\hat Z$ acting on $\mathcal H_B$, and
$\hat X$ arbitrary on $\mathcal H_A\otimes\mathcal H_B$):
\begin{align}
    \mathrm{Tr}_B\!\big[(\hat A\otimes \mathbb I_B)\hat X\big]
    &= \hat A\,\mathrm{Tr}_B[\hat X],
    \label{eq:ptrace_left_A}
    \\
    \mathrm{Tr}_B\!\big[\hat X(\hat A\otimes \mathbb I_B)\big]
    &= \mathrm{Tr}_B[\hat X]\,\hat A,
    \label{eq:ptrace_right_A}
    \\
    \mathrm{Tr}_B\!\big[(\mathbb I_A\otimes \hat Z)\hat X\big]
    &= \mathrm{Tr}_B\!\big[\hat X(\mathbb I_A\otimes \hat Z)\big].
    \label{eq:ptrace_cyclic_B}
\end{align}
(Identity \eqref{eq:ptrace_cyclic_B} is the “cyclicity” of the trace on subsystem $B$ for operators acting
\emph{only} on $B$. It follows from \eqref{eq:partial_trace_def} by inserting a resolution of the identity
on $\mathcal H_B$ and relabeling dummy indices.)
A useful corollary of \eqref{eq:ptrace_cyclic_B} is
\begin{equation}
    \mathrm{Tr}_B\!\big([\mathbb I_A\otimes \hat Z,\hat X]\big)=0
    \qquad \text{for all }\hat Z\in\mathcal B(\mathcal H_B),\ \hat X\in\mathcal B(\mathcal H_A\otimes\mathcal H_B).
    \label{eq:ptrace_commutator_B_zero}
\end{equation}

\smallskip
\noindent
Now use the Liouville--von Neumann equation $\dot{\hat\rho}(t)=-i[\hat H(t),\hat\rho(t)]$ and
differentiate under the partial trace, we get
\begin{equation}
    \dot{\hat\rho}_A(t)
    =
    \mathrm{Tr}_B[\dot{\hat\rho}(t)]
    =
    -i\,\mathrm{Tr}_B\!\big([\hat H(t),\hat\rho(t)]\big).
    \label{eq:dotrhoA_start}
\end{equation}
Assume the usual decomposition 
\begin{equation}
    \hat H(t)
    =
    \hat H_A\otimes \mathbb I_B
    +\mathbb I_A\otimes \hat H_B
    +g\,\hat M\otimes \hat L
    \;+\; \text{(possible drive terms acting only on $B$)}.
    \label{eq:H_decomp_for_derivation}
\end{equation}
We now insert \eqref{eq:H_decomp_for_derivation} into \eqref{eq:dotrhoA_start} and treat each contribution.

Using \eqref{eq:ptrace_left_A}--\eqref{eq:ptrace_right_A},
\begin{align}
    \mathrm{Tr}_B\!\big([\hat H_A\otimes \mathbb I_B,\hat\rho]\big)
    &=
    \mathrm{Tr}_B\!\big((\hat H_A\otimes \mathbb I_B)\hat\rho\big)
    -\mathrm{Tr}_B\!\big(\hat\rho(\hat H_A\otimes \mathbb I_B)\big)
    \nonumber\\
    &=
    \hat H_A\,\mathrm{Tr}_B[\hat\rho]
    -\mathrm{Tr}_B[\hat\rho]\,\hat H_A
    \nonumber\\
    &=
    [\hat H_A,\hat\rho_A].
    \label{eq:ptrace_HA_term}
\end{align}

For any operator $\hat Z$ acting only on $B$, \eqref{eq:ptrace_commutator_B_zero} gives
\begin{equation}
    \mathrm{Tr}_B\!\big([\mathbb I_A\otimes \hat Z,\hat\rho]\big)=0.
    \label{eq:ptrace_B_terms_zero}
\end{equation}
Hence $\mathbb I_A\otimes \hat H_B$ (and likewise any classical-drive term of the form
$\mathbb I_A\otimes \hat Z(t)$) does not contribute \emph{directly} to $\dot{\hat\rho}_A$ at this stage;
it affects the reduced dynamics only indirectly through the full state $\hat\rho(t)$.

We compute the partial trace of the commutator $ [\hat M\otimes \hat L,\hat\rho]$ explicitly:
\begin{align}
    \mathrm{Tr}_B\!\big([\hat M\otimes \hat L,\hat\rho]\big)
    &=
    \mathrm{Tr}_B\!\big((\hat M\otimes \hat L)\hat\rho\big)
    -\mathrm{Tr}_B\!\big(\hat\rho(\hat M\otimes \hat L)\big).
    \label{eq:interaction_comm_start}
\end{align}
For the first term, pull $\hat M$ out using \eqref{eq:ptrace_left_A}:
\begin{align}
    \mathrm{Tr}_B\!\big((\hat M\otimes \hat L)\hat\rho\big)
    &=
    \mathrm{Tr}_B\!\big((\hat M\otimes \mathbb I_B)(\mathbb I_A\otimes \hat L)\hat\rho\big)
    \nonumber\\
    &=
    \hat M\,\mathrm{Tr}_B\!\big((\mathbb I_A\otimes \hat L)\hat\rho\big).
    \label{eq:interaction_first_term}
\end{align}
For the second term, first use cyclicity on subsystem $B$, \eqref{eq:ptrace_cyclic_B}, to move
$(\mathbb I_A\otimes \hat L)$ to the left, and then pull $\hat M$ out on the right using
\eqref{eq:ptrace_right_A}:
\begin{align}
    \mathrm{Tr}_B\!\big(\hat\rho(\hat M\otimes \hat L)\big)
    &=
    \mathrm{Tr}_B\!\big(\hat\rho(\hat M\otimes \mathbb I_B)(\mathbb I_A\otimes \hat L)\big)
    \nonumber\\
    &=
    \mathrm{Tr}_B\!\big((\mathbb I_A\otimes \hat L)\hat\rho(\hat M\otimes \mathbb I_B)\big)
    \qquad\text{(by \eqref{eq:ptrace_cyclic_B})}
    \nonumber\\
    &=
    \mathrm{Tr}_B\!\big((\mathbb I_A\otimes \hat L)\hat\rho\big)\,\hat M
    \qquad\text{(by \eqref{eq:ptrace_right_A})}.
    \label{eq:interaction_second_term}
\end{align}
Subtracting \eqref{eq:interaction_second_term} from \eqref{eq:interaction_first_term} yields
\begin{equation}
    \mathrm{Tr}_B\!\big([\hat M\otimes \hat L,\hat\rho]\big)
    =
    \big[\hat M,\ \mathrm{Tr}_B\!\big((\mathbb I_A\otimes \hat L)\hat\rho\big)\big].
    \label{eq:interaction_comm_result}
\end{equation}

Combining \eqref{eq:dotrhoA_start}, \eqref{eq:ptrace_HA_term}, \eqref{eq:ptrace_B_terms_zero},
and \eqref{eq:interaction_comm_result} gives the exact identity
\begin{equation}
    \dot{\hat\rho}_A(t)
    =
    -i[\hat H_A,\hat\rho_A(t)]
    -ig\,
    \big[\hat M,\ \mathrm{Tr}_B\!\big((\mathbb I_A\otimes \hat L)\hat\rho(t)\big)\big].
    \label{eq:exact_reduced_identity_stepwise}
\end{equation}
The nontrivial object is the operator-valued ``field''
\begin{equation}
    \hat X_L(t):=\mathrm{Tr}_B\!\big((\mathbb I_A\otimes \hat L)\hat\rho(t)\big)\in\mathcal B(\mathcal H_A),
    \label{eq:operator_valued_field_def}
\end{equation}
which depends on the full joint state $\hat\rho(t)$ and therefore, in general, on
device--channel correlations.

We now impose the common but physically motivated \emph{weak-loading} assumption \cite{BreuerPetruccione2002}, or Born-Oppenheimer approximation.
There exists a regime in which the device perturbs the channel only weakly, so that to leading order
\begin{equation}
\hat\rho(t)=\hat\rho_A(t)\otimes \hat\rho_B^{(u)}(t) + \mathcal O(g),
    \label{eq:born_ansatz}
\end{equation}
where $\hat\rho_B^{(u)}(t)$ is the channel state produced by the drive $u(t)$ \emph{in the absence
of the device}, i.e.\ under the driven channel Hamiltonian
\begin{equation}
    \hat H_B^{(u)}(t)=\hat H_B - u(t)\hat F.
    \label{eq:H_B_driven}
\end{equation}

Under \eqref{eq:born_ansatz},
\begin{equation}
    \mathrm{Tr}_B(\hat B\,\hat\rho(t))
    =
    \hat\rho_A(t)\,\Phi(t) + \mathcal O(g),
    \qquad
    \Phi(t):=\mathrm{Tr}_B\!\big(\hat\rho_B^{(u)}(t)\,\hat L\big),
    \label{eq:Phi_def_born}
\end{equation}
and \eqref{eq:dotrhoA_start}becomes, to leading order in $g$,
\begin{equation}
    \dot{\hat\rho}_A(t)
    =
    -i\big[\hat H_A + g\,\Phi(t)\,\hat M,\;\hat\rho_A(t)\big]
    +\mathcal O(g^2).
    \label{eq:unitary_effective_H}
\end{equation}
Thus, at the Hamiltonian level, the device is driven by the c-number realized field $\Phi(t)$
generated by the channel.
We now compute $\Phi(t)$ from \eqref{eq:Phi_def_born} in \emph{linear response} in the source $u$,
being explicit about which time each density operator refers to.

Fix an initial time $t_0$ at which the drive is specified and the channel is prepared in a stationary
reference state for the unforced Hamiltonian $\hat H_B$:
\begin{equation}
\hat\rho_B^{(u)}(t_0)=\hat\rho_B,
\qquad
[\hat\rho_B,\hat H_B]=0.
\label{eq:rhoB_initial_stationary}
\end{equation}
Here $\hat\rho_B$ is time-independent in the Schr\"odinger picture under $\hat H_B$, but it is the
\emph{initial condition} for the driven evolution under $\hat H_B-u(t)\hat F$.

We work in the interaction picture with respect to $\hat H_B$:
\begin{equation}
\hat O_I(t)=e^{i\hat H_B(t-t_0)}\hat O\,e^{-i\hat H_B(t-t_0)}.
\label{eq:IP_def_channel}
\end{equation}
In particular,
\begin{equation}
\hat F_I(t)=e^{i\hat H_B(t-t_0)}\hat F\,e^{-i\hat H_B(t-t_0)},
\qquad
\hat L_I(t)=e^{i\hat H_B(t-t_0)}\hat L\,e^{-i\hat H_B(t-t_0)}.
\label{eq:FI_LI_defs}
\end{equation}

Let $\hat\rho_B^{(u)}(t)$ denote the channel density operator at \emph{observation time} $t$ in the
Schr\"odinger picture, evolved from the initial condition \eqref{eq:rhoB_initial_stationary} under the
time-dependent Hamiltonian $\hat H_B-u(t)\hat F$.
Define its interaction-picture counterpart
\begin{equation}
\hat\rho_I(t)=e^{i\hat H_B(t-t_0)}\hat\rho_B^{(u)}(t)\,e^{-i\hat H_B(t-t_0)}.
\label{eq:rhoI_def}
\end{equation}
Then $\hat\rho_I(t_0)=\hat\rho_B$ and $\hat\rho_I(t)$ satisfies
\begin{equation}
\frac{d}{dt}\hat\rho_I(t)=i\,u(t)\,[\hat F_I(t),\hat\rho_I(t)].
\label{eq:rhoI_eom}
\end{equation}
Integrating \eqref{eq:rhoI_eom} from $t_0$ to $t$ and expanding to first order in $u$ gives
\begin{equation}
\hat\rho_I(t)=\hat\rho_B
+i\int_{t_0}^{t} ds\;u(s)\,[\hat F_I(s),\hat\rho_B]
+\mathcal O(u^2).
\label{eq:rhoI_linear}
\end{equation}
Transforming back to the Schr\"odinger picture at time $t$,
\begin{equation}
\hat\rho_B^{(u)}(t)
=e^{-i\hat H_B(t-t_0)}\hat\rho_I(t)\,e^{i\hat H_B(t-t_0)}
=\hat\rho_B(t)+\delta\hat\rho_B(t)+\mathcal O(u^2),
\label{eq:rhoB_expand}
\end{equation}
where, because $[\hat\rho_B,\hat H_B]=0$,
\begin{equation}
\hat\rho_B(t):=e^{-i\hat H_B(t-t_0)}\hat\rho_B\,e^{i\hat H_B(t-t_0)}=\hat\rho_B,
\label{eq:rhoB_free_timeindep}
\end{equation}
and the first-order correction is
\begin{equation}
\delta\hat\rho_B(t)
=
i\int_{t_0}^{t} ds\;u(s)\,
e^{-i\hat H_B(t-t_0)}[\hat F_I(s),\hat\rho_B]e^{i\hat H_B(t-t_0)}.
\label{eq:delta_rhoB_S}
\end{equation}
Equivalently (and often more convenient), one can keep the correction in the interaction picture as
in \eqref{eq:rhoI_linear}.

By definition,
\begin{equation}
\Phi(t)=\mathrm{Tr}_B\!\big(\hat\rho_B^{(u)}(t)\hat L\big)
=\mathrm{Tr}_B\!\big(\hat\rho_I(t)\hat L_I(t)\big),
\label{eq:Phi_IP_identity}
\end{equation}
where the second equality follows from \eqref{eq:rhoI_def} and cyclicity of trace.
Insert the linearized state \eqref{eq:rhoI_linear}:
\begin{eqnarray}
\Phi(t)
&=&
\mathrm{Tr}_B\!\big(\hat\rho_B\hat L_I(t)\big)
+i\int_{t_0}^{t} ds\;u(s)\,
\mathrm{Tr}_B\!\big([\hat F_I(s),\hat\rho_B]\hat L_I(t)\big)
+\mathcal O(u^2)
\nonumber\\
&=&
\Phi_0
+i\int_{t_0}^{t} ds\;u(s)\,
\langle[\hat L_I(t),\hat F_I(s)]\rangle_B
+\mathcal O(u^2),
\label{eq:Phi_linear_explicit_times}
\end{eqnarray}
with
\begin{equation}
\Phi_0:=\mathrm{Tr}_B(\hat\rho_B\hat L)
=\mathrm{Tr}_B\!\big(\hat\rho_B\hat L_I(t)\big),
\qquad
\langle\cdot\rangle_B:=\mathrm{Tr}_B(\hat\rho_B\,\cdot).
\label{eq:Phi0_def}
\end{equation}
In the second line of \eqref{eq:Phi_linear_explicit_times} we used the trace identity
$\mathrm{Tr}_B([\hat F_I(s),\hat\rho_B]\hat L_I(t))=\mathrm{Tr}_B(\hat\rho_B[\hat L_I(t),\hat F_I(s)])$.

Because $[\hat\rho_B,\hat H_B]=0$, the correlator depends only on the time difference $\tau=t-s$,
so we define
\begin{equation}
\chi_{LF}(\tau)
:=
i\,\Theta(\tau)\,\langle[\hat L_I(\tau),\hat F]\rangle_B
=
i\,\Theta(\tau)\,\langle[\hat L_I(t),\hat F_I(s)]\rangle_B,
\qquad
\tau=t-s.
\label{eq:chi_retarded_explicit}
\end{equation}
Therefore,
\begin{equation}
\Phi(t)=\Phi_0+\int_{t_0}^{t} ds\;\chi_{LF}(t-s)\,u(s)+\mathcal O(u^2).
\label{eq:Phi_linear_response_explicit}
\end{equation}
Taking $t_0\to -\infty$ (with the usual adiabatic switch-on prescription if desired) yields the causal
convolution form
\begin{equation}
\Phi(t)=\Phi_0+\int_{-\infty}^{t} K(t-s)\,u(s)\,ds+\mathcal O(u^2),
\qquad
K(\tau)\equiv \chi_{LF}(\tau).
\label{eq:Phi_convolution_final_explicit}
\end{equation}

Equation \eqref{eq:Phi_convolution_final_explicit} is precisely the kernel model used in the main text: the
control-channel memory kernel is the retarded susceptibility of the driven channel between the
input-port operator $\hat F$ and the output-port operator $\hat B$.

\subsection{Closed versus open control channels: structure of the retarded kernel}
\label{app:closed_vs_open_kernel}

In this appendix we treat the \emph{control channel} $B$ as the only system whose
input--output map carries memory. The device $A$ may be treated separately (and may
remain unitary); the kernel below is defined purely in terms of $B$.
We consider a classical command $u(t)$ that couples at the \emph{input port} to a
Hermitian channel operator $\hat F$, and we define the \emph{delivered field} at the
device node by a Hermitian channel operator $\hat L$ (notation $\hat L\equiv \hat B$ in
other parts of the manuscript). In linear response the realized field
$\Phi(t):=\langle \hat L\rangle_t$ is governed by a retarded susceptibility.

We now treat the channel. Let $\hat\rho_B$ be a stationary reference state for the unforced channel dynamics.
The retarded response function is
\begin{equation}
    \chi_{LF}(\tau)
    :=
    i\,\Theta(\tau)\,\big\langle[\hat L_I(\tau),\hat F_I(0)]\big\rangle_B,
    \qquad
    \langle\cdot\rangle_B:=\mathrm{Tr}_B(\hat\rho_B\,\cdot),
    \label{eq:chi_retarded_recalled}
\end{equation}
where the interaction-picture operators $\hat O_I(t)$ are defined with respect to the
\emph{unforced} channel generator (Hamiltonian for the closed case, Liouvillian for the open case).
In frequency domain we write $\chi_{LF}(\omega)$ for the Fourier transform of
$\chi_{LF}(\tau)$ (under the usual convergence prescription).

The question addressed here is: \emph{under what conditions can $\chi_{LF}(\tau)$ be written (or well
approximated) as a sum of decaying exponentials, as assumed in the RC-like model of the main text?}
We distinguish two cases.

\subsubsection{Closed (lossless) channel: unitary kernel and Lehmann representation}
\label{app:closed_channel_lehmann}

Assume first that the channel $B$ is \emph{closed}, so its unforced dynamics is unitary
under a time-independent Hamiltonian $\hat H_B$. Take the reference state to commute
with $\hat H_B$, hence diagonal in the energy basis. Let
$\hat H_B\ket{n}=E_n\ket{n}$ and
\begin{equation}
    \hat\rho_B=\sum_n p_n \ket{n}\!\bra{n},
    \qquad p_n\ge 0,\ \sum_n p_n=1.
    \label{eq:rhoB_diagonal}
\end{equation}
Then $\hat L_I(\tau)=e^{i\hat H_B\tau}\hat L e^{-i\hat H_B\tau}$ admits the expansion
\begin{equation}
    \hat L_I(\tau)=\sum_{m,n} e^{i(E_m-E_n)\tau}\,L_{mn}\,\ket{m}\!\bra{n},
    \qquad L_{mn}:=\bra{m}\hat L\ket{n}.
    \label{eq:L_lehmann}
\end{equation}
Using \eqref{eq:rhoB_diagonal}--\eqref{eq:L_lehmann}, the commutator expectation becomes
\begin{equation}
    \big\langle[\hat L_I(\tau),\hat F]\big\rangle_B
    =
    \sum_{m,n}(p_n-p_m)\,L_{nm}F_{mn}\,e^{i(E_n-E_m)\tau},
    \qquad F_{mn}:=\bra{m}\hat F\ket{n}.
    \label{eq:comm_expect_lehmann}
\end{equation}
Therefore the retarded kernel has the Lehmann representation
\begin{equation}
    \chi_{LF}(\tau)
    =
    i\,\Theta(\tau)\sum_{m,n}(p_n-p_m)\,L_{nm}F_{mn}\,e^{i(E_n-E_m)\tau}.
    \label{eq:chi_lehmann_closed}
\end{equation}

We now list a few consequences of this description.
First, for finite-dimensional $B$, \eqref{eq:chi_lehmann_closed} is a finite sum of purely
oscillatory terms $e^{i\omega_{nm}\tau}$ with Bohr frequencies $\omega_{nm}=E_n-E_m$.
Hence $\chi_{LF}(\tau)$ does not generically decay as $\tau\to\infty$ (except for
special degeneracies or cancellations). In frequency space, the imaginary part is a
sum of delta-peaked spectral lines rather than a smooth dissipative response.
Second, the factor $\Theta(\tau)$ enforces causality. For Hermitian $\hat L,\hat F$,
$\langle[\hat L_I(\tau),\hat F]\rangle_B$ is purely imaginary, so $\chi_{LF}(\tau)$ is
real for $\tau>0$. The Fourier transform satisfies the reality condition
$\chi_{LF}(\omega)^\ast=\chi_{LF}(-\omega)$ and Kramers--Kronig relations (analyticity in
the upper half-plane).
Lastly, a closed channel can still exhibit \emph{memory} (phase lag, propagation delay, reactive
distortion), but its kernel is not generically a sum of \emph{decaying} exponentials.
Effective decay can occur only in a macroscopic limit (continuous spectrum) through
dephasing (destructive interference), in which case $\chi_{LF}(\tau)$ typically has
non-exponential tails and corresponding branch-cut structure in Laplace/Fourier space.
Thus the RC-like \emph{relaxational} kernels of the main text require either genuine
dissipation or an effective coarse-graining that produces such an irreversible semigroup.

This said, we now provide a general overview of how model dissipation in the \emph{control channel $B$} itself (resistive losses,
thermalized attenuators, leakage into many uncontrolled degrees of freedom) by treating
$B$ as an \emph{open} system with an undriven GKLS generator $\mathcal L_0$.
Let $\hat\rho_{B,\mathrm{ss}}$ be a stationary state, $\mathcal L_0\hat\rho_{B,\mathrm{ss}}=0$.
We drive the channel at the input port by the classical command $u(t)$ through the
Hamiltonian perturbation $-u(t)\hat F$, so the channel master equation is
\begin{equation}
    \dot{\hat\rho}_B(t)
    =
    \mathcal L_0\hat\rho_B(t)\;-\;i\,u(t)\,[\hat F,\hat\rho_B(t)].
    \label{eq:channel_master_driven}
\end{equation}
Define the realized field and its linear response about stationarity:
\begin{equation}
    \Phi(t):=\mathrm{Tr}_B(\hat\rho_B(t)\hat L),
    \qquad
    \delta\Phi(t):=\Phi(t)-\mathrm{Tr}_B(\hat\rho_{B,\mathrm{ss}}\hat L).
    \label{eq:Phi_deltaPhi_defs}
\end{equation}

We then write $\hat\rho_B(t)=\hat\rho_{B,\mathrm{ss}}+\delta\hat\rho_B(t)$ and keep only terms linear in $u$:
\begin{equation}
    \dot{\delta\hat\rho}_B(t)
    =
    {\mathcal L}_0\,{\delta\hat\rho}_B(t)\;-\;i\,u(t)\,[\hat F,\hat\rho_{B,\mathrm{ss}}].
    \label{eq:linearized_master}
\end{equation}
Assuming the steady-state (no-transient) regime, the causal solution is
\begin{equation}
    {\delta\hat\rho}_B(t)
    =
    \int_{-\infty}^{t} ds\;
    e^{\mathcal L_0 (t-s)}\Big(-i[\hat F,\hat\rho_{B,\mathrm{ss}}]\Big)\,u(s).
    \label{eq:delta_rho_solution}
\end{equation}
Taking the expectation of $\hat L$ yields the convolution form
\begin{equation}
    \delta\Phi(t)
    =
    \int_{-\infty}^{t} ds\;K(t-s)\,u(s),
    \qquad
    K(\tau):=
    \mathrm{Tr}_B\!\Big(\hat L\,e^{\mathcal L_0 \tau}\big(-i[\hat F,\hat\rho_{B,\mathrm{ss}}]\big)\Big),
    \quad \tau\ge 0.
    \label{eq:K_from_L0}
\end{equation}
Introduce the adjoint semigroup $e^{\mathcal L_0^\dag \tau}$ acting on observables via
$\mathrm{Tr}(\hat X\,e^{\mathcal L_0 \tau}\hat Y)=\mathrm{Tr}((e^{\mathcal L_0^\dag \tau}\hat X)\hat Y)$.
Then
\begin{equation}
    K(\tau)
    =
    i\,\mathrm{Tr}_B\!\Big(\hat\rho_{B,\mathrm{ss}}\,[\,e^{\mathcal L_0^\dag \tau}\hat L,\hat F\,]\Big),
    \qquad
    \chi_{LF}(\tau)=\Theta(\tau)\,K(\tau).
    \label{eq:chi_open_commutator}
\end{equation}
Equation \eqref{eq:chi_open_commutator} is the Markovian open-channel analogue of the unitary Kubo
formula \eqref{eq:chi_retarded_recalled}: dissipation enters through the semigroup
$e^{\mathcal L_0^\dag \tau}$.

At this point, 
we assume that, on the operator subspace relevant for the port observables, $\mathcal L_0^\dag$
admits a discrete spectral decomposition (e.g.\ after a finite-mode truncation of the channel):
\begin{equation}
    \mathcal L_0^\dag \hat R_k=\lambda_k \hat R_k,
    \qquad \mathrm{Re}(\lambda_k)\le 0,
    \label{eq:Ldag_eigs}
\end{equation}
and expand $e^{\mathcal L_0^\dag \tau}\hat L$ in this eigenbasis. Inserting into
\eqref{eq:chi_open_commutator} yields
\begin{equation}
    \chi_{LF}(\tau)
    =
    \Theta(\tau)\sum_k c_k\,e^{\lambda_k \tau},
    \label{eq:chi_mode_expansion_general}
\end{equation}
with coefficients $c_k$ determined by the overlaps of $\hat L$ and $\hat F$ with the eigenoperators
and by the stationary state $\hat\rho_{B,\mathrm{ss}}$. Stability implies that all non-steady modes
decay: $\mathrm{Re}(\lambda_k)<0$ for $k\neq 0$. In general $\lambda_k=-\nu_k+i\omega_k$, producing
damped oscillations.

The exponential (classical) kernels used in the main text correspond to a \emph{relaxational} regime in which
the relevant channel modes are overdamped and non-oscillatory in the drive band. Concretely, the
dominant eigenvalues satisfy $\lambda_k=-\nu_k$ with $\nu_k>0$, so
\begin{equation}
    \chi_{LF}(\tau)
    =
    \Theta(\tau)\sum_{k=1}^{K_{\max}} c_k e^{-\nu_k\tau},
    \qquad \nu_k>0.
    \label{eq:chi_sum_decays_rc_like}
\end{equation}
Keeping only the slowest rate $\nu_{\min}$ yields the single-pole (single-RC) approximation.

This purely relaxational spectral structure is exactly what is enforced by passive RC ladder models
of the control stack: the transfer function is stable with poles on the negative real axis, implying a
sum (finite ladder) or mixture (distributed line) of decaying exponentials (Appendix~\ref{app:rc_realizability}).
Thus the RC-like approximation employed in the main text may be read equivalently as:
(i) a classical network-synthesis statement about passive RC transfer functions; or
(ii) an open-channel statement that the effective generator $\mathcal L_0^\dag$ has a dominantly real,
negative spectrum on the relevant port-observable subspace.

\subsection{Thermal Lehmann representation of the kernel and its temperature dependence}
\label{app:thermal_lehmann_kernel}

In this subsection, we make the kernel
$K(\tau)\equiv \chi_{LF}(\tau)$ in \eqref{eq:chi_retarded_recalled} fully explicit
for a \emph{closed} control channel in a thermal state. This yields a concrete
expression for how the temperature and the channel energy spectrum determine the
linear filter seen by the commanded waveform $u(t)$.

Assume the unforced channel Hamiltonian $\hat H_B$ is time independent and the
reference state is thermal,
\begin{equation}
    \hat\rho_B=\frac{e^{-\beta \hat H_B}}{Z},
    \qquad
    Z=\Tr_B(e^{-\beta \hat H_B}),
    \qquad
    \beta=(k_B T)^{-1}.
    \label{eq:rhoB_thermal}
\end{equation}
Let $\{\ket{n}\}$ be an eigenbasis of $\hat H_B$,
\begin{equation}
    \hat H_B\ket{n}=E_n\ket{n},
    \qquad
    p_n:=\bra{n}\hat\rho_B\ket{n}=\frac{e^{-\beta E_n}}{Z}.
    \label{eq:thermal_weights}
\end{equation}
We write the interaction-picture operators (with respect to $\hat H_B$) as
$\hat O_I(t)=e^{i\hat H_B(t-t_0)}\hat O\,e^{-i\hat H_B(t-t_0)}$ and define the
output-port observable $\hat L$ (device-side field) and input-port drive operator
$\hat F$ (boundary generalized force), consistent with Sec.~\ref{app:kubo_phi_derivation}.

Stationarity implies
\begin{equation}
    K(\tau)\equiv \chi_{LF}(\tau)
    =
    i\,\Theta(\tau)\,\Tr_B\!\Big(\hat\rho_B\,[\hat L_I(\tau),\hat F_I(0)]\Big).
    \label{eq:K_retarded_thermal_def}
\end{equation}
We now expand \eqref{eq:K_retarded_thermal_def} in the energy eigenbasis.

Insert resolutions of the identity into $\hat L_I(\tau)$ and $\hat F_I(0)=\hat F$.
Writing $L_{mn}:=\bra{m}\hat L\ket{n}$ and $F_{mn}:=\bra{m}\hat F\ket{n}$, we have
\begin{equation}
    \hat L_I(\tau)
    =
    e^{i\hat H_B\tau}\hat L e^{-i\hat H_B\tau}
    =
    \sum_{m,n} e^{i(E_m-E_n)\tau}\,L_{mn}\,\ket{m}\!\bra{n}.
    \label{eq:L_I_lehmann}
\end{equation}
Similarly,
\begin{equation}
    \hat F
    =
    \sum_{m,n}F_{mn}\,\ket{m}\!\bra{n}.
    \label{eq:F_lehmann}
\end{equation}

We compute the thermal correlators. We have $\langle \hat L_I(\tau)\hat F\rangle_\beta$ and
$\langle \hat F\hat L_I(\tau)\rangle_\beta$ with
$\langle \cdot\rangle_\beta:=\Tr_B(\hat\rho_B\,\cdot)$.
Using $\hat\rho_B=\sum_n p_n\ket{n}\!\bra{n}$ from \eqref{eq:thermal_weights},
\begin{align}
    \langle \hat L_I(\tau)\hat F\rangle_\beta
    &=
    \sum_n p_n \bra{n}\hat L_I(\tau)\hat F\ket{n}
    \nonumber\\
    &=
    \sum_{n,m} p_n\, e^{i(E_n-E_m)\tau}\,L_{nm}F_{mn},
    \label{eq:thermal_LF}
    \\
    \langle \hat F\hat L_I(\tau)\rangle_\beta
    &=
    \sum_n p_n \bra{n}\hat F\hat L_I(\tau)\ket{n}
    \nonumber\\
    &=
    \sum_{n,m} p_n\, e^{i(E_m-E_n)\tau}\,F_{nm}L_{mn}
    \nonumber\\
    &=
    \sum_{n,m} p_m\, e^{i(E_n-E_m)\tau}\,L_{nm}F_{mn},
    \label{eq:thermal_FL}
\end{align}
where the last step in \eqref{eq:thermal_FL} is a relabeling $n\leftrightarrow m$.

Subtracting \eqref{eq:thermal_FL} from \eqref{eq:thermal_LF} yields
\begin{equation}
    \langle[\hat L_I(\tau),\hat F]\rangle_\beta
    =
    \sum_{n,m}(p_n-p_m)\,L_{nm}F_{mn}\,e^{i(E_n-E_m)\tau}.
    \label{eq:commutator_thermal_lehmann}
\end{equation}
Therefore the thermal kernel is
\begin{equation}
    K(\tau)
    =
    i\,\Theta(\tau)\sum_{n,m}(p_n-p_m)\,L_{nm}F_{mn}\,e^{i(E_n-E_m)\tau},
    \qquad
    p_n=\frac{e^{-\beta E_n}}{Z}.
    \label{eq:K_lehmann_thermal}
\end{equation}
This is the thermal specialization of \eqref{eq:chi_lehmann_closed}.

Using now the explicit exponential form $p_n=\frac{e^{-\beta E_n}}{Z}$, we can rewrite the weight difference as
\begin{equation}
    p_n-p_m
    =
    \frac{e^{-\beta E_m}}{Z}\Big(e^{-\beta(E_n-E_m)}-1\Big)
    =
    -\,\frac{e^{-\beta E_m}}{Z}\Big(1-e^{-\beta\omega_{nm}}\Big),
    \qquad
    \omega_{nm}:=E_n-E_m.
    \label{eq:pn_pm_detailed_balance_weight}
\end{equation}
Thus each Bohr frequency $\omega_{nm}$ is weighted by an explicit thermal factor
$1-e^{-\beta\omega_{nm}}$, encoding detailed balance.

At this point, we define the (two-sided) Fourier transform of the retarded kernel
\begin{equation}
    K(\omega):=\int_{-\infty}^{\infty} d\tau\,e^{i\omega\tau}K(\tau).
    \label{eq:K_omega_def}
\end{equation}
Using $\int_0^\infty d\tau\,e^{i(\omega-\omega_{nm})\tau}
=\pi\delta(\omega-\omega_{nm})+i\,\mathcal P\frac{1}{\omega-\omega_{nm}}$,
\eqref{eq:K_lehmann_thermal} implies
\begin{equation}
    K(\omega)
    =
    \sum_{n,m}(p_n-p_m)\,L_{nm}F_{mn}
    \Big(\mathcal P\frac{1}{\omega-\omega_{nm}} - i\pi\delta(\omega-\omega_{nm})\Big).
    \label{eq:K_omega_lehmann}
\end{equation}
The dissipative (loss) part is therefore
\begin{equation}
    \Im\,K(\omega)
    =
    -\pi \sum_{n,m}(p_n-p_m)\,L_{nm}F_{mn}\,\delta(\omega-\omega_{nm}).
    \label{eq:ImK_lehmann}
\end{equation}
It is often convenient to introduce the (non-symmetrized) dynamical correlation spectrum
\begin{equation}
    S_{LF}(\omega)
    :=
    \int_{-\infty}^{\infty} d\tau\,e^{i\omega\tau}\,
    \langle \hat L_I(\tau)\hat F\rangle_\beta
    =
    2\pi\sum_{n,m} p_n\,L_{nm}F_{mn}\,\delta(\omega-\omega_{nm}),
    \label{eq:SLF_def}
\end{equation}
for which thermal weights imply the KMS/detailed-balance relation
\begin{equation}
    S_{LF}(-\omega)=e^{-\beta\omega}\,S_{FL}(\omega).
    \label{eq:detailed_balance_S}
\end{equation}
Comparing \eqref{eq:ImK_lehmann} with \eqref{eq:SLF_def} gives the standard
fluctuation--dissipation relation between the retarded response and equilibrium fluctuations:
\begin{equation}
    \Im\,K(\omega)
    =
    -\frac{1}{2}\Big(1-e^{-\beta\omega}\Big)\,S_{LF}(\omega).
    \label{eq:FDT_ImK}
\end{equation}
Equivalently, in terms of the symmetrized spectrum
$S_{LF}^{\mathrm{sym}}(\omega):=\tfrac12(S_{LF}(\omega)+S_{FL}(-\omega))$,
\begin{equation}
    S_{LF}^{\mathrm{sym}}(\omega)
    =
    -\coth\!\Big(\frac{\beta\omega}{2}\Big)\,\Im\,K(\omega),
    \label{eq:FDT_sym}
\end{equation}
(up to the sign convention relating $K$ to $\chi$; here $K\equiv \chi_{LF}$ as in
discussed earlier).

Equations \eqref{eq:K_lehmann_thermal}--\eqref{eq:FDT_sym} show that, for a closed channel,
the kernel is determined by (i) the  frequencies $\omega_{nm}$, (ii) the port matrix
elements $L_{nm}F_{mn}$, and (iii) the thermal factors $p_n-p_m$, or equivalently the
detailed-balance weights $1-e^{-\beta\omega}$. In particular, for a finite-dimensional
closed channel, $K(\tau)$ is a quasi-periodic sum of purely oscillatory terms and does
not generically decay as $\tau\to\infty$. Thus thermal equilibrium alone does not yield
an RC-like relaxational kernel; obtaining sums of \emph{decaying} exponentials requires
either dissipation (open channel) or an effective continuum/coarse-graining limit, as
discussed above.

Let us now observe that the general kernel $K(\tau)=\chi_{LF}(\tau)$ becomes a familiar linear-response
\emph{admittance} (or more generally a \emph{transadmittance}) once we identify the
input and output operators with standard port variables.

\smallskip
\noindent

At an electrical input port, a classical voltage drive $u(t)$ couples to the port
charge operator $\hat q$ via the work term
\begin{equation}
    \hat H_{\rm drive}(t)=-u(t)\,\hat q.
    \label{eq:port_voltage_drive}
\end{equation}
Comparing with $\hat V(t)=-u(t)\hat F$, we identify the input operator as
$\hat F=\hat q$.

\smallskip
\noindent

If we choose the output observable to be the port current
$\hat I:=\dot{\hat q}=\frac{i}{\hbar}[\hat H_B,\hat q]$ (here $\hbar=1$), then the
kernel
\begin{equation}
    \chi_{I q}(\tau)
    =
    i\,\Theta(\tau)\,\langle[\hat I_I(\tau),\hat q]\rangle_\beta
    \label{eq:chi_Iq_time}
\end{equation}
governs the linear response $\delta\langle \hat I(t)\rangle=\int_{-\infty}^t ds\,
\chi_{I q}(t-s)\,u(s)$.
In the frequency domain,
\begin{equation}
    \delta\langle \hat I(\omega)\rangle = Y(\omega)\,u(\omega),
    \qquad
    Y(\omega)\equiv \chi_{I q}(\omega),
    \label{eq:Y_equals_chi_Iq}
\end{equation}
so the complex admittance $Y(\omega)$ is exactly a retarded susceptibility.

\smallskip
\noindent

Using $\hat I=\dot{\hat q}$, one may integrate by parts (or use
$\hat I(\omega)=-i\omega \hat q(\omega)$ at the level of linear response) to write
\begin{equation}
    Y(\omega)=-\,i\omega\,\chi_{q q}(\omega),
    \qquad
    \chi_{q q}(\tau)=i\,\Theta(\tau)\,\langle[\hat q_I(\tau),\hat q]\rangle_\beta.
    \label{eq:Y_from_chiqq}
\end{equation}
This is the standard Kubo identification of admittance with a retarded correlator.

\smallskip
\noindent
It follows that the cycle-averaged power delivered by a monochromatic drive
$u(t)=\Re(u_\omega e^{-i\omega t})$ is
\begin{equation}
    \overline P(\omega)=\frac{1}{2}\,\Re\,Y(\omega)\,|u_\omega|^2.
    \label{eq:power_from_admittance}
\end{equation}
For a passive channel, $\overline P(\omega)\ge 0$, hence $\Re\,Y(\omega)\ge 0$.
In Lehmann form, one sees this nonnegativity transparently: specializing
\eqref{eq:ImK_lehmann} to $L=I$, $F=q$, and using
$\hat I_{nm}=i(E_n-E_m)q_{nm}=i\omega_{nm}q_{nm}$,
the dissipative part $\Re\,Y(\omega)$ can be expressed in terms of squared matrix
elements $|q_{nm}|^2$ with thermal weights and delta functions, ensuring that the
absorption at $\omega>0$ is nonnegative for physical port couplings.

\smallskip

In thermal equilibrium, the symmetrized current noise spectrum
\begin{equation}
    S_{II}^{\rm sym}(\omega)
    :=
    \int_{-\infty}^{\infty} d\tau\,e^{i\omega\tau}\,
    \frac{1}{2}\langle\{\delta\hat I_I(\tau),\delta\hat I\}\rangle_\beta
    \label{eq:SII_sym_def}
\end{equation}
is related to the dissipative response by the fluctuation--dissipation theorem,
\begin{equation}
    S_{II}^{\rm sym}(\omega)
    =
    \omega\,\coth\!\Big(\frac{\beta\omega}{2}\Big)\,\Re\,Y(\omega),
    \label{eq:Johnson_Nyquist_quantum}
\end{equation}
(where $\hbar=1$). In the classical (high-temperature) limit $\beta\omega\ll 1$,
this reduces to the Johnson--Nyquist form
$S_{II}^{\rm sym}(\omega)\approx 2k_BT\,\Re\,Y(\omega)$.
Thus, the same susceptibility that determines the deterministic filtering of the
control waveform also fixes the equilibrium fluctuations of the channel output.

\smallskip
\noindent

More generally, if $\hat L$ is a downstream voltage/flux (node coordinate) or a
current at a different port, then $\chi_{LF}$ is a \emph{transfer} response
(transimpedance/transadmittance). In this interpretation, the kernel $K$ used in the
main text is precisely an input--output response function between an input-port drive
operator (e.g.\ $\hat q$) and an output-port readout operator (e.g.\ a downstream node
flux $\hat\varphi$ or current $\hat I$). This provides a direct bridge between the
Kubo kernel $K(\tau)$ and the usual circuit notion of frequency-dependent transfer
functions.

\subsection{Adiabatic switch-on in Kubo: from discrete lines to RC-like kernels}
\label{app:eta_broadening_rc}

The GKLS/semi-group viewpoint in Sec.~\ref{app:closed_vs_open_kernel}
makes dissipation explicit through the spectrum of $\mathcal L_0^\dag$ and yields
sums of decaying exponentials directly. There is a complementary (and historically
standard) linear-response viewpoint, often phrased in terms of an $\eta\downarrow 0^+$
prescription in Kubo theory, which clarifies how a strictly closed thermal channel
fails to produce decay and how even infinitesimal damping broadens the spectrum.
This is useful here because our goal is precisely to justify why an experimentally
observed control stack behaves like a classical causal filter $K(\tau)$ rather than
a quasi-periodic kernel without resorting to the open channel.

In the Kubo setup, the channel is assumed to be in an equilibrium (e.g. thermal)
state in the remote past, and the perturbation is switched on adiabatically so that the
state is not abruptly driven out of equilibrium at $t_0\to -\infty$. Concretely, one
replaces the drive by an adiabatically switched source, for example
$$
u(t)\ \mapsto\ u_\eta(t):=e^{\eta t}\,u(t), \qquad \eta>0,
$$
and only at the end takes $\eta\downarrow 0^+$. This implements (i) convergence of
integrals from $t_0=-\infty$, and (ii) selection of the retarded (causal) solution.
This is the same physical role played by the $+i0^+$ prescription in retarded Green
functions and susceptibilities. (See Kubo's original development of linear response.)
\cite{Kubo1957Irreversible1}
In this context, we start from the closed-channel retarded response (unitary interaction picture w.r.t. $H_B$),
$$
K(\tau)\equiv \chi_{LF}(\tau)= i\,\Theta(\tau)\,\langle [L_I(\tau),F_I(0)]\rangle_B.
$$
With the adiabatic switch-on, the same derivation yields an $\eta$-regularized retarded kernel
$$
K_\eta(\tau)= i\,\Theta(\tau)\,e^{-\eta\tau}\,\langle [L_I(\tau),F]\rangle_B.
$$
Thus, even before invoking any bath, each oscillatory component acquires a factor
$e^{-\eta\tau}$: the $\eta$ prescription is literally an (infinitesimal) exponential damping
in the time domain.

For a closed thermal channel, the Lehmann form produces discrete spectral lines at $\omega_{mn}$
frequencies. With $\eta>0$, denominators acquire a finite imaginary part:
$$
K_\eta(\omega)
=
\sum_{n,m} (p_n-p_m)\,L_{nm}F_{mn}\,
\frac{1}{\omega-\omega_{nm}+i\eta},
\qquad
\omega_{nm}:=E_n-E_m.
$$
Taking the imaginary part makes the broadening explicit:
$$
\Im\,K_\eta(\omega)
=
-\pi\sum_{n,m}(p_n-p_m)\,L_{nm}F_{mn}\,\delta_\eta(\omega-\omega_{nm}),
\qquad
\delta_\eta(x):=\frac{1}{\pi}\frac{\eta}{x^2+\eta^2}.
$$
So the $\delta$-peaks of the strictly closed channel are replaced by Lorentzians of width $\eta$.
In condensed-matter and transport practice, one often replaces $\eta$ by a physical relaxation
rateto model dissipation, i.e. $\eta\to\Gamma$, producing the same
kind of spectral broadening but with a non-infinitesimal width. 

A strictly closed finite-dimensional channel has purely real frequencies, so $K(\tau)$ is a
sum of oscillatory terms and does not generically decay. To obtain a classical-looking memory
kernel, the spectral lines must be broadened by genuine relaxation mechanisms:
leakage into uncontrolled degrees of freedom, resistive losses, thermalization in attenuators,
or (equivalently) an effective continuum/coarse-graining limit. Operationally, this means that
the retarded response is controlled by poles displaced into the lower half-plane,
$$
\omega\ \mapsto\ \omega+i\gamma_k, \qquad \gamma_k>0,
$$
so that each contribution behaves like a damped oscillation
$$
K_k(\tau)\propto \Theta(\tau)\,e^{-\gamma_k\tau}\,e^{i\omega_k\tau}.
$$
In the relaxational (overdamped) regime relevant for an RC-like control stack, the dominant poles
in the drive band are approximately on the imaginary axis,
$$
K(\omega)\approx \sum_{k}\frac{c_k}{\nu_k-i\omega},
\qquad \nu_k>0,
$$
and the inverse transform gives a sum of pure decays,
$$
K(\tau)\approx \Theta(\tau)\sum_k c_k e^{-\nu_k\tau}.
$$
This is exactly the same functional form obtained from the GKLS semigroup expansion
in Eq.~\eqref{eq:chi_sum_decays_rc_like} and in the RC-ladder case; the two viewpoints differ only in language. First, in the open-channel/GKLS picture, the decay rates $\nu_k$ are $\nu_k=-\Re(\lambda_k)$
from the spectrum of $\mathcal L_0^\dag$ (microscopic dissipation made explicit).
Second, in the $\eta$-regularized Kubo picture, $\eta$ is the adiabatic retarded prescription,
and replacing $\eta$ by a physical linewidth $\Gamma$ amounts to modeling the same
pole displacement caused by dissipation/broadening mechanisms in the channel.

Both routes justify reading the experimentally observed control line as a causal filter with a
stable pole structure; the RC-ladder approximation used in the main text corresponds to the
special case where the relevant poles are dominantly real and negative in the bandwidth of interest.

\end{document}